\documentclass[a4paper,12pt]{article}
\usepackage{indentfirst}
\usepackage{graphicx}
\usepackage{epstopdf}
\usepackage{epsfig}
\usepackage{array}
\usepackage{color}
\usepackage{multirow}
\usepackage{float}
\usepackage{subfigure}
\usepackage{amsmath,amssymb, amsthm}
\usepackage[mathscr]{euscript}
\usepackage{latexsym,bm}
\usepackage{cite}
\usepackage{booktabs}
\usepackage{bbm}
\usepackage{diagbox}
 \usepackage[all]{xy}
 \usepackage{color}
\usepackage{algorithm}  
\usepackage{algorithmicx} 
\usepackage{algpseudocode}
\usepackage{appendix}  

\newtheorem{theorem}{Theorem}
\newtheorem{lemma}{Lemma}

\newtheorem{example}{\textup{\textbf{Example}}}

\newcommand\D{\textup{d}}

\newcommand\comment[1]{}

\definecolor{purple}{RGB}{255,0,255}

{\theoremstyle{remark} \newtheorem{remark}{\textup{\textbf{Remark}}}}

\def\dif{\mathrm{d}}
\def\mi{\mathbbm{i}}
\def\me{\mathbbm{e}}
\def\mone{\mathbbm{1}}

\def\bx{\bm{x}}

\def\bk{\bm{k}}
\def\by{\bm{y}}

\makeatletter
\@addtoreset{equation}{section}
\makeatother

\begin{document}
\bibliographystyle{unsrt}

\title{The Wigner branching random walk: Efficient implementation and performance evaluation}
\author{Yunfeng Xiong
\and Sihong Shao\footnotemark[2] $^,$\footnotemark[1]}
\renewcommand{\thefootnote}{\fnsymbol{footnote}}
\footnotetext[2]{LMAM and School of Mathematical Sciences, Peking University, Beijing 100871, China.}
\footnotetext[1]{To
whom correspondence should be addressed. Email:
\texttt{sihong@math.pku.edu.cn}}
\date{\today}
\maketitle



%
%

\begin{abstract}

To implement the Wigner branching random walk, the particle carrying a signed weight, either $-1$ or $+1$, is more friendly to data storage and arithmetic manipulations than that taking a real-valued weight continuously from $-1$ to $+1$. The former is called a signed particle and the latter a weighted particle. In this paper, we propose two efficient strategies to realize the signed-particle implementation. One is to interpret the multiplicative functional as the probability to generate pairs of particles instead of the incremental weight, and the other is to utilize a bootstrap filter to adjust the skewness of particle weights. Performance evaluations on the Gaussian barrier scattering (2D) and a Helium-like system (4D) demonstrate the feasibility of both strategies and the variance reduction property of the second approach. We provide an improvement of the first signed-particle implementation that
partially alleviates the restriction on the time step and perform a thorough theoretical and numerical comparison among all the existing signed-particle implementations. Details on implementing the importance sampling according to the quasi-probability density and an efficient resampling or particle reduction are also provided.

\vspace*{4mm}

\end{abstract}

\section{Introduction}
\label{sec:intro}

The Wigner function $f(\bm{x}, \bm{k}, t)$ for a $N$-body $d$-dimensional quantum system lives 
in the phase space $(\bm{x},\bm{k})\in\mathbb{R}^{2n}$ with $n=Nd$ for position $\bm{x}$ and wavevector $\bm{k}$, and satisfies 
the following Wigner equation (WEQ) \cite{Wigner1932}
\begin{equation}
\frac{\partial }{\partial t}f(\bm{x}, \bm{k}, t)+\frac{\hbar \bm{k}}{\bm{m}} \cdot \nabla_{\bm{x}} f(\bm{x},\bm{k}, t)= \Theta_V[f](\bx, \bk, t)\label{eq.Wigner},
\end{equation}
where the pseudo-differential operator is characterized by a convolution with the Wigner kernel $V_W$
\begin{align}
\Theta_V[f](\bx, \bk, t) &= \int_{\mathbb{R}^{n}} \textup{d} \bm{\bm{k}^{\prime}} V_{W}(\bm{x},\bm{k}-\bm{k}^{\prime},t)f(\bm{x},\bm{k}^{\prime},t),\\
V_{W}(\bm{x},\bm{k},t) &=\frac{1}{\mi\hbar (2\pi)^{n}}\int_{\mathbb{R}^{n}} \text{d}\bm{y} \me^{-\mi\bm{k}\cdot \bm{y}} \left[ V(\bm{x}+\frac{\bm{y}}{2}, t)-V(\bm{x}-\frac{\bm{y}}{2}, t) \right], \label{Wigner_kernel}
\end{align}
provided that $V(\bx, t)$ belongs to an appropriate symbol class, $\hbar$ is the reduced Planck constant, $\bk/\bm{m}$ is short for $(\bk_{1}/{m}_1, \ldots, \bk_{N}/ {m}_{N})$ and $m_i$ is the mass of the $i$-th body. In the past few decades, WEQ has been drawing a growing attention, especially in the simulations of nanodevices \cite{NedjalkovKosinaSelberherrRinghoferFerry2004, KosinaNedjalkovSelberherr2000, KosinaNedjalkovSelberherr2003jce,ZhanColomesOriols2016,WoloszynSpisak2017,bk:QuerliozDollfus2010} as well as the many-body quantum mechanics \cite{SellierNedjalkovDimov2015, ShaoXiong2016}, due to its theoretical advantage \cite{bk:Balescu1975, tatarskiui1983}. 

The huge challenge to the numerical resolution of WEQ lies in the high dimensionality of the phase space, which are unfriendly to traditional deterministic solvers. A class of stochastic algorithms has recently been proposed and its mathematical theory consists of three components:  the probabilistic interpretation of WEQ, the principle of importance sampling and the technique of density estimation \cite{ShaoXiong2016, Wagner2016}. The corresponding implementations borrow several fundamental concepts from the direct simulation Monte Carlo method for the Boltzmann equation, including particle weights, random jumps and particle cancelations \cite{SellierNedjalkovDimov2014, RjasanowWagner1996, MuscatoWagnerStefano2010,YanCaflisch2015}. In this manner, it paves a promising way to overcome the curse of dimensionality, and several benchmarks have also revealed its reliability in capturing fine structures of 2D and 4D quantum systems \cite{ShaoXiong2016, ShaoSellier2015, MuscatoWagner2016}. 

Although almost all Monte Carlo approaches rely on the equivalent stochastic interpretation of WEQ, there exist significant differences between various realizations and thereby resulting in distinct performances. Thus it screams for fair and detailed evaluations to demonstrate how {numerical} accuracy is influenced by several elements, including the way to truncate the WEQ, the choice of the auxiliary function $\gamma(\bx)$ that characterizes the life-length of particles and the interpretation of the multiplicative functionals. The first signed particle Wigner Monte Carlo method (abbreviated as \textbf{sp0}) is suggested to choose a variable auxiliary function $\gamma(\bx)$ and confine the particle weights to either $+1$ or $-1$  (termed the signed particle) \cite{SellierNedjalkovDimov2014, KosinaNedjalkovSelberherr2004}. Such setting greatly facilitates the data storage and the arithmetic operation, but poses a limitation on the time step in order to maintain the accuracy.  In this work, we will propose an improvement (\textbf{I}) to alleviate such limitation via a proper treatment of the bias term and denote the resulting scheme by $\textbf{sp0-I}$.  Another class of stochastic algorithms based on the random cloud model (abbreviated as $\textbf{RC}$) has been proposed in \cite{Wagner2016, MuscatoWagner2016}, where a constant time technique and a rejection sampling technique are adopted for generating the scattering time and state, respectively \cite{MuscatoWagnerStefano2010}. 
In our previous work, we have proposed the Wigner branching random walk (WBRW) model, in which the multiplicative functionals are interpreted as the importance weights, yielding the weighted-particle branching random walk algorithm (abbreviated as \textbf{wp}) \cite{ShaoXiong2016}. Simulating the WBRW ameliorates the restriction on the time step and allows a reduction in variance by choosing a large constant $\gamma(\bx)$. The price to pay is that particle weights have to take values continuously in $[-1, 1]$ (termed the weighted particle). As a special case, \textbf{sp0} is fully recovered from $\textbf{wp}$ by choosing an auxiliary function according to the Wigner kernel $V_W$. But a constant auxiliary function $\gamma(\bx) \equiv \gamma_0$ is much more preferable when simulating {actual} many-body quantum  systems.

This paper intends to discuss another two efficient strategies to implement WBRW and make a thorough comparison with the existing signed-particle implementations. The performance of  stochastic methods is usually related to the variance reduction, data storage, as well as an appropriate choice of sample size $N_\alpha$. Therefore, our goals are to get both advantages of $\textbf{sp0}$ and $\textbf{wp}$, say, the setting of signed weights and the manner to improve the accuracy systematically. Besides, we will discuss the criterion to choose an appropriate particle size, which may be related to both the accuracy of the Monte Carlo simulations and the efficiency of the resampling techniques. First, the signed-particle branching random walk algorithm (abbreviated as $\textbf{sp1}$) is proposed, where the multiplicative functional is interpreted as the {rejecting ratio} and thus the particle weights are confined to be either 1 or -1. We put all the signed-particle implementations in a unified framework and analyze the sources of errors. However, numerical experiments demonstrate that the variance of $\textbf{sp1}$ is usually larger than that of $\textbf{wp}$ and cannot be diminished by increasing $\gamma(\bx)$. To save it, we introduce a bootstrap filter in \textbf{wp} to adjust particle weights to $\pm 1$ and preserve the variance reduction property, yielding another signed-particle implementation (abbreviated as $\textbf{sp2}$). A theoretical error bound on the bootstrap filtering is analyzed and ensures the convergence as $N_{\alpha} \to \infty$. Detailed performance evaluations demonstrate the convergence, accuracy and efficiency of the proposed strategies, with the reference solutions produced by several highly accurate deterministic solvers \cite{ShaoLuCai2011, XiongChenShao2016}. For resampling, it is observed that the sample size should be comparable to the partition size if a uniform histogram is adopted, otherwise both efficiency and accuracy are undermined. In fact, similar phenomena have also been reported in statistical community, referring to the current challenge to the multivariate density estimation (particle cancelation) \cite{bk:LaszloGyorfi2002, bk:HastieTibshiraniFriedman2009}. 

The rest is organized as follows. In Section \ref{sec:wbrw}, we briefly review the mathematical theory of the WBRW. 
Two strategies of signed-particle implementation, as well as a theoretical comparison with the existing approaches, are given in Section \ref{sec:sp}.
Several important issues on the implementation are illustrated in Section \ref{sec:imp}. For the sake of readers' convenience, we try our best to illustrate all the details and unfold the `blackbox' of coding WBRW.  A sequence of performance evaluations is reported in Section \ref{sec:num}, {associated with} a thorough comparison among $\textbf{sp1}$, $\textbf{sp2}$, $\textbf{wp}$, $\textbf{sp0}$, $\textbf{sp0-I}$, and $\textbf{RC}$. The paper is concluded in Section \ref{sec:con} with a few remarks.


\section{The Wigner branching random walk}
\label{sec:wbrw}

The mathematical framework for exploring the inherent relation between the WBRW and the WEQ has been established recently from the viewpoint of computational mathematics \cite{ShaoXiong2016}. For the sake of completeness as well as readers' convenience, we will give a brief overview of the main findings there in this section and assume that the potential in Eq.~\eqref{Wigner_kernel} is time-independent for brevity hereafter, 
and the generalization to the time-dependent scenario is straightforward (see Remark \ref{remark_1}).

To connect rigorously the WBRW to Eq.~\eqref{eq.Wigner}, 
an auxiliary function $\gamma(\bx)$ is added on both sides of Eq.~\eqref{eq.Wigner} in an equivalent manner for producing the exponential distribution,
and then the resulting equation can be cast into 
a renewal-type integral equation
\begin{equation}\label{Def:renewal_type_eq}
\begin{split}
f(\bm{x}, \bm{k}, t)=&\left[1-\mathcal{H}(0; \bm{x}, t)\right] f(\bm{x}(t), \bm{k}, 0) \\
&+ \int_{0}^{t} \D \mathcal{H}(t^{\prime}; \bm{x}, t) \int_{\mathbb{R}^n} \D \bm{x}^\prime 
 \int_{\mathcal{K}} \D \bm{k}^{\prime}
 \frac{\Gamma(\bm{x}(t-t^{\prime}), \bm{k}; \bm{x}^{\prime}, \bm{k}^{\prime})}{\gamma(\bm{x}(t-t^{\prime}))}
 f(\bm{x}^\prime, \bm{k}^{\prime}, t^{\prime}) ,
\end{split}
\end{equation}
where 
\begin{equation}\label{eq:H}
\mathcal{H}(t^{\prime} ; \bm{x}, t )=\int_{t^{\prime}}^{t}  \gamma(\bm{x}(t-\tau)) \me^{-\int_{\tau}^{t} \gamma(\bm{x}(t-s)) \D s}~  \D \tau
\end{equation}
is a probability measure with respect to $t^{\prime}$ for a given $(\bm{x}, t)$ on $t^{\prime} \leq t$
provided that the auxiliary function $\gamma(\bm{x})$ satisfies
\begin{equation}\label{eq:gamma_condition}
\gamma(\bm{x})\geq 0, \quad \lim_{t^\prime\to-\infty}\int_{t^\prime}^t \gamma(\bm{x}(t-s))\dif s = +\infty, \quad \forall\, \bm{x}\in\mathbb{R}^n,
\end{equation}
$\bm{x}(\Delta t)=\bm{x}-{\hbar \bm{k}\Delta t}/{\bm{m}}$ denotes the backward-in-time trajectory of $(\bm{x}, \bm{k})$ with a positive time increment $\Delta t$. It's readily seen from Eq.~\eqref{Def:renewal_type_eq} that $\gamma(\bx)$ serves as the intensity of an exponential distribution. The kernel $\Gamma(\bm{x}, \bm{k}; \bm{x}^{\prime}, \bm{k}^{\prime})$ is given by 
\begin{equation}\label{eq.Gamma}
\Gamma(\bm{x}, \bm{k}; \bm{x}^{\prime}, \bm{k}^{\prime}) =\left[V^{+}_{W}(\bm{x},\bm{k}-\bm{k}^{\prime})-V^{-}_{W}(\bm{x},\bm{k}-\bm{k}^{\prime})+\gamma(\bm{x})\delta(\bm{k}-\bm{k}^{\prime})\right] \delta(\bm{x}-\bm{x}^{\prime}),
\end{equation}
and we have adopted  the $k$-truncated Wigner kernel for the purpose of numerical computation,
namely, the $\bm{k}$-space is truncated into a finite domain $\mathcal{K}$ and a simple nullification adopted outside $\mathcal{K}$ by exploiting the decay of the Wigner function when $|\bm{k}|\to+\infty$ due to the Riemann-Lebesgue lemma.

Splitting the Wigner kernel into positive and negative parts in \eqref{eq.Gamma} is of great importance to a probabilistic interpretation.  In general, the Wigner kernel $V_W$ is composed of $M$ parts
$V_{W} = V_{W, 1}+V_{W, 2}+\cdots V_{W, M}$, that corresponds to the potential $V=V_1+V_2+\cdots+V_M$, then the Wigner kernel can be split into $M$ pairs as follows 
\begin{equation}\label{eq:M}
V_{W} = V_{W}^{+} - V_{W}^{-},\quad V_{W}^{\pm} = \sum_{m=1}^M V_{W, m}^{\pm}, \quad
V_{W, m}^{\pm} =\frac{1}{2}\left|V_{W, m}\right| \pm \frac{1}{2}V_{W, m}.
\end{equation}
Such setting is also very helpful in dealing with complicated systems composed of two-body interactions, as combining it with the Fourier completeness relation helps to decompose the convolution term  in Eq.~\eqref{eq.Wigner} into several lower dimensional integrals, for instance, see Example \ref{example1}.

\begin{example}\rm\label{example1}
We consider a 4D Helium-like system, in which the potential is composed of electron-nucleus attractive Yukawa interactions and electron-electron repulsive Yukawa interaction, 
\begin{equation}\label{eq:M3}
\begin{split}
V(x_1, x_2) &= V_{\textup{ne}}(x_1) +  V_{\textup{ne}}(x_2)  + V_{\textup{ee}}(x_1, x_2) \\
&= -\frac{2\me^{-\kappa|x_1-x_A|}}{2\kappa}-\frac{2\me^{-\kappa|x_2-x_A|}}{2\kappa}+\frac{\me^{-\kappa|x_1-x_2|}}{2\kappa}. 
\end{split}
\end{equation}
The parameter $\kappa$ expresses the screening strength, 
$x_A$ denotes the position of the nucleus, and $x_i (i=1,2)$ is the position of the $i$-th electron. The corresponding Wigner kernel reads
\begin{align}\label{wigner_kernel_yukawa}
V_{W, \textup{ne}}(x_i,  k_i) &= -\frac{2}{\hbar \pi} \cdot \frac{\sin(2k_i(x_i-x_A))}{4k_i^2+\kappa^2} , \\
V_{W, \textup{ee}}(x_1, x_2, k_1, k_2) &= \frac{1}{\hbar \pi}  \cdot \frac{\sin(2k_1 x_1+2k_2 x_2)}{|k_1-k_2|^2+\kappa^2} \cdot \delta(k_1 + k_2).
\end{align}
As a result, the convolution term is decomposed into three 1D integrals: 
\begin{align}
& \int_{\mathbb{R}^2} V_{W}(\bx, \bk - \bk^{\prime}) f(\bx, \bk^{\prime}, t) \D \bk^{\prime} = I_1(\bx, \bk, t) + I_2(\bx, \bk, t)+I_3(\bx, \bk, t),\\
& I_1(\bx, \bk, t) = -\frac{2}{\hbar \pi} \int_{\mathbb{R}} \frac{\sin(2k^{\prime} \cdot (x_1-x_A))}{4(k^{\prime})^2+\kappa^2} f(x_1, x_2, k_1 - k^{\prime}, k_2, t) \D k^{\prime}, \\
& I_2(\bx, \bk, t) = -\frac{2}{\hbar \pi} \int_{\mathbb{R}} \frac{\sin(2k^{\prime} \cdot  (x_2-x_A))}{4(k^{\prime})^2+\kappa^2} f(x_1, x_2, k_1 , k_2 - k^{\prime}, t) \D k^{\prime}, \\
& I_3(\bx, \bk, t) = \frac{1}{\hbar \pi} \int_{\mathbb{R}} \frac{\sin(2k^{\prime} \cdot (x_1-x_2))}{4(k^{\prime})^2+\kappa^2} f(x_1, x_2, k_1 - k^{\prime} , k_2 + k^{\prime}, t) \D k^{\prime}.
\end{align}
\end{example}

The main use of the Wigner function is to calculate observes $\langle \hat{A} \rangle_{T}$ at a given final time $T$, such as the average position of particles, electron density, etc., all of which are attributed to the inner product problem
\begin{align}
&\langle g_0, f \rangle= \int_{0}^{T}  \D t  \int_{\mathbb{R}^n}  \D \bm{x}  \int_{\mathcal{K}} \D \bm{k}~g_0(\bm{x}, \bm{k}, t) f(\bm{x}, \bm{k} ,t).\label{Def:inner_product}\\
&\langle \hat{A} \rangle_T = \langle g_0, f \rangle \,\, \text{ with } \,\,
g_0(\bm{x}, \bm{k},t)=A(\bm{x}, \bm{k})\delta(t-T).\label{eq:AT_fT} 
\end{align}
Here $\hat{A}(\hat{\bm{x}}, \hat{\bm{k}})$ is an arbitrary quantum operator, and $A(\bm{x}, \bm{k})$ the corresponding Weyl symbol. Actually, the WBRW algorithms are devoted into estimating $\langle \hat{A} \rangle_T$  
and bottomed on the dual theory of WEQ.  
It has been proved that the average value $\langle \hat{A} \rangle_T$ 
can be determined {\sl only by the `initial' data} as follows 
\begin{equation}\label{eq:important}
\langle \hat{A} \rangle_{T}=  \int_{\mathbb{R}^n} \D \bm{x} \int_{\mathcal{K}}\D \bm{k} ~f(\bm{x}, \bm{k}, 0) \varphi(\bm{x}, \bm{k}, 0),
\end{equation}
where the dual variable $\varphi(\bm{x}, \bm{k}, t)$ satisfies
\begin{equation}\label{Adjoint_renewal_type_equation}
\begin{split}
\varphi(\bm{x}, \bm{k}, t)=&\left[1-\mathcal{G}(T; \bm{x}, t)\right] A(\bm{x}(T-t), \bm{k})\\
&+\int_{t}^{T} \D \mathcal{G}(t^{\prime}; \bm{x}, t)  \int_{\mathbb{R}^n} \D \bm{x}^{\prime} \int_{\mathcal{K}} \D \bm{k}^{\prime}  \frac{\Gamma(\bm{x}^{\prime}, \bm{k}^{\prime}; \bm{x}(t^{\prime}-t), \bm{k})}{\gamma(\bm{x}(t^{\prime}-t))} \varphi(\bm{x}^{\prime}, \bm{k}^{\prime}, t^{\prime}).
\end{split}
\end{equation}
Eq. \eqref{Adjoint_renewal_type_equation} can be regarded as the adjoint equation  of Eq. \eqref{Def:renewal_type_eq} ({also called the backward equation}). Noting that it is required $t^\prime\geq t$ for convenience in the former, but $t^\prime\leq t$ is always assumed in the latter. Here, 
\begin{equation}\label{eq:measure2}
\mathcal{G}(t^{\prime}; \bm{x}, t)=\int_{t}^{t^{\prime}} \gamma(\bm{x}(\tau-t)) \me^{-\int^{\tau}_{t} \gamma(\bm{x}(s-t)) \D s} ~ \D \tau
\end{equation}
is again a probability measure with respect to $t^{\prime}$ for given $(\bm{x}, t)$ on $t^{\prime} \geq t$
under the assumption that
the auxiliary function satisfies
\begin{equation}\label{eq:gamma_conditioN^+}
\gamma(\bm{x})\geq 0, \quad \lim_{t^\prime\to+\infty}\int_{t}^{t^\prime} \gamma(\bm{x}(t-s))\dif s = +\infty, \quad \forall\, \bm{x}\in\mathbb{R}^n, 
\end{equation}
and $\bm{x}(\Delta t)=\bm{x}+{\hbar \bm{k}\Delta t}/{\bm{m}}$ denotes the forward-in-time trajectory of $(\bm{x}, \bm{k})$ with a positive time increment $\Delta t$. 

The reasons why we mainly focus on the adjoint equation, instead of the original WEQ,  are as two-fold. First, Eq.~\eqref{eq:important} allows the use of the importance sampling approach, with $f_I \propto |f(\bx, \bk, 0)| $ chosen as the instrumental distribution. Second, estimating $\varphi(\bx, \bk, 0)$ can be implemented in a time-marching manner. By contrast, the backward-in-time stochastic algorithms are pointwise in nature and only attractive when one is interested in estimating $f(\bx, \bk, t)$  {at a few given points} \cite{bk:Liu2001}.

A branching random walk model has been introduced in \cite{ShaoXiong2016} with its expectation consistent with the mild solution of the adjoint equation \eqref{Adjoint_renewal_type_equation}. In the branching particle system, particles carrying importance weights are indexed by a branching random tree, and their motions are described by deterministic travels and random jumps (without any diffusion). In order to obtain a time series of quantum observables $\langle \hat{A} \rangle_{t}$, we adopt the setting of an equidistant partition of the time interval $[0, t_{fin}]$ 
\begin{equation}\label{def.time_partition}
0 = t_0 \le t_1 \le t_2 \le \cdots \le t_{n-1} \le t_n = t_{fin}, ~~t_{l+1} - t_l = \Delta t,
\end{equation}
and summarize the rules of the branching particle system for a certain interval $[t_l, t_{l+1}]$ in Alg.~\ref{weighted_particle_WBRW}. The implementations of WBRW are resorted to the sequential Monte Carlo techniques \cite{bk:DoucetDeFreitasGordon2001}.
  
\begin{algorithm}
\caption{The Wigner branching random walk}\label{weighted_particle_WBRW}

\vspace{3mm} 

Suppose each particle in the branching particle system, carrying an initial weight either $1$ or $-1$, starts at state $(\bx_\alpha, \bk_\alpha)$ at time $t_l$ and moves until $t_{l+1} = t_l +\Delta t$ according to the following rules.  
\begin{itemize}

\item[1.]  For $\textbf{sp0}, \textbf{sp1}, \textbf{sp2}$ and $\textbf{wp}$, the particle at $(\bm{x}, \bm{k}, t)$ dies in the age time interval $(t, t^{\prime}) \subset [t_l, t_{l+1}]$  with probability  $1- \int_{t}^{t+\tau} \me^{-\gamma(\bx(s-t))}  \D s $, with a random life-length $\tau = t^{\prime}-t$. For $\textbf{RC}$, $\tau$ is fixed to be $\tau = \Delta t$.

\item[2.] If $t+\tau > t_{l+1}$, say, the life-length of the particle exceeds $t_{l+1}-t$, the particle immigrates to the state $(\bm{x}(t_{l+1}-t), \bk)$ and becomes frozen. 

\item[3.] If $t+\tau \le t_{l+1}$, the particle carrying the weight $w$ dies at age $t^{\prime}=t+\tau$ at state $(\bm{x}(\tau) , \bm{k})$ and produces at most $2M+1$ offsprings at states $(\bm{x}^{\prime}_{(1)}, \bm{k}^{\prime}_{(1)})$,  $\cdots$, $(\bm{x}^{\prime}_{(2M+1)}, \bm{k}^{\prime}_{(2M+1)})$, endowed with updated weights $w^{\prime}_{(1)}$, $\cdots$, $w^{\prime}_{(2M+1)}$.
\begin{align}
\bm{x}^{\prime}_{(1)} = \bm{x}^{\prime}_{(2)} = \cdots = \bm{x}^{\prime}_{(2M+1)} = \bm{x}(\tau),\\
\bk - \bm{k}^{\prime}_{(2m-1)}\propto \frac{V^{-}_{W,m}(\bm{x}(\tau), \bm{k})}{\xi_m(\bm{x}(\tau)) },~
\bm{k}-\bm{k}^{\prime}_{(2m)} \propto \frac{V^{+}_{W, m}(\bm{x}(\tau), \bm{k})}{\xi_m(\bm{x}(\tau)) }, \label{def.k_sample}\\
\bm{k}^{\prime}_{(2M+1)}=\bm{k},\\
\xi_{m}(\bm{x})=\int_{2\mathcal{K}} V^{+}_{W, m}(\bm{x}, \bm{k}) \D \bm{k} =  \int_{2\mathcal{K}} V^{-}_{W, m}(\bm{x}, \bm{k}) \D  \bm{k}. \label{def.xi_m}
\end{align}

\setlength{\fboxsep}{0.2cm} 
\leftline{
\fbox{ \shortstack[l] {
For $\textbf{sp1}$,  the $i$-th offspring is generated with probability $\Pr(i)$. \\
$\Pr{(2m-1)} = \frac{\xi_{m}(\bm{x}(\tau))}{\gamma(\bx(\tau))} , ~~ \Pr{(2m)} = \frac{\xi_{m}(\bm{x}(\tau))}{\gamma(\bx(\tau))},~~ \Pr(2M+1) = 1.$\\
$w^{\prime}_{(2m-1)} = w \cdot \mone_{ \{ \bm{k}^\prime_{2m-1} \in \mathcal{K} \} }, ~~ w^{\prime}_{(2m)} = - w \cdot  \mone_{ \{ \bm{k}^\prime_{2m} \in \mathcal{K} \} }, ~~ w^{\prime}_{(2M+1)} = w$. \quad \quad \quad \quad ~~~}
}
}
\vspace{0.2cm}

\leftline{
\fbox{ \shortstack[l] {
For $\textbf{wp}$ and $\textbf{sp2}$,  the $i$-th offspring is generated with probability $1$. \\
$w^{\prime}_{(2m-1)} = w \cdot \frac{\xi_{m}(\bm{x}(\tau))}{\gamma(\bx(\tau))} \mone_{ \{ \bm{k}^\prime_{2m-1} \in \mathcal{K} \} }, ~ w^{\prime}_{(2m)} = -w \cdot \frac{\xi_{m}(\bm{x}(\tau))}{\gamma(\bx(\tau))} \mone_{ \{ \bm{k}^\prime_{2m} \in \mathcal{K} \} }, ~ w^{\prime}_{(2M+1)}= w$. 
}}
}

\vspace{0.2cm}
\leftline{
\fbox{ \shortstack[l] {
For $\textbf{sp0}$, $M = 1$ and the $i$-th offspring is generated with probability $1$. \quad \quad ~~~\\
$ w^{\prime}_{(1)}= w \cdot \mone_{ \{ \bm{k}^\prime_{1} \in \mathcal{K} \} }, ~~ w^{\prime}_{(2)}= -w \cdot \mone_{ \{ \bm{k}^\prime_{2} \in \mathcal{K} \} }, ~~ w^{\prime}_{(3)}= w$.
}}
}

\vspace{0.2cm}
\leftline{
\fbox{ \shortstack[l] {
For $\textbf{RC}$, $M = 1$  and the $i$-th offspring is generated with probability $\Pr(i)$. \quad \quad\\
$\Pr(1) = \gamma(\bx(\Delta t)) \Delta t, ~~ \Pr(2) = \gamma(\bx(\Delta t)) \Delta t,  ~~\Pr(3) = 1.$\\
$ w^{\prime}_{(1)}= w \cdot \mone_{ \{ \bm{k}^\prime_{1} \in \mathcal{K} \} }, ~~ w^{\prime}_{(2)}= -w \cdot \mone_{ \{ \bm{k}^\prime_{2} \in \mathcal{K} \} } , ~~ w^{\prime}_{(3)}= w$.
}}
}

\item[4.] Frozen particles are denoted by the collection $\mathcal{S} = \{ (\bx_i, \bk_i)\}_{i=1}^N$ and are weighted by $ \mathcal{W} = \{ w_i\}_{i=1}^N$, with $w_i  \in [-1, 1]$. Any quantum observables $\langle \hat{A} \rangle_t$  can be estimated through Eq.~\eqref{eq.estimator}. 

\item[5.] The $(2m-1)$-th and $(2m)$-th particles are suggested to be produced in pair to maintain the mass conservation, say, $\sum_{i=1}^N w_i = N_\alpha$. 

\end{itemize}

\end{algorithm}

Our ultimate goal is to derive the following estimator of the data point $\langle \hat{A} \rangle_{t_{l+1}}$, where $t_{l}$ and $t_{l+1}$ denote the initial and final instant, respectively.
\begin{equation}\label{eq.estimator}
\langle \hat{A} \rangle_{t_{l+1}} \approx  \frac{1}{N_\alpha}\sum_{\alpha} \sum_{j \in \mathcal{E}_{\alpha}} \varphi(\bx_{j, \alpha}, \bk_{j, \alpha}, t_l) \cdot w_{j, \alpha}(t_l) = \langle A, \frac{1}{N_\alpha}\sum_{i = 1}^{N} w_{i} \delta_{\bx_i, \bk_i} \rangle.
\end{equation}
Here $N_\alpha$ is the sample size, the indices $\alpha \in \{ 1, \cdots, N_\alpha \}$ mark the draws (particles) from the instrumental probability density $f_I$ and the index set $\mathcal{E}_{\alpha}$ marks the offsprings produced by $\alpha$-th particle, and the equality is only a rearrangement of the indices. For more details, one can refer to \cite{ShaoXiong2016}.

\begin{remark}\label{remark_1}
Alg.~\ref{weighted_particle_WBRW} can be readily generalized to the problem with a time-varying potential $V(\bx, t) = \sum_{m=1}^M V_m(x, t)$, as we can simply replace Eq.~\eqref{def.k_sample} by
\begin{equation}
\bk - \bm{k}^{\prime}_{(2m-1)}\propto \frac{V^{-}_{W,m}(\bm{x}(\tau), \bm{k}, t+\tau)}{\xi_m(\bm{x}(\tau), t+\tau) },\quad
\bm{k}-\bm{k}^{\prime}_{(2m)} \propto \frac{V^{+}_{W, m}(\bm{x}(\tau), \bm{k}, t+\tau)}{\xi_m(\bm{x}(\tau), t+\tau) }, \label{def.k_sample}
\end{equation}
and Eq.~\eqref{def.xi_m} by
\begin{equation}
\xi_{m}(\bm{x}, t)=\int_{2\mathcal{K}} V^{+}_{W, m}(\bm{x}, \bm{k}, t) \D \bm{k} =  \int_{2\mathcal{K}} V^{-}_{W, m}(\bm{x}, \bm{k}, t) \D  \bm{k}.
\end{equation}
Here the auxiliary function $\gamma(\bx)$ are chosen to satisfy $\displaystyle{\gamma(\bx) \ge \max_{m} \max_{t\in [t_l, t_{l+1}]} \xi_m(\bx, t)}$.
\end{remark}

\section{Signed-particle WBRW: different strategies}
\label{sec:sp} 

Several ways have been proposed to realize the branching random walk algorithms in Alg.~\ref{weighted_particle_WBRW}. The major differences, to the best of our knowledge, are attributed to two issues.  

One is the probabilistic interpretation of the multiplicative functional $\xi_m(\bx)/\gamma(\bx)$ 
or $\hat{\xi}_m(\bx)/\gamma(\bx)$ (both denoted by $\xi/\gamma$ for brevity) in Eq.~\eqref{eq.majorant}.
\begin{equation}\label{eq.majorant}
\begin{split}
&\sum_{m=1}^{M}\frac{\xi_m(\bx)}{\gamma(\bx)} \int_{\mathcal{K}} \left[\frac{V^+_{W,m}(\bx, \bk - \bk^{\prime})}{\xi_m(\bx)} - \frac{V^-_{W,m}(\bx, \bk - \bk^{\prime})}{\xi_m(\bx)}\right] \varphi(\bx, \bk^{\prime}, t) \D \bk^\prime \\
&=  \sum_{m=1}^{M} \frac{\hat{\xi}_m(\bx)}{\gamma(\bx)} \int_{\mathcal{K}} \frac{V_{W,m}(\bx, \bk-\bk^{\prime})}{\left|V_{W,m}(\bx, \bk-\bk^{\prime})\right|} \cdot \frac{\left| V_{W,m}(\bx, \bk-\bk^{\prime}) \right|}{\hat{V}_{W,m}(\bx, \bk-\bk^{\prime})} \cdot \frac{\hat{V}_{W,m}(\bx, \bk-\bk^{\prime})}{\hat{\xi}_m(\bx)} \varphi(\bx,\bk^{\prime}, t) \D \bk^\prime,
 \end{split}
 \end{equation}
where $\hat{V}_W(\bx, \bk)$ is a positive semidefinite majorant function such that 
\begin{equation}
V^\pm_{W,m}(\bx, \bk) \le \hat{V}_{W,m}(\bx, \bk), \quad (\bx, \bk) \in \mathbb{R}^n \times \mathbb{R}^n,
\end{equation}
and $\hat{\xi}_m(\bx)$ denotes the corresponding normalizing function of $\hat{V}_{W,m}$. 

The other is the generation of the life-length $\tau$ for estimating the following integral (the second term in Eq.~\eqref{Adjoint_renewal_type_equation})
\begin{equation}\label{def.time_integral}
\int_{t_l}^{t_{l+1}}\Theta_V[\varphi](\bx(t^{\prime}-t_l), \bk, t^{\prime}) \me^{-\int^{t^{\prime}}_{t_l} \gamma(\bm{x}(s-t_l)) \D s} \D t^{\prime}. 
\end{equation}

Sections \ref{sec:gamma} and \ref{sec:tau_generation} illustrate the underlying idea of the existing approach. In fact, the basic tools are the principles of both rejection sampling and importance sampling, and the explicit Euler evolution. This motivates us to propose two other implementations: $\textbf{sp1}$ and $\textbf{sp0-I}$. The former is the signed-particle counterpart of $\textbf{wp}$ and the latter is a direct improvement of $\textbf{sp0}$. A comparison among all these strategies will be given in Section \ref{sec:comparison}.

Section \ref{sec:bootstrap} turns to seek a new approach $\textbf{sp2}$, which is essentially a weighted-particle implementation but associated with a bootstrap filtering to adjust the continuous weights into signed weights. The standard bootstrap filtering, such as Alg.~\ref{res_bootstrapping}, provides a useful way to adjust the skewness of particle weights, but usually requires that weight functions to be positive semidefinite. In order to deal with particles carrying weights in $[-1, 1]$, an extended bootstrapping is provided.



\subsection{Probabilistic interpretation of $\xi/\gamma$}
\label{sec:gamma}

No matter what strategies are used, the starting point turns out to be Eq.~\eqref{eq.majorant}, where each term can be endowed with a probabilistic interpretation. With the principle of the important sampling, we may find some majorants $\hat{V}_{W, m}(\bx, \bk)$ and draw samples from it. Then two factors ${|V_{W, m}}(\bx, \bk)|/{\hat{V}_{W,m}(\bx, \bk)}$ and $\xi_m(\bx)/\gamma(\bx)$ are treated as either the importance weights or the rejecting ratios. The former yields the weighted-particle implementation and the latter yields the signed-particle counterpart. Besides, the choices of  $\hat{V}_{W, m}$ and $\xi_m(\bx)$ also play a fundamental role in constructing the trajectories. A natural choice is $\hat{V}_{W,m} =V_{W,m}^+$ as adopted in \textbf{sp0} when $M=1$ and thus also employed in our strategies.


In our previous work \cite{ShaoXiong2016}, $\textbf{wp}$ suggests to choose a $\gamma(\bx) \ge \max_m \xi_m(\bx)$ and treats the $\xi_m(\bx)/\gamma(\bx)$ as the incremental particle weights. And the expectation of growth rate of particle number is bounded by $\me^{2M\gamma_0(T-t)}$.  In particular, when
\begin{equation}\label{def_gamma}
\gamma(\bx)=\xi(\bx)= \int_{2\mathcal{K}} V^{+}_{W}(\bx, \bm{k}) \D \bm{k},
\end{equation}
 the multiplicative functionals $\xi/\gamma$ in Alg.~\ref{weighted_particle_WBRW} are chosen to be either ${\xi(\bx(\tau))}/{\gamma(\bx(\tau))} = 1$ or $-{\xi(\bx(\tau))}/{\gamma(\bx(\tau))} = -1$, resulting in the signed weights in $\textbf{sp0}$ \cite{SellierNedjalkovDimov2014, NedjalkovSchwahaSelberherr2013}. 
 
By contrast, we can also interpret $\xi_m(\bx)/\gamma(\bx)$ as the probability of generating the $(2m-1)$-th and $(2m)$-th particles as an alternative. In particular, when a constant $\gamma(\bx)\equiv\gamma_0$ is used, such approach is our first signed-particle implementation of the Wigner branching random walk: $\textbf{sp1}$. One can easily give an estimation of the growth rate of total particle number for \textbf{sp1}. It is found that the growth rate of particle number, as well as the computational complexity, is suppressed.

\begin{theorem}\label{th:exp}
Denote $\mathbb{E}Z_T$ by the expectation of the total number of frozen particles in time interval $[0, T]$. For \textbf{sp1}, we have
\begin{equation}\label{signed_particle_number_formula}
\mathbb{E}Z_T \le \me^{2M\check{\xi} T},
\end{equation}
where $\check{\xi} := \max_{m} \sup_{\bx} \xi_m(\bx)$. Thus the bound for the average particle number is independent of the choice of $\gamma_0$.
\end{theorem}

\begin{proof}
Each particle is produced with probability no more than $\check{\xi}/\gamma_0$ except the last particle, which is produced with probability 1. Thus the expectation $\mathbb{E}Z_{t}$ satisfies the following inequality
\begin{equation}
\mathbb{E}Z_{t} \le 1-G(t) + (\frac{2M \check{\xi}}{\gamma_0}+1) \int_0^t \mathbb{E}Z_{t-u}  \D G(u),
\end{equation}
with $G(u) = 1 - \me^{-\gamma_0 u}$.  By the Gronwall's inequality, we have
\begin{equation}
\mathbb{E}Z_t \le \me^{2M \check{\xi} t}. 
\end{equation}
By replacing $t$ with $T$, we have completed the proof.
\end{proof}

The calculation of $\xi_m(\bx)$ has to resort to numerical integrations of oscillatory integrals, which might not be an easy task.  In $\textbf{RC}$ \cite{MuscatoWagner2016}, the authors suggest to choose a majorant function, the normalizing function $\xi(\bx)$ of which can be easily obtained, instead of $V^+_W$. Besides, they suggest to choose $\gamma(\bx) = \xi(\bx)/2$, say,
\begin{equation}
\gamma(\bx) = \frac{1}{2}\xi(\bx) = \frac{1}{2}\int_{2\mathcal{K}} \hat{V}_W(\bx, \bk) \D \bk.
\end{equation}
Then probabilistic interpretation of Eq.~\eqref{eq.majorant} is given as follows. ${\hat{V}_{W}(\bx, \bk-\bk^{\prime})}/{\xi(\bx)}$ is the instrumental probability density, ${|V_{W}(\bx, \bk-\bk^{\prime})|}/{\hat{V}_W(\bx, \bk-\bk^{\prime})} $ is the probability to generate particles, and ${V_{W}(\bx, \bk-\bk^{\prime})}/{|V_W(\bx, \bk-\bk^{\prime})|} $ determines the particle sign. In addition, the multiplicative functional $\xi(\bx)/\gamma(\bx) = 2$ is regarded as generating two particles carrying opposite signs, according to the fact that 
\begin{equation}
\frac{2V_{W}(\bx, \bk-\bk^{\prime})}{|V_W(\bx, \bk-\bk^{\prime})|} = \frac{V_{W}(\bx, \bk-\bk^{\prime})}{{|V_W(\bx, \bk-\bk^{\prime})|}} - \frac{V_{W} (\bx, \bk^{\prime} - \bk)}{{|V_W(\bx, \bk^{\prime}-\bk)|}}.
\end{equation}
In practice, a hybrid of different strategies sometimes turns out to be a wise choice and more details can be found in Remark \ref{remark_combination}. 

\subsection{Generation of life-length $\tau$}
\label{sec:tau_generation}

The life-length $\tau$ of a particle characterizes the arrival time of the branching event (or scattering event in physical terminology). Mathematically speaking, the purpose of the generation of $\tau$, either stochastic or non-stochastic (in e.g., \textbf{RC}), is to estimate the integral like Eq.~\eqref{def.time_integral}. Different techniques have been summarized in \cite{MuscatoWagnerStefano2010}. 
Below we list several commonly-used techniques in the Wigner simulations and clarify the source of errors.

\subsubsection{Approximate inverse transform method}  

A simple way to evaluate Eq.~\eqref{def.time_integral} is to generate random $\tau$ from the distribution $\mathcal{G}(t^{\prime}; x, t_l)$ through the inverse transform method. It requires to solve the following equation
\begin{equation}\label{expotential_random_lifetime}
u =\mathcal{G}(t_l+\tau;\bm{x},t_l), \quad  -\ln (1-u) = \int_{t_l}^{t_l+\tau} \gamma(\bm{x}(s-t_l)) \D s,
\end{equation}
where $u$ is a uniform random number in $[0, 1)$. 

Unfortunately, solving Eq.~\eqref{expotential_random_lifetime} for $\tau$ is entirely not trivial as the explicit form of $\gamma(\bx)$ is unknown.  The basic version of $\textbf{sp0}$  \cite{SellierNedjalkovDimov2014,NedjalkovSchwahaSelberherr2013} makes use of the choices $\hat{V}_W(\bx, \bk) = V_W^+(\bx, \bk)$ (or $V_W^-(\bx, \bk)$) and $\gamma(\bx) = \xi(\bx)$ in Eq.~\eqref{eq.majorant}
and further adopts the approximation like
\begin{equation}\label{1stapprox1}
\me^{ - \int_{t_l}^{t^{\prime}} \xi(\bx(s-t_l)) \D s} \approx \me^{ -\int_{t_l}^{t^{\prime}} \xi(\bx) \D s} = \me^{-\xi(\bx)(t^{\prime} -t_l)},
\end{equation}
and thus yields the approximation of Eq.~\eqref{def.time_integral}
\begin{equation}\label{def.time_integral_approximation}
\begin{split}
&\int_{t_l}^{t_{l+1}} \Theta_V[\varphi](\bx(t^{\prime}-t_l), \bk, t^{\prime}) \me^{-\int^{t^{\prime}}_{t_l} \xi(\bm{x}(s-t)) \D s} \D t^{\prime} \\
&\approx \int_{t_l}^{t_{l+1}} \left[ \xi(\bx) \me^{-\xi(\bx)(t^{\prime} -t_l)}\right] \cdot \frac{\xi(\bx(t^{\prime}-t_l))}{\xi(\bx)} \cdot \frac{\Theta_V[\varphi](x(t^{\prime}-t_l), \bk^{\prime}, t^{\prime} )}{\xi(\bx(t^{\prime}-t_l)) } \D t^{\prime}.
\end{split}
\end{equation} 
Therefore, it is rather simple to evaluate Eq.~\eqref{def.time_integral_approximation} by generating a life-length $\tilde{\tau}$ depending on the state of particle
\begin{equation}\label{tildetau}
u = \int_{t_l}^{t_l + \tilde{\tau}} \xi(\bx) \me^{-\xi(\bx)(t^{\prime}- t_l)} \D t^{\prime}, \quad \tilde{\tau} =  -{\ln(1-u)}/{\xi(\bx)}.
\end{equation}
The approximation \eqref{def.time_integral_approximation} is usually reasonable as $ \xi(\bx)$ is very flat, but might be too rough to be used when $\xi(\bx)$ undergoes a sharp decrease, see e.g., Fig.~2 in \cite{ShaoXiong2016}.  
That is, $\textbf{sp0}$ ignores the bias caused by the factor $\xi(\bx(\tilde{\tau}))/\xi(\bx)$ and simply sets $\xi(\bx(\tilde{\tau}))/\xi(\bx) \approx 1$. We find such setting poses some restrictions on $\Delta t$ and numerical accuracy is undermined unless $\Delta t$ is sufficiently small because 
\begin{equation}
\frac{\xi(\bx(\tilde{\tau}))}{\xi(\bx)} = \frac{\xi(\bx) + \frac{\hbar \bk \tilde{\tau}}{m} \xi^{\prime}(\bx) + \mathcal{O}(\tilde{\tau}^2)}{\xi(\bx)} \approx 1 
\end{equation}
holds only when $\tilde{\tau} \le \Delta t$ or $\xi^{\prime}(\bx)$ is very small (this corresponds to the case where $\xi(\bx)$ is flat, but $\xi^{\prime}(\bx)$ allows a large value in a neighborhood of $\bx = 0$). In fact, numerical results in Section \ref{sec:comparison} will show that the errors of \textbf{sp0} are almost proportional to $\Delta t$. 

Surprisedly, the approximation \eqref{def.time_integral_approximation} becomes fairly well if the bias $\xi(\bx(\tilde{\tau}))/\xi(\bx)$ is taken into account, even when a very large $\Delta t$ is allowed,
and this results in $\textbf{sp0-I}$, an improved version of $\textbf{sp0}$. When $\tilde{\tau}$ is randomly generated and a branching event occurs, the following rule is employed instead of that of \textbf{sp0} in Alg.~\ref{weighted_particle_WBRW}. 
\leftline{
\fbox{ \shortstack[l] {
For $\textbf{sp0-I}$, $M = 1$ and the $i$-th offspring is generated with $n(i)$ replicas. \quad \quad \quad ~~~\\
$n(1) = n(2) = [\frac{\xi(\bx(\tilde{\tau})}{\xi(\bx)}] + r(\frac{\xi(\bx(\tilde{\tau})}{\xi(\bx)}), \quad n(3) = 1$,\\
$ w^{\prime}_{(1)}= w \cdot \mone_{ \{ \bm{k}^\prime_{1} \in \mathcal{K} \} }, ~~ w^{\prime}_{(2)}= -w \cdot \mone_{ \{ \bm{k}^\prime_{2} \in \mathcal{K} \} }, ~~ w^{\prime}_{(3)}= w$. 
}}
}
Here we adopt the convention that $[m]$ is the integer part of $m$ and 
\begin{equation}\label{1}
r(m) = \left\{
\begin{split}
&1, \quad \textup{with probability} ~~  m - [m],\\
&0, \quad  \textup{with probability} ~~  1 - (m - [m]).
\end{split}
\right.
\end{equation}
Note in passing that the bias $\xi(\bx(t^{\prime}-t))/\xi(\bx)$ was treated as the particle weight in $\textbf{wp}$ in a natural manner,
and a careful benchmark in \cite{ShaoXiong2016} has shown that a longer time step is consequently allowed while the accuracy is still maintained.

It deserves to mention that Eq.~\eqref{expotential_random_lifetime} has an exact solution 
$$\tau = -\ln(1-u)/\gamma_0$$ when $\gamma(\bx) \equiv \gamma_0$, and thus the accuracy is no longer influenced by the choice of $\Delta t$. This is much more preferable in real applications and thus adopted in $\textbf{sp1}, \textbf{sp2}$ and $\textbf{wp}$.
Actually, once using a constant auxiliary function, whatever large or small time step $\Delta t$ we choose, it does not produce any impact on the expectation of the total particle number of the branching process (see Theorem 7 of \cite{ShaoXiong2016}).
That is, the errors will not be amplified by a larger time step $\Delta t$. 
Actually, $\Delta t$ is only required for measuring a quantum system.

\subsubsection{Self-scattering technique}

The self-scattering technique is a common approach in simulating the transport equations \cite{KosinaNedjalkovSelberherr2003jce, RjasanowWagner1996}. It suggests to choose a fictitious self-scattering rate $\gamma_s(\bx)$ and set $\xi(\bx) + \gamma_s(\bx)$ to be the auxiliary function, instead of $\xi(\bx)$. In this manner, the splitted transition kernel \eqref{eq.Gamma} becomes
\begin{equation}\label{eq.Gamma_self}
\begin{split}
\Gamma(\bm{x}, \bm{k}; \bm{x}^{\prime}, \bm{k}^{\prime}) = & \frac{\xi(\bx)}{\xi(\bx) + \gamma_s(\bx)} \left[\frac{V^{+}_{W}(\bm{x},\bm{k}-\bm{k}^{\prime})}{\xi(\bx)} - \frac{V^{-}_{W}(\bm{x},\bm{k}-\bm{k}^{\prime})}{{\xi(\bx)}}\right] \cdot \delta(\bm{x}-\bm{x}^{\prime}) \\
& + \left[\frac{\xi(\bx)}{\xi(\bx) + \gamma_s(\bx)} + \frac{\gamma_s(\bx)}{\xi(\bx) + \gamma_s(\bx)}\right] \cdot \delta(\bk-\bk^{\prime}) \cdot \delta(\bm{x}-\bm{x}^{\prime}),
\end{split}
\end{equation}
It is convenient to set the total scattering rate to be $\xi(\bx) + \gamma_s(\bx) = \gamma_0$, 
and thus ${\xi(\bx)}/{\gamma_0}$ is interpreted as the probability to generate two offsprings according to the splitted Wigner kernel. 

It can be easily verified that this self-scattering technique reduces to $\textbf{sp1}$ when $M=1$. However, when $M \ge 2$, since different components of the Wigner kernel have distinct normalizing factors, one need to replace $\xi(\bx)$ in Eq.~\eqref{eq.Gamma_self} by $\sum_{m=1}^M \xi_m(\bx)$. Thus an appropriate choice of the self-scattering rate may be $\gamma_s(\bx) = \gamma_0 - \sum_{m=1}^M \xi_{m}(\bx)$ and each branch is selected with probability $\xi_m(\bx)/\gamma_0$. By contrast, $\textbf{sp1}$ suggests to choose $\gamma_0 \ge \check{\xi}$ and at most $M$ branches of particles may be generated when a branching event occurs. Furthermore, in some occasions, we can even allow $\gamma_0 < \xi_m(\bx)$ for some $m$ and $\bx$ and interpret $\xi_m(\bx)/\gamma_0$ as generating $[\xi_m(\bx)/\gamma_0] + r(\xi_m(\bx)/\gamma_0)$ replicas of particles, as implemented in \textbf{sp0-I}. This gives us more freedom in choosing the auxiliary function $\gamma_0$.

\subsubsection{Constant time technique}

Unlike above-mentioned strategies, this approach is suggested to use the constant time step $\Delta t$ for time evolution, instead of a random $\tau$. Suppose $\Delta t$ is sufficiently small. Starting from a simple approximation 
\begin{equation}
\me^{-\int_{t_l}^{t_{l+1}} \gamma(\bx(s-t))\D s} = 1 - \gamma(\bx(\Delta t))\Delta t + \mathcal{O}((\Delta t)^2),
\end{equation}
the constant time technique is actually an explicit one-step Euler method as
\begin{equation}\label{def.time_integral_backward_Euler}
\begin{split}
&\int_{t_l}^{t_{l+1}}\Theta_V[\varphi](\bx(t^{\prime}-t_l), \bk, t^{\prime}) \me^{-\int^{t^{\prime}}_{t_l} \gamma(\bm{x}(s-t_l)) \D s} \D t^{\prime} \\
&=  \Theta_V[\varphi](\bx(\Delta t), \bk, t_{l+1}) \me^{-\int_{t_l}^{t_{l+1}} \gamma(\bx(s-t_l)) \D s} \Delta t + \mathcal{O}((\Delta t)^2) \\
& =\frac{\Theta_V[\varphi](\bx(\Delta t), \bk, t_{l+1})}{\gamma(\bx(\Delta t))} \cdot \left[\gamma(\bx(\Delta t))-\gamma^2(\bx(\Delta t)) \Delta t\right] \Delta t + \mathcal{O}((\Delta t^2)).
\end{split}
\end{equation}
When the $\mathcal{O}((\Delta t)^2)$ terms are omitted, it yields the approximation to Eq.~\eqref{Adjoint_renewal_type_equation}
\begin{equation}
\begin{split}
\varphi(\bx, \bk, t_{l}) \approx &\varphi_T(\bx(\Delta t), \bk) + \left(\gamma(\bx(\Delta t))\Delta t\right) \cdot \frac{\Theta_V[\varphi](\bx(\Delta t), \bk, t_{l+1})}{\gamma(\bx(\Delta t))}. 
\end{split}
\end{equation}
Therefore the branching event occurs with the probability $\gamma(\bx(\Delta t)) \Delta t$,
which requires $0\leq \gamma(\bx(\Delta t)) \Delta t \le 1$.
Thus it might not be very efficient when $\gamma(\bx)$ has a large upper bound.

\subsection{A bootstrap filter for weighted particles}
\label{sec:bootstrap}

From the view of computation, \textbf{sp1} is advantageous over \textbf{wp} in data storage as the sign functions can be operated as integers. However, later we will show that the variance of \textbf{sp1} is usually larger than that of \textbf{wp} and the accuracy cannot be improved by adjusting $\gamma_0$. 
The second strategy \textbf{sp2} is  to get the advantages in both data storage and variance reduction.  Based on \textbf{wp}, we introduce a bootstrap filtering step to adjust $N$ weighted particles into $N$ signed particles.

The key idea of the bootstrap filtering is to eliminate the particles having low importance weights and to multiply particles having higher importance weights, thereby avoiding the skewness of the particle weights \cite{bk:DoucetDeFreitasGordon2001, bk:Liu2001, bk:RobertCasella2004}.  Mathematically speaking, the target of the bootstrap resampling is to approximate a weighted empirical measure by a unweighted one.
In this work, we only utilize the simplest version of bootstrapping for a prototype testing. In fact, more sophisticated bootstrapping techniques, such as the particle island model, may be adopted and one can refer to  \cite{VergeDubarryDelMoralMoulines2015} for more details.

Suppose we have a collection of $N$ particles, $\mathcal{S} = \{ (\bx_i, \bk_i)\}_{i=1}^N$, which are weighted by the collection $\mathcal{W} = \{ w_i\}_{i=1}^N$, $\sum_{i=1}^{N} w_i = N$. First, we consider $w_i  > 0$ and Alg.~\ref{res_bootstrapping} is just a typical way to achieve the filtering \cite{bk:Liu2001}, with Lemma \ref{theorem_bootstrap_bound} giving its variance estimation.

\begin{algorithm} 
\caption{Residual bootstrap filtering \label{res_bootstrapping}}
\begin{itemize}

\item[1.] Retain $k_i = [w_i]$ copies of $(\bx_i, \bk_i) \in \mathcal{S}$, $i=1, \cdots, N$, where $[w_i]$ indicates the largest integer that doesn't exceed $w_i$. Let $N_r = N- k_1-\cdots - k_N$.

\item[2.] Obtain the remaining $N_r$ i.i.d. draws from $\mathcal{S}$, where each particle $(\bx_i, \bk_i)$ is selected with probability $(w_i-k_i)/N_r$, $i=1, \cdots, N$.

\item[3.] Assign the weight $1$ to  new sample set,  denoted by $\tilde{\mathcal{S}} = \{(\tilde{\bm{x}}_i, \tilde{\bk}_i)\}_{i=1}^N$.
\end{itemize}

\end{algorithm}

\begin{lemma} \label{theorem_bootstrap_bound}
For any bounded measurable function $\varphi$,  we have
\begin{equation} \label{bootstrap_bound}
\mathbb{E} \left| \Big \langle \varphi, \frac{1}{N}\sum_{i=1}^N w_i \delta_{(\bx_i, \bk_i)}\rangle  - \langle \varphi,  \frac{1}{N}\sum_{i=1}^N \delta_{(\tilde{\bx}_i, \tilde{\bk}_i)} \Big \rangle \right|^2  \le 2\left(\frac{N_r}{N}+1\right) \frac{\Vert \varphi \Vert^2}{N},
\end{equation}
where $w_i > 0$, $\sum_{i=1}^N w_i = N$, $\Vert \varphi \Vert = \max |\varphi(x)|$ and the set $\{(\tilde{\bx}_i, \tilde{\bk}_i)\}_{i=1}^N$ is produced by Alg.~\ref{res_bootstrapping}.
\end{lemma}

Now we turn to extend the bootstrap filter to the signed weighted empirical measure, where $w_i \in [-1, 1]$. A natural way is to split the signed particles into two batches according to their signs and thus the signed weighted empirical measure can be decomposed into two weighted ones. Afterwards, we use bootstrap filters to tackle such two empirical measures through Alg.~\ref{res_bootstrapping} separately. 

Suppose we have a collection of $N$ particles, $\mathcal{S} = \{ (\bx_i, \bk_i)\}_{i=1}^N$, which are weighted by the collection $\mathcal{W} = \{ w_i\}_{i=1}^N$, with $w_i  \in [-1, 1]$ and $\sum_{i=1}^{N} w_i = N_\alpha$. {As discussed above, we can also divide $\mathcal{S}$ into $N^+$ positive weighted particles $\mathcal{S}^+ = \{ (\bx_i^{+}, \bk_i^+) \}_{i=1}^{N^+}$ and $N^-$ negative weighted ones $\mathcal{S}^- = \{ (\bx_i^-, \bk_i^-) \}_{i=1}^{N^-}$. Then we can rewrite the signed weighted empirical measure as
\begin{equation}
\begin{split}
\frac{1}{N_\alpha} \sum_{i=1}^{N} w_i \delta_{(\bx_i, \bk_i)} &= \frac{1}{N_\alpha} \sum_{i=1}^{N^+} w^+_i \delta_{(\bx^+_i, \bk^+_i)} + \frac{1}{N_\alpha} \sum_{i=1}^{N^-} w_i^- \delta_{(\bx_i^-, \bk^-_i)} \\
&= \lambda^+ \sum_{i=1}^{N^+} \frac{w_i^+}{\sum_{i=1}^{N^+} w_i^+} \delta_{(\bx^+_i, \bk^+_i)} +  \lambda^- \sum_{i=1}^{N^-} \frac{w_i^-}{\sum_{i=1}^{N^-} w_i^-} \delta_{(\bx^-_i, \bk^-_i)} \\ 
& = \lambda^+  \sum_{i=1}^{N^+} \tilde{w}_i^+ \delta_{(\bx^+_i, \bk^+_i)} + \lambda^- \sum_{i=1}^{N^-} \tilde{w}_i^- \delta_{(\bx^-_i, \bk^-_i)},
\end{split}
\end{equation}
where 
\begin{equation}\label{eq:lambda}
\lambda^\pm = \frac{1}{{N_\alpha}} {\sum_{i=1}^{N^\pm} w_i^\pm},
\quad \lambda^+ > 0, \quad \lambda^- < 0, 
\end{equation}
and $\tilde{w}_i^+$ and $ \tilde{w}_i^-$ are normalized weights for $\mathcal{S}^+$ and $\mathcal{S}^-$, respectively. 

It remains to employ the bootstrap filter to convert the weighted empirical measures into unweighted ones, yielding the following approximation
\begin{equation}\label{double_bootstrap}
\frac{1}{N_\alpha} \sum_{i=1}^{N} w_i \delta_{(\bx_i, \bk_i)}  \sim   \frac{\lambda^+}{N^+} \sum_{i=1}^{N^+} \delta_{(\tilde{\bx}_i^+, \tilde{\bk}^+_i)} + \frac{\lambda^-}{N^-} \sum_{i=1}^{N^-} \delta_{(\tilde{\bx}_i^-, \tilde{\bk}^-_i)},
\end{equation} 
Similarly, we can derive the variance estimation for such bootstrap filter.
\begin{theorem} \label{theorem_double_bootstrap_bound}
For any bounded measurable function $\varphi$,  we have
\begin{equation} \label{double_bootstrap_bound}
\mathbb{E} \left| \Big \langle \varphi, \frac{1}{N_\alpha}\sum_{i=1}^{N} w_i \delta_{(\bx_i, \bk_i)} - \frac{\lambda^+}{N^+}\sum_{i=1}^{N^+} \delta_{(\tilde{\bx}_i^+, \tilde{\bk}_i^+)} - \frac{\lambda^-}{N^-}\sum_{i=1}^{N^-} \delta_{(\tilde{\bx}^-_i, \tilde{\bk}^-_i)}  \Big \rangle \right|^2  \le  8\left(\frac{|\lambda^+|^2}{N^+} + \frac{|\lambda^-|^2}{N^-} \right) \Vert \varphi \Vert^2,
\end{equation}
where $w_i \in [-1, 1]$, $\sum_{i=1}^{N} w_i  = N_\alpha = N^+ - N^-$, $\Vert \varphi \Vert = \max |\varphi(x)|$ and the sets $\{(\tilde{\bx}_i^\pm, \tilde{\bk}_i^\pm)\}_{i=1}^{N^\pm}$ are produced by Alg.~\ref{res_bootstrapping}.
\end{theorem}

\section{Implementation details}
\label{sec:imp}

A complete period of the WBRW is constituted of three parts: sampling from the Wigner function, running the branching random walk and resampling, in which both rejecting sampling and importance sampling play a crucial role. All the details on the implementation are provided in this section, as sketched by the diagram below. 
\begin{equation*}
\boxed{\small \text{Histogram}} \xrightarrow{\textup{(Alg.\ref{initial_sampling}})} \boxed{\small \text{Particles}} \xrightarrow[\textup{\textbf{sp2} (Alg.\ref{weighted_particle_WBRW}+Alg.\ref{res_bootstrapping}})]{\textup{\textbf{sp0}, \textbf{sp1}, \textbf{wp}, \textbf{RC} (Alg.\ref{weighted_particle_WBRW}})} \boxed{\small \text{Particles}} \xrightarrow[\textup{Eq.~\eqref{boostrap_histogram}}]{\textup{Eq.~\eqref{piecewise_constant_function}}} \boxed{\small \text{Histogram}}
\end{equation*}
The histogram provides a compact representation of the particle system and stores all the information in a single matrix. The number of signed particles in each bin and their coordinates are recovered through Alg.~\ref{initial_sampling}. Then we simulate their branching random walks according to Alg.~\ref{weighted_particle_WBRW} and finally construct a new histogram through Eq.~\eqref{piecewise_constant_function} for \textbf{sp0}, \textbf{sp1}, \textbf{wp} and \textbf{RC} or Eq.~\eqref{boostrap_histogram} for \textbf{sp2}.

The WBRW algorithm can be performed in a simple parallelizing manner. For each processor,  we pick up several bins and draw one cluster of particles, then run the branching random walk and reconstruct the histogram independently. The remaining task is to merge the histograms.  The signed-particle implementations are certainly advantageous because both the updating of particle weights and the merging of histograms, that involves a summation of large matrices, can be performed by integer arithmetic operation.


\subsection{Sampling from the Wigner function}

Suppose {$\{\mathsf{D}_\nu \}_{\nu=1}^{N_h}$} gives a partition of $\mathsf{D}$, say, $\mathsf{D}_\nu$ are mutually disjoint bins and $\mathsf{D} = \bigcup_{\nu=1}^{N_h} \mathsf{D}_\nu $. The Wigner function at instant $t_l$ is a piecewise constant function (histogram), produced either by the particle reduction in the previous step or by the following construction for the initial data:
\begin{equation}
f(\bx, \bk, t_l) \approx \sum_{\nu=1}^{N_h} d_\nu(t_l)
 \cdot \mone_{\mathsf{D}_\nu}(\bx, \bk),~~d_\nu(t_l) = f(\bx_\nu, \bk_\nu, t_l),
\end{equation}
where $N_{h}$ is the partition size of the histogram, $(\bx_\nu, \bk_\nu)$ is the centre of the bin $\mathsf{D}_\nu$ and the instrumental distribution is simply given by
\begin{equation}\label{instrumental_density}
f_I(\bx, \bk, t_l) = H(t_l)^{-1} \sum_{\nu=1}^{N_h} |d_\nu(t_l)| \cdot \mone_{\mathsf{D}_\nu}(\bx, \bk),~~H(t_l) = \sum_{\nu=1}^{N_{h}} |d_\nu(t_l)|.
\end{equation}

Now for a given sample size $N_\alpha$, the initial weighted particles at $t_l$ are sampled through the Alg.~\ref{initial_sampling}. The number of particles allocated in $\mathsf{D}_\nu$ is determined by the following rounding step and the initial weight is determined by $ |d_\nu(t_l)|/ d_\nu(t_l)$ according to Eq.~\eqref{instrumental_density}. The coordinate of each particle is produced randomly according to a locally uniform distribution.

\begin{algorithm}
\caption{Importance sampling from the quasi-probability density}\label{initial_sampling}

\begin{itemize}

\item[1.] (Rounding step) Draw a random number $u$ from the uniform distribution in $[0, 1)$ and determine the particle number $n_\nu$ in  $\mathsf{D}_\nu$ by
\begin{equation}
n_\nu = [m] + p, \quad m = N_\alpha \cdot \mu(\mathsf{D}_\nu) \cdot |d_\nu(t_l)|,
\end{equation}
where $[m]$ is the integer part of $m$ and  
\begin{equation}
p = \left\{
\begin{split}
&1, \quad u \le  m - [m],\\
&0, \quad u >  m-[m].
\end{split}
\right.
\end{equation} 

\item[2.] Endow each particle in the bin $\mathsf{D}_\nu$ with a weight $|d_\nu(t_l)|/d_\nu(t_l) = 1$ or $-1$. 

\item[3.] Endow each particle in the bin $\mathsf{D}_\nu$ with a random coordinate $(\bx, \bk)$ according to the locally uniform distribution $(\mu({\mathsf{D}_\nu}))^{-1} \cdot \mone_{\mathsf{D}_\nu}(\bx, \bk)$. 

\end{itemize}

\end{algorithm}

\begin{remark}
It notes that for \textbf{sp1}, the importance sampling is greatly simplified as $m$ is an integer according to Eq.~\eqref{signed_histogram}. Therefore, the rounding step in Alg.~\ref{initial_sampling} can be ignored. 
\end{remark}

\subsection{Run the Wigner Branching Random Walk}

After allocating the particles at the instant $t_l$, we begin to simulate the branching random walk of particles until $t_{l+1}$ according to Alg.~\ref{weighted_particle_WBRW}. We require to draw samples according to the splitted Wigner kernel and calculate $\xi(\bx)$ (or $\xi_m(\bx)$), both of which can be resolved by a simple rejection method. 

The key of rejection sampling is to find an appropriate sampling distribution and a covering constant \cite{bk:RobertCasella2004}.  Suppose we would like to draw samples according to $V_{W}^+(\bx, \bk)/ \xi(\bx)$. For a fixed $\bx$, it's necessary to find a normalized sampling distribution $g(\bk)$ and a covering constant $M$, such that
\begin{equation}
V_{W}^+(\bx, \bk) \le M g(\bk).
\end{equation}

The rejection sampling can also be used in calculating the normalizing factor $\xi(\bx)$. In order to reduce the computational cost, we can calculate $\xi(\bx)$ on a grid mesh in advance, and use interpolations to obtain the points that are not located on the mesh. Actually, linear interpolation is a desirable choice to strike a balance between accuracy and efficiency, as $\xi(\bx)$ is very flat except a small central region, where it decays to $0$ sharply.

\begin{algorithm} 
\caption{Rejection sampling from $V_W^+(\bx, \cdot)$}\label{algo_rejection_sampling}
\begin{itemize}

\item[1.] Draw a sample $\bk$ from $g(\bk)$ and compute the ratio $p = V_{W}^+(\bx, \bk)/[Mg(\bk)]$.

\item[2.] Draw a uniform random number $u$ in [0, 1).

\item[3.] If $p < u$, we accept and return $\bk$; otherwise, we reject $\bk$ and return to step 1.

\end{itemize}

\end{algorithm}

We mainly focus on the potential with symmetry along an axis $\bx = \bx_c$
\begin{equation}
V(\bm{x}_c+\bm{x}) = V(\bm{x}_c - \bm{x}), 
\end{equation}
such as the Gaussian barrier and two-body Yukawa interactions, and the corresponding Wigner kernel reads
\begin{equation}\label{wigner_kernel_Gaussian}
V_W(\bm{x}, \bm{k}) = \frac{1}{\hbar\pi} \mathcal{F}[V](2\bm{k}) \sin(2(\bm{x}-\bm{x}_c)\cdot \bm{k}),
\end{equation}
where $\mathcal{F}[V]$ denotes the Fourier transform of $V$. In the experiments, we will use the following relations: 
\begin{align}
&V(x) = \frac{1}{\sqrt{2\pi}}\me^{-{x^2}/{2}}, \quad \mathcal{F}[V](k) =  \sqrt{\frac{2}{\pi}} \me^{-2k^2}, \\
&V(x) = \frac{\me^{-\kappa |x |}}{2\kappa}, \quad \mathcal{F}[V](k) = \frac{1}{k^2 + \kappa^2}.
\end{align}
Subsequently, it suffices to choose the majorant $g(k) = \mathcal{F}[V](2k)$ and the covering constant $M = 1/(\hbar\pi)$. The ratio $p$ reduces to $[\sin(2(\bm{x}-\bm{x}_c)\cdot \bm{k})]^+$, where $[\cdot ]^+$ means the positive part. When $\mathcal{F}[V]$ does not have a close form, drawing according to $g(\bk)$ becomes a tough task and one has to resort to the rejection sampling, like Eq.~\eqref{eq.majorant} or other advanced techniques, such as the Markov Chain Monte Carlo method \cite{bk:RobertCasella2004}. 

\begin{remark}\label{remark_combination}
It notes that in Alg.~\ref{algo_rejection_sampling},  the probability of accepting a proposed sample is proportional to $p$, thus its efficiency will deteriorate for a small $|\bm{x}-\bm{x}_c|$. In this case, it is better to choose a majorant $\hat{V}_{W, m}(\bx, \bk) = \mathcal{F}[V_m](2\bk)$ directly in Eq.~\eqref{eq.majorant} and
\begin{equation}\label{majorant_approach}
\xi_m(\bx) \equiv \int_{2\mathcal{K}} \mathcal{F}[V_m](\bk) \D \bk, \quad \frac{|V_{W, m}(\bx, \bk)|}{\hat{V}_{W, m}(\bx, \bk)} = \frac{\left|\sin(2(\bx- \bx_c) \cdot \bk \right|}{\hbar \pi}.
\end{equation}
Thus $\frac{\hat{V}_{W, m}(\bx, \bk)}{\xi_m(\bx)}$ is the instrumental probability density, $\left|\sin(2(\bx- \bx_c) \cdot \bk \right|$ determines the particle sign and $\frac{\xi_m(\bx)}{\gamma_0} \cdot \frac{|V_{W, m}(\bx, \bk)|}{\hat{V}_{W, m}(\bx, \bk)}$ is treated as either the importance weight or the rejecting ratio. The strategy overcomes the potential weakness of Alg.~\ref{algo_rejection_sampling} and saves the efficiency. Although $\xi_m(\bx)$ in Eq.~\eqref{majorant_approach} is larger than $\int_{2\mathcal{K}}V_{W, m}^+(\bx, \bk) \D \bk$, a relatively small $\gamma_0 < \xi_m(\bx)$ is still allowed when $|\bx -\bx_c|$ is sufficiently small.
\end{remark}

\subsection{Resampling of particles} 

The resampling technique intends to control the particle number in the simulations, as well as suppress the growth of variances. It arises from the fact that the branching treatment will inevitably lead to an exponential growth of particle number, which is undesirable for long-time simulations (see Theorem \ref{th:exp}). 

The terminology `resampling', in fact,  has different meanings in applications. In the diffusion Monte Carlo or the sequential importance sampling, the resampling procedure, also called the bootstrap filtering, is to adjust the skewness of importance weights \cite{bk:DoucetDeFreitasGordon2001}.  In the Boltzmann simulations, the so-called particle resampling is aimed at reducing the number of positive and negative particles, which describe the deviation from the Maxwell distribution \cite{RjasanowWagner1996, YanCaflisch2015}. As summarized in  \cite{YanCaflisch2015}, the particle reduction strategies are roughly divided into three categories: the particle cancelation, the particle thermalization and the spectral filtering.
In the Wigner simulations, the particle cancellation is naturally adopted due to the sign weights which reflect
both positive and negative values of the Wigner function. It borrows the idea of histogram  \cite{bk:LaszloGyorfi2002, bk:HastieTibshiraniFriedman2009} and is essentially a piecewise constant reconstruction of a (quasi-)density function. As presented in Alg.~\ref{initial_sampling},  the output will be set as the initial data for the next period, from which a new set of particles is drawn through the importance sampling.

Suppose all particles lie in a compact set $\mathsf{D} \subset \mathbb{R}^n$ and $\{\mathsf{D}_\nu \}_{\nu=1}^{N_h}$ is a partition of $\mathsf{D}$. For \textbf{wp}, a piecewise constant function $p(\bx, \bk)$ has the form
\begin{equation}\label{piecewise_constant_function}
p(\bx, \bk) = \sum_{\nu = 1}^{N_h} \frac{\sum_{i=1}^{N} w_i \cdot \mone_{\mathsf{D}_\nu}(\bx_i, \bk_i)}{N_\alpha} \cdot \frac{\mone_{\mathsf{D}_\nu}(\bx, \bk)}{\mu(\mathsf{D}_\nu)}.
\end{equation}
In particular, for \textbf{sp1}, since the weight of each particle is either $-1$ or $1$, we only need to count the particle number in each bin and sum over their signs, yielding
\begin{equation}\label{signed_histogram}
p(\bx, \bk) = \sum_{\nu =1}^{N_h} \frac{n_\nu^+ - n_\nu^-}{N_\alpha} \cdot \frac{\mone_{\mathsf{D}_\nu}(\bx, \bk)}{\mu(\mathsf{D}_\nu)},
\end{equation}
where $n_\nu^+$ and $n_\nu^-$ are counts of the positive and negative particles in $\mathsf{D}_\nu$, respectively. Since the particles carrying opposite signs are eliminated, such procedure is also called particle cancelation (annihilation) \cite{SellierNedjalkovDimov2014, YanCaflisch2015}.  

The construction of the histogram in \textbf{sp2} is a composition of two histograms, based on filtered positive and negative particles. Starting from Eq.~\eqref{double_bootstrap}, we have
\begin{equation}\label{boostrap_histogram}
\tilde{p}(\bx, \bk) = \sum_{\nu=1}^{N_h}  (\frac{\lambda^+ \tilde{n}_\nu^+}{N^+}+\frac{\lambda^- \tilde{n}_\nu^-}{N^-})\cdot \frac{\mone_{\mathsf{D}_\nu}(\bx, \bk)}{\mu(\mathsf{D}_\nu)},
\end{equation}
where  $\tilde{n}_\nu^+$ and $\tilde{n}_\nu^-$ are counts of the filtered $N^+$ positive and $N^-$ negative particles in $\mathsf{D}_\nu$, respectively. 


The remaining task is to choose the bins $\mathsf{D}_\nu$, which are usually (hyper-)rectangles based on a given partition of $\mathsf{D}$. The simplest and most ubiquitous choice is to place a uniform grid mesh $\mathsf{D}_\nu = \mathcal{X}_{\nu_1} \times \mathcal{K}_{\nu_2}$, with $\varepsilon$  denoting the maximal diameter
\begin{equation}
\varepsilon = \max_{\nu} \sup_{\by_1, \by_2 \in \mathsf{D}_\nu} \Vert \by_1 - \by_2 \Vert. 
\end{equation}
The following theorem presents the error bound $\mathcal{O}(\varepsilon)$ for such kind of histogram.

\begin{theorem}\label{error_uniform_histogram}
Suppose  $\{\mathsf{D}_\nu\}_{\nu=1}^{N_h}$ gives a partition on a compact set $\mathsf{D} \subset \mathbb{R}^{2n}$ with the maximal diameter $\varepsilon$.  
Let $\varphi$ be a bounded measurable function on $\mathsf{D}$ and locally H\"{o}lder continuous in each bin $\mathsf{D}_\nu$, say, 
\begin{equation}\label{Holder_condition}
\Vert \varphi \Vert_{C^{0,\alpha}} = \max_\nu \sup_{\by_1 \ne \by_2 \in \mathsf{D}_\nu} \frac{ \big | \varphi(\by_1) -  \varphi (\by_2) \big|}{\Vert \by_1 - \by_2 \Vert^{\alpha} } < \infty
\end{equation}
for an exponent $\alpha > 0$. Then we have 
\begin{equation}\label{histogram_error}
\left| \Big \langle \varphi, \frac{1}{N_\alpha}\sum_{i=1}^{N} w_i \delta_{(\bx_i, \bk_i)} \Big \rangle - \langle \varphi, p \rangle \right| \le  (|\lambda^+| +|\lambda^-|) \cdot \varepsilon^{\alpha} \Vert \varphi \Vert_{C^{0,\alpha}} ,
\end{equation}
where $p$ is given by Eq.~\eqref{piecewise_constant_function}. Moreover, 
for the piecewise constant reconstruction $\tilde{p}$ in Eq.~\eqref{boostrap_histogram},
it further yields 
\begin{equation}\label{histogram_bootstrap_error}
\mathbb{E}\left| \Big \langle \varphi, \frac{1}{N_\alpha}\sum_{i=1}^{N} w_i \delta_{(\bx_i, \bk_i)} \Big \rangle - \langle \varphi, \tilde{p} \rangle \right|^2 \le C_1 \varepsilon ^{2\alpha} \Vert \varphi \Vert_{C^{0,\alpha}}^2    + 
 C_2 \frac{\Vert \varphi \Vert^2}{N_\alpha},
\end{equation}
where
\[
C_1 = 2(|\lambda^+| + |\lambda^-|)^2, \quad C_2 = 16N_\alpha (\frac{|\lambda^+|^2 }{N^+} + \frac{|\lambda^-|^2 }{N^- }).
\]
\end{theorem}

Theorem \ref{error_uniform_histogram} guarantees the convergence of both the uniform histogram and the bootstrap filtering if letting $\varepsilon \to 0$ and $N_\alpha \to \infty$. Numerical experiments in the next section will demonstrate that the resampling based on a uniform partition works well when sample size $N_\alpha$ is comparable to $N_h$, and works poorly when $N_\alpha \ll N_h$ due to the curse of dimensionality as expected \cite{YanCaflisch2015, bk:HastieTibshiraniFriedman2009}.  

\section{Performance evaluation}
\label{sec:num}

We have already illustrated both theoretical and numerical aspects of the computable WBRW, and now it resorts to some benchmark tests for a detailed performance evaluation and to show how the choices of parameters influence the accuracy and the efficiency. 

First, we consider two typical test problems as already employed in \cite{ShaoXiong2016}: the 2D Gaussian scattering and the 4D Helium-like system.  All the parameters are identical to those in \cite{ShaoXiong2016}, except that the average position $x_0$ of the initial data in the 2D problem is reset to be $-10$. Under this 
new initial position can we observe the scattering phenomenon clearly for a relatively short final time $t_{fin} = 10$fs. Performance metrics include the normalized relative errors for the Wigner function $\textup{err}_{wf}(t)$, the spatial marginal density $\textup{err}_{sm}(t)$ and the momental marginal density $\textup{err}_{mm}(t)$. For details, the readers can refer to \cite{ShaoXiong2016}.
Then,  we will give additional two experiments. The first is a comparison study of all the existing signed-particle implementations, including \textbf{sp0}, \textbf{RC} and the proposed \textbf{sp0-I}, \textbf{sp1}, \textbf{sp2}. The second is the 2D Gaussian scattering under a time-dependent double barrier potential. Such example is used to demonstrate the applicability of $\textbf{sp1}$, $\textbf{sp2}$ and $\textbf{wp}$ in dealing with time-dependent potentials, which may reflect the properties of band-structure in many nano-electronic applications.

For resampling, we divide equally the time interval $[0, t_{fin}]$ into $n_A$ subintervals with the partition being
\begin{equation}
0 = t^0 < t^1 < t^2 < \cdots < t^{n_A} = t_{fin}, \quad n_A = t_{fin}/T_A.
\end{equation}
The resampling occurs at $t^i$, with $1/T_A$ the resampling frequency. In general, the resampling instants $\{t^i\}_{i=1}^{n_A}$ constitutes a subset of the temporal partition $\{t_i\}_{i=1}^{n}$ given in \eqref{def.time_partition}, as such procedure is performed only when the particle number is about to exceed a given threshold. The time step is denoted by $\Delta t = t_{l+1}- t_l$. The number of particles after resampling at the instant $t$ is denoted by $\#_{P}^{a}(t)$.

\subsection{Sample size $N_\alpha$}

We claim that the accuracy can be systematically improved by increasing the sample size $N_{\alpha}$.  The simulations of the 2D Gaussian scattering are performed with the auxiliary function $\gamma = 1.5\check{\xi}$ and $T_A= 1$fs fixed. First, pure Monte Carlo simulations are performed to examine the theoretical convergence order $\mathcal{O}(N_\alpha^{-1/2})$ at $t_{fin} = 10$fs. Afterwards the resampling is turned on and the experiments are reinitialized under different sample sizes $N_\alpha$ ranging from $1\times 10^5$ to $10^7$, as presented in Fig.~\ref{G_ss_rate} and Fig.~\ref{G_ss_error}. 
Six groups of the 4D Helium-like system are also simulated with $N_\alpha$ ranging from $10^6$ to $10^8$, $\gamma = 2$ and $T_A = 0.5$a.u. The convergence rate with respect to $N_\alpha$ (at $t_{fin}=20$a.u.) is shown in Fig.~\ref{He_ss_rate}. To visualize the numerical errors, we plot both the spatial and momental marginal distributions under different instants $t=5, 10, 15, 20$a.u. in Fig.~\ref{He_ss_error} and the reduced Wigner function at the instant $t=20$a.u. in Fig.~\ref{He_ss_pic}. Based on these numerical results, we are able to find out the following observations.

\begin{description}

\item[(1)] The accuracy can be systematically improved by increasing the sample size $N_\alpha$. 
We find that the convergence order of the pure Monte Carlo method perfectly coincides with the theoretical prediction of $-1/2$, whereas there exists some deviations when the resampling is turned on.
Too small $N_\alpha$ yields very poor numerical results, as shown in Figs.~\ref{He_ss_error} and \ref{He_ss_pic}. The serve oscillations even shadow the true solution and the reduced Wigner function is too noisy to be observed.  As $N_\alpha$ goes larger, the results produced by  \textbf{sp1} and \textbf{wp} fit the spatial marginal density and momentum marginal better, and the stochastic noises in the (reduced) Wigner function are significantly suppressed.

\item[(2)] The stochastic errors grow exponentially in time, so does the variance of the branching random walk model \cite{bk:Harris1963}. The accumulation of stochastic errors, in fact, may be related to the time evolution of the variance of the branching random walk. On the contrary, the resampling, although introducing some deterministic errors as shown in Theorem \ref{error_uniform_histogram}, helps to suppress the random noises significantly.

\item[(3)] Under the same constant auxiliary function $\gamma$, \textbf{wp} provides more accurate results than \textbf{sp1}. However, when $N_\alpha$ is chosen sufficiently large (such as $N_\alpha = 10^{7}$ in the 2D test and $N_\alpha = 10^8$ in the 4D test), both implementations produce the numerical with almost the same accuracy. Although the particle number after resampling in \textbf{sp1} is larger than that of \textbf{wp}, the requirement of memory is even alleviated as the histogram is stored in an integer-valued matrix for the former.
\end{description}

\begin{figure}[!h]
\subfigure[Convergence rate with respect to $N_\alpha$.]{{\includegraphics[width=0.49\textwidth,height=0.27\textwidth]{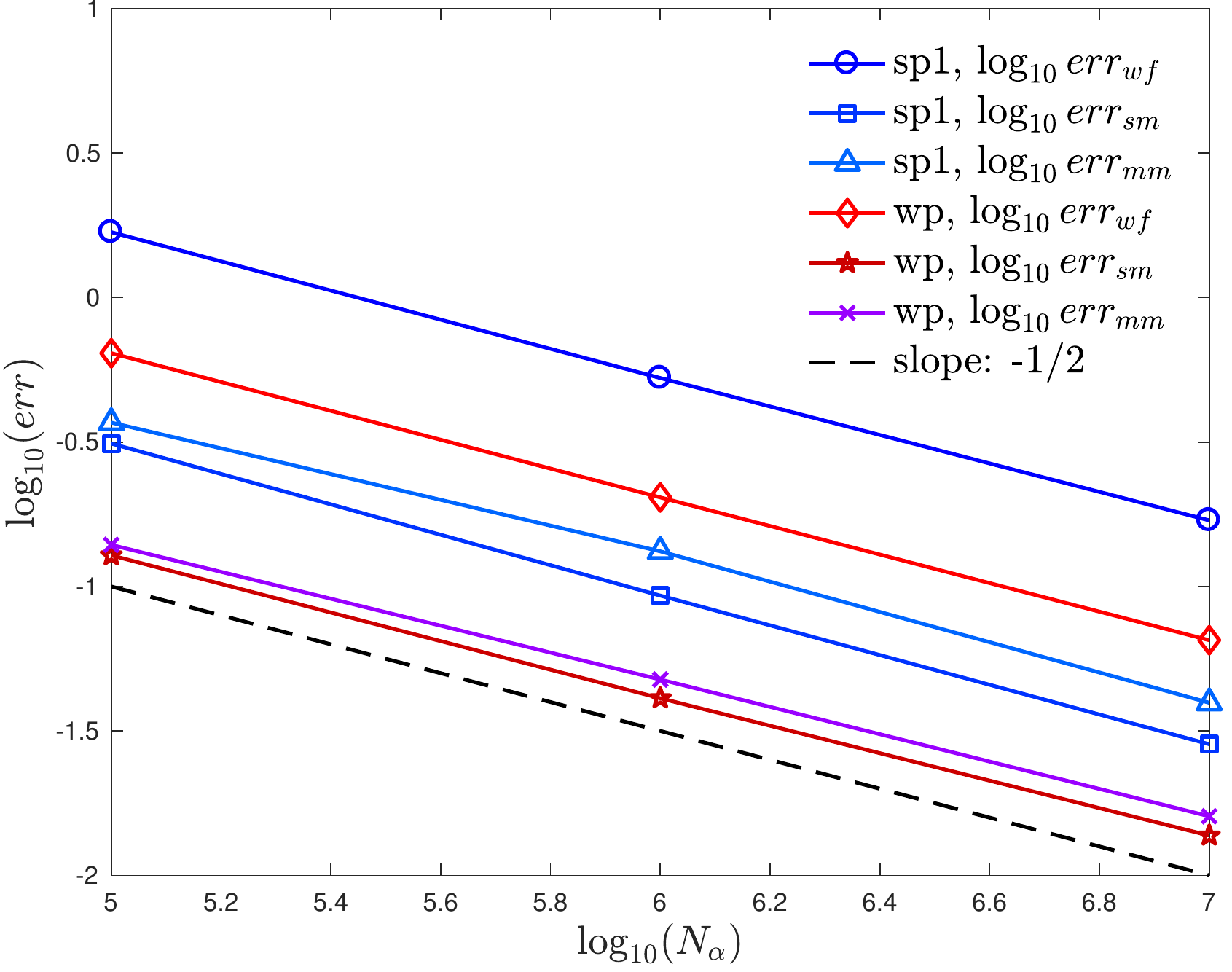}}
{\includegraphics[width=0.49\textwidth,height=0.27\textwidth]{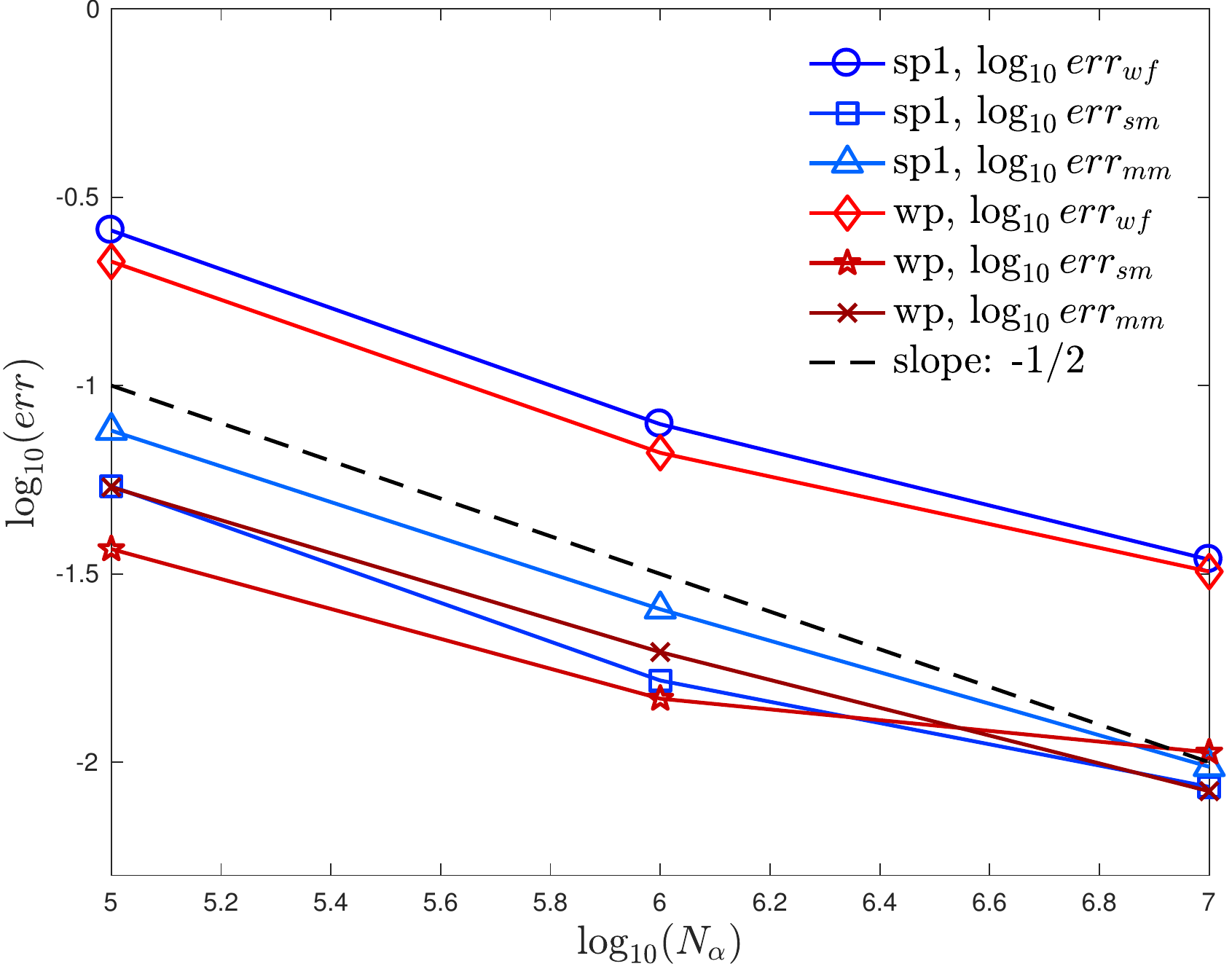}}}
\subfigure[$\textup{err}_{wf}$ under different $N_\alpha$.]{{\includegraphics[width=0.49\textwidth,height=0.27\textwidth]{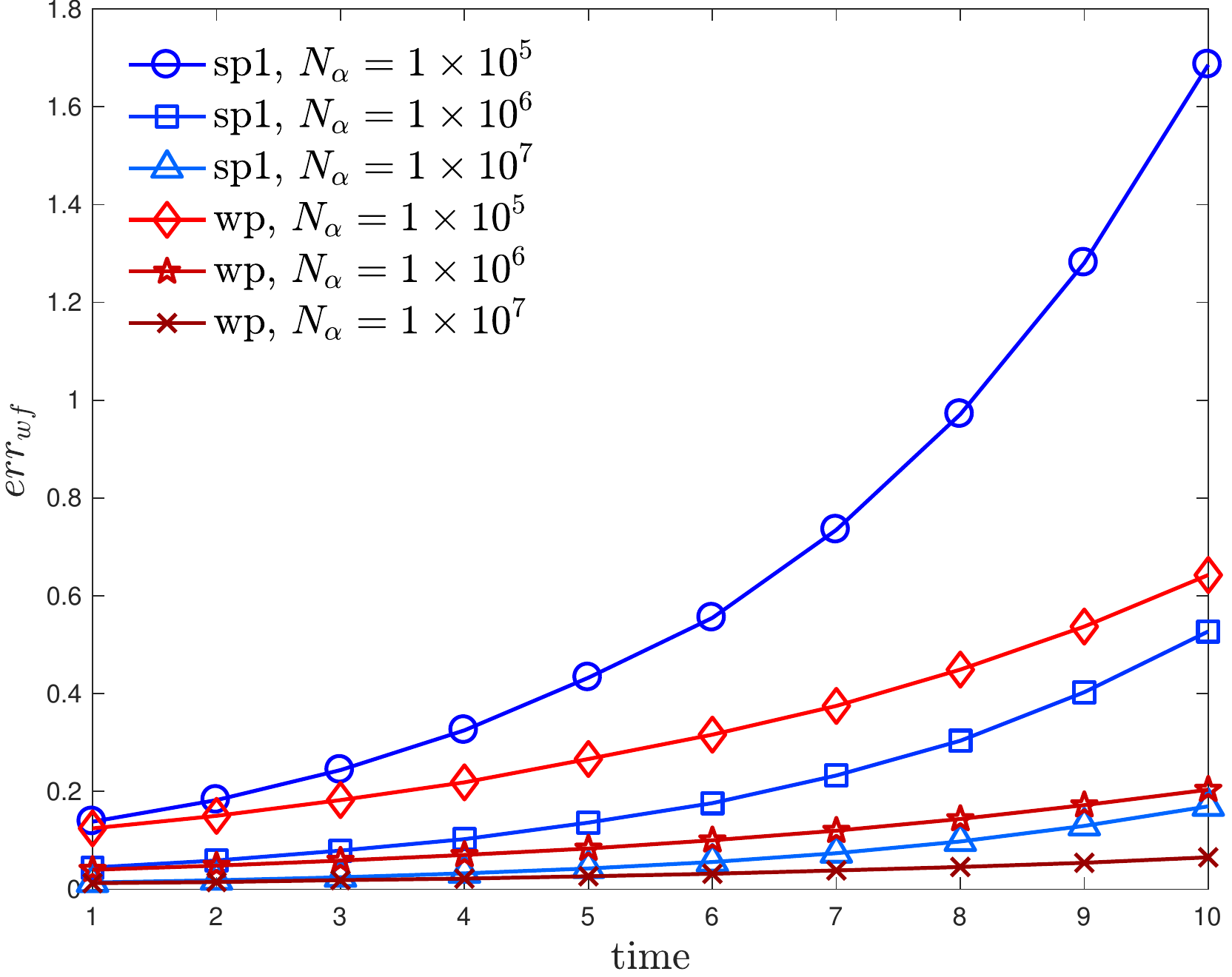}}
{\includegraphics[width=0.49\textwidth,height=0.27\textwidth]{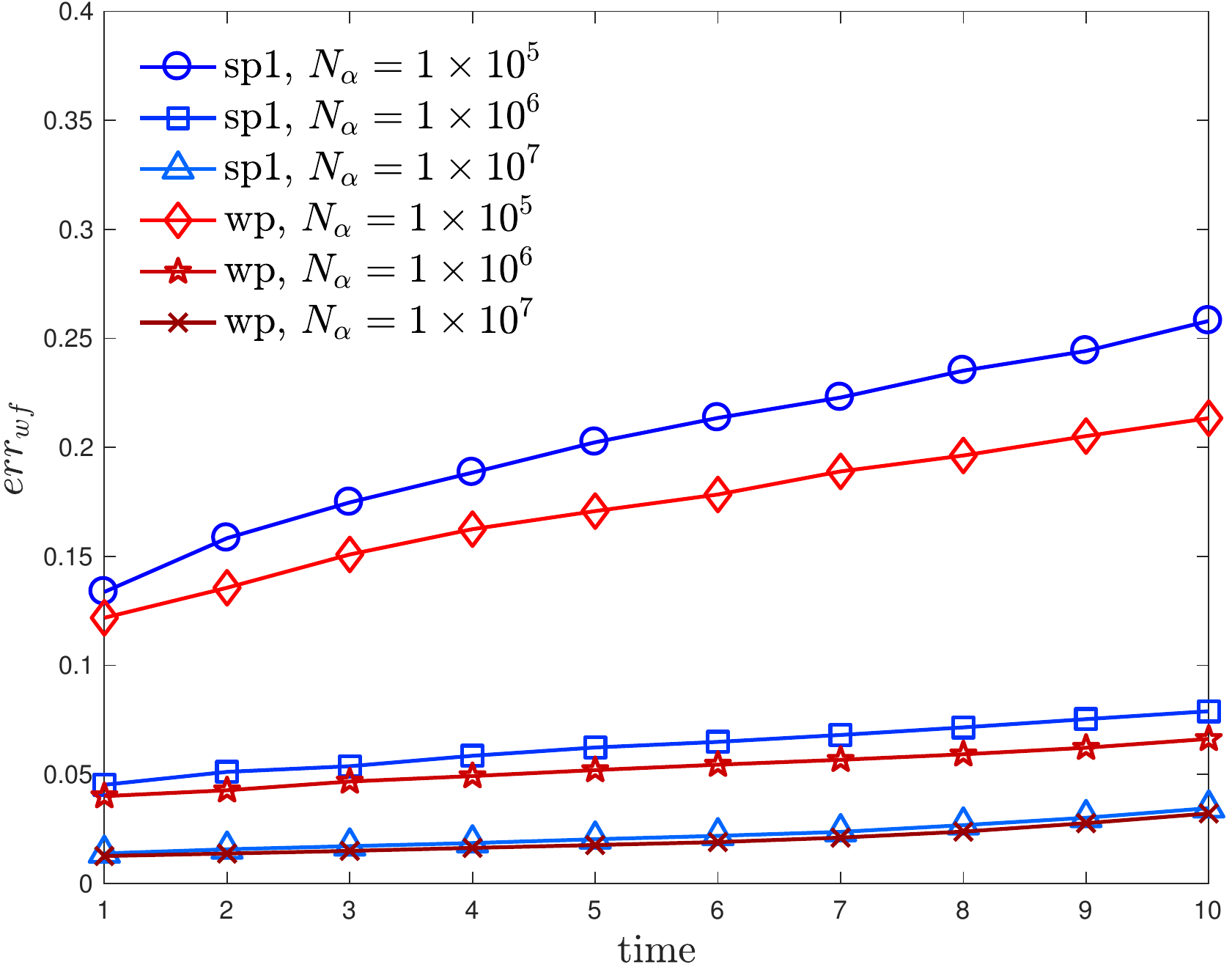}}}
\caption{\small A comparison between pure Monte Carlo simulations (left) and resampling-included simulations (right) for the 2D Gaussian scattering: Convergence rate with respect to $N_\alpha$ and time evolution of relative errors of the Wigner function. For the pure Monte Carlo, the convergence rates of both \textbf{sp1} and \textbf{wp} accord with the theoretical value $-1/2$. The resampling mechanisms effectively suppress the exponential growth of stochastic noises. Here we set $\gamma = 1.5\check{\xi}$ and $T_A = 1$fs.
}
\label{G_ss_rate}
\end{figure}

\begin{figure}[!h]
\subfigure[$\textup{err}_{sm}$ under different $N_\alpha$.]{{\includegraphics[width=0.49\textwidth,height=0.27\textwidth]{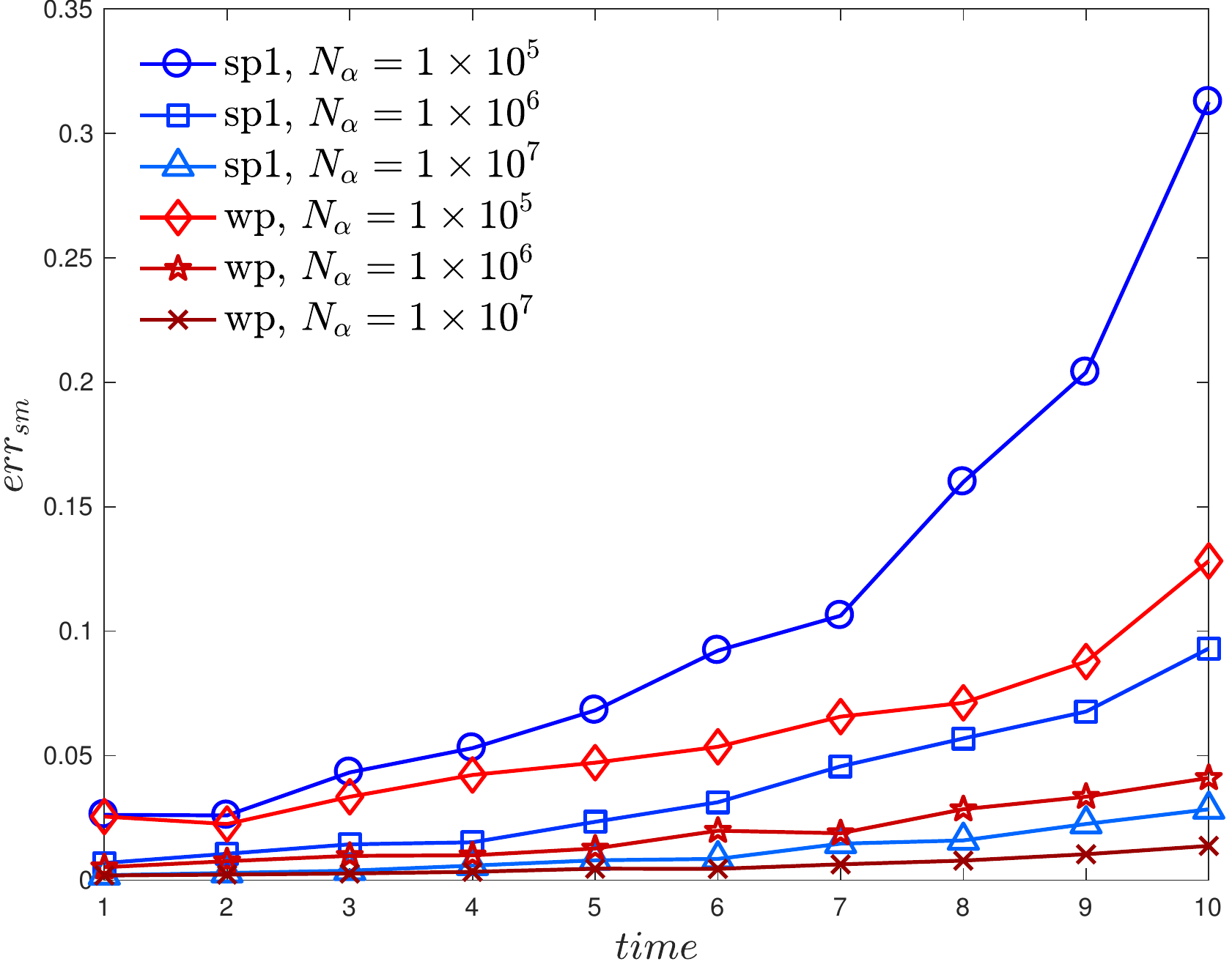}}
{\includegraphics[width=0.49\textwidth,height=0.27\textwidth]{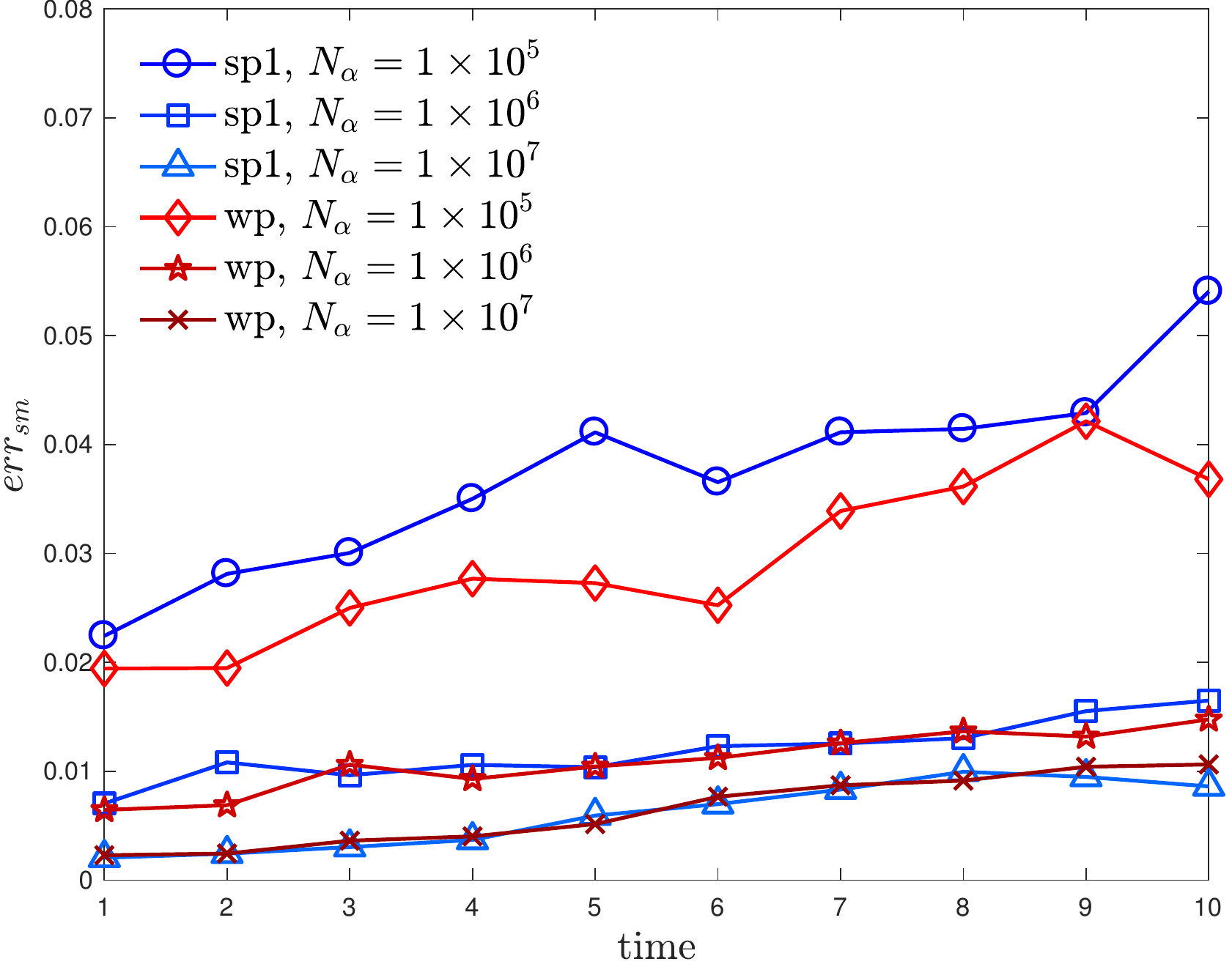}}}
\subfigure[$\textup{err}_{mm}$ under different $N_\alpha$.]{{\includegraphics[width=0.49\textwidth,height=0.27\textwidth]{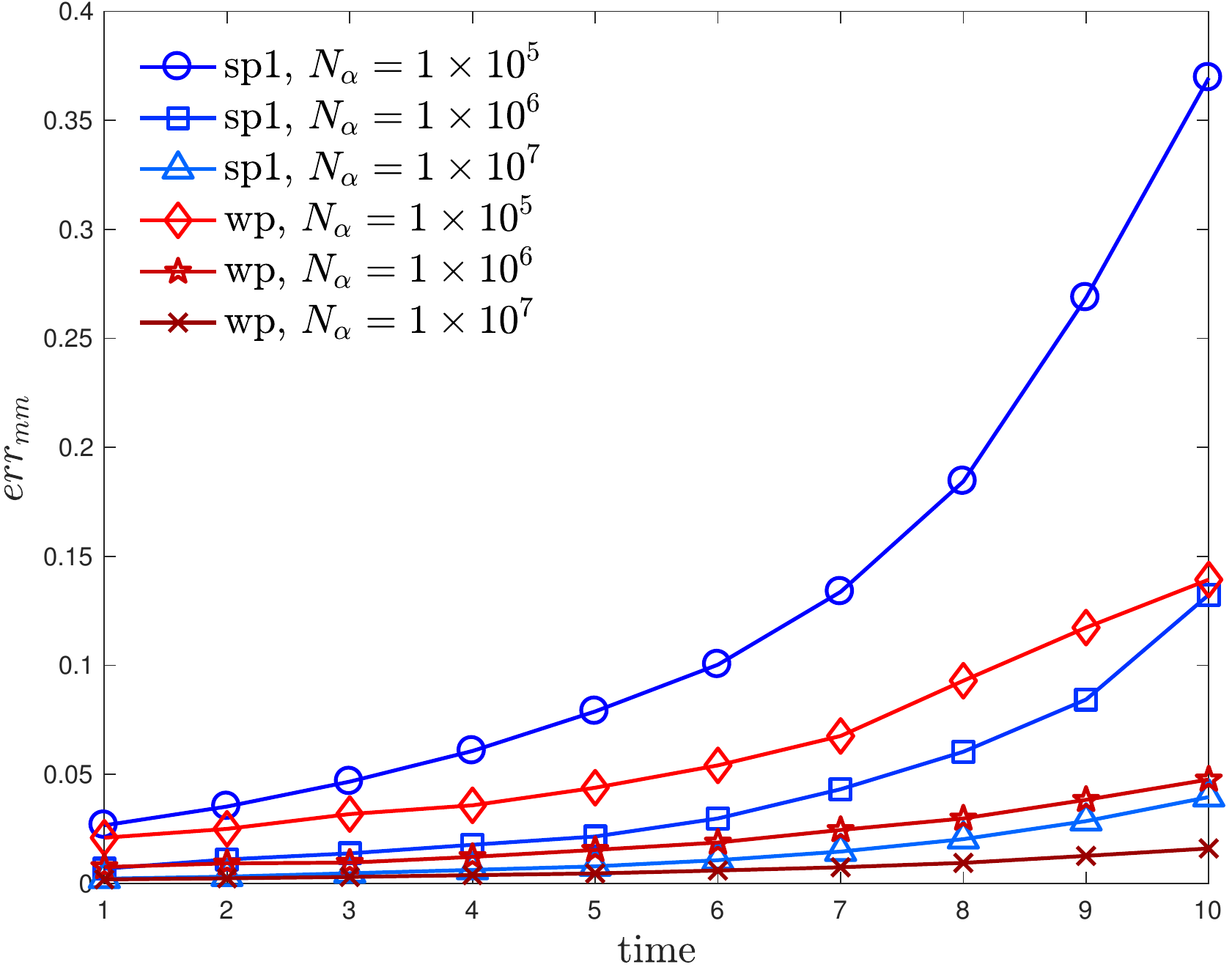}}
{\includegraphics[width=0.49\textwidth,height=0.27\textwidth]{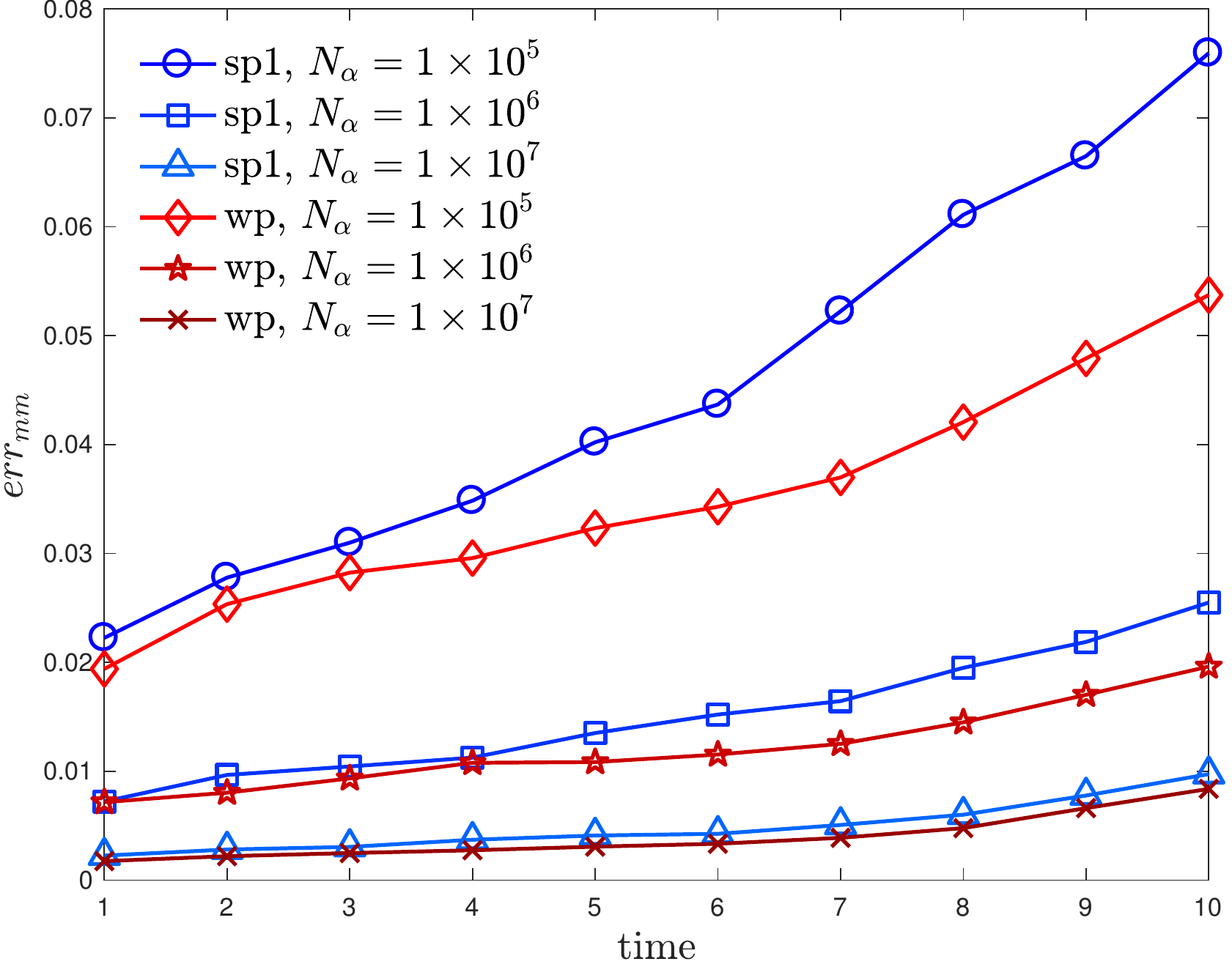}}}
\subfigure[Spatial marginal density at $t=10$fs.]{{\includegraphics[width=0.49\textwidth,height=0.27\textwidth]{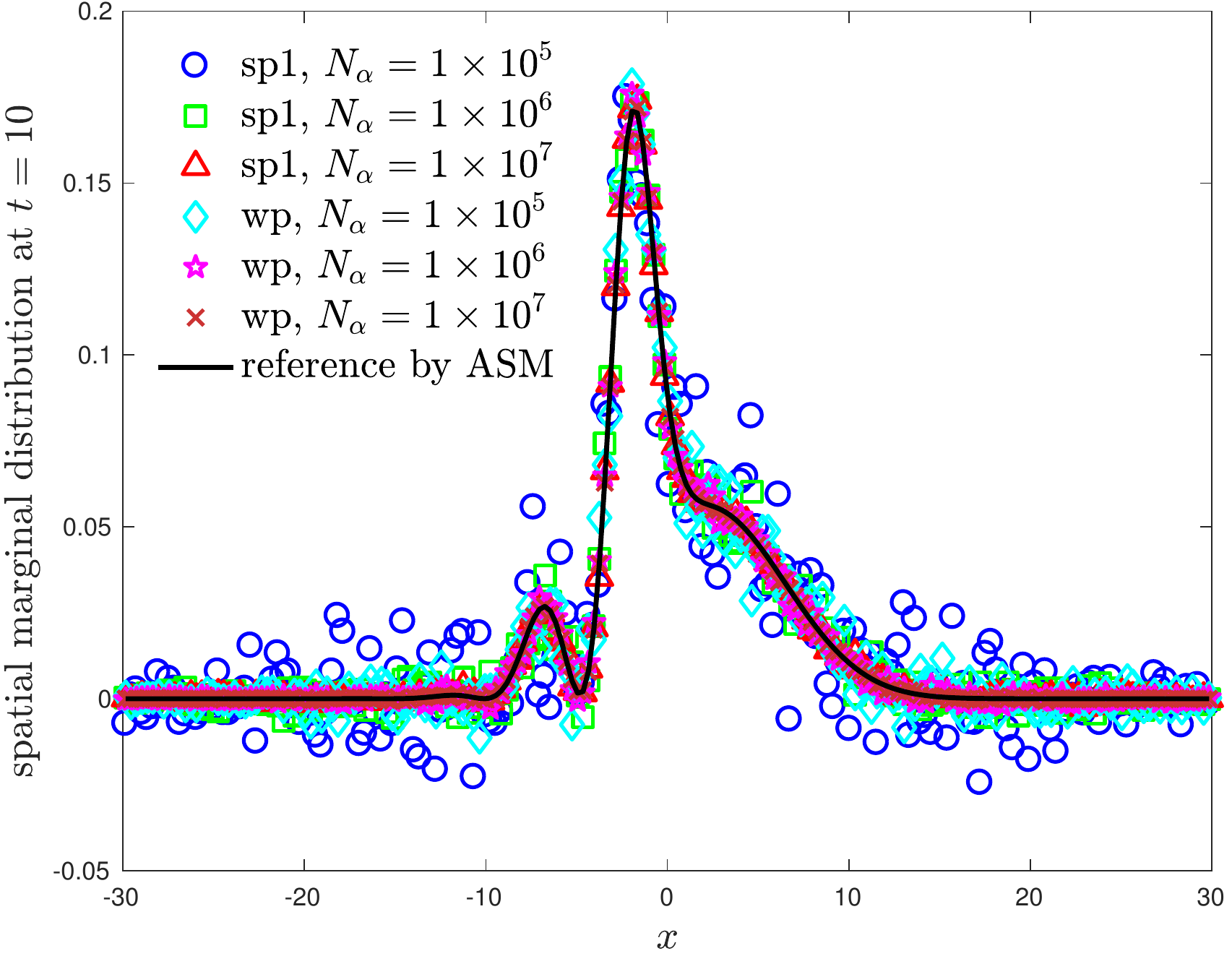}}
{\includegraphics[width=0.49\textwidth,height=0.27\textwidth]{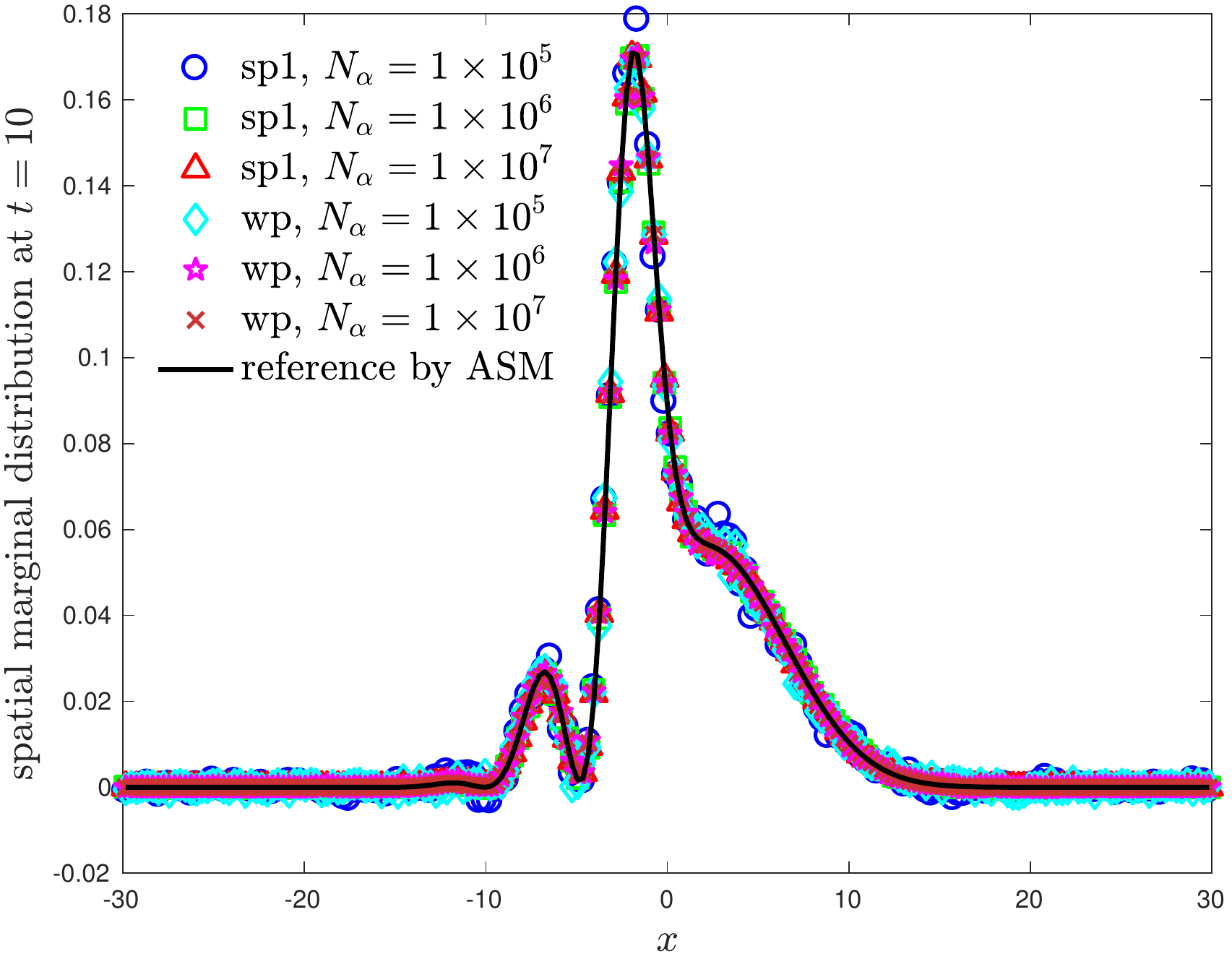}}}
\subfigure[Momental marginal density at $t=10$fs.]{{\includegraphics[width=0.49\textwidth,height=0.27\textwidth]{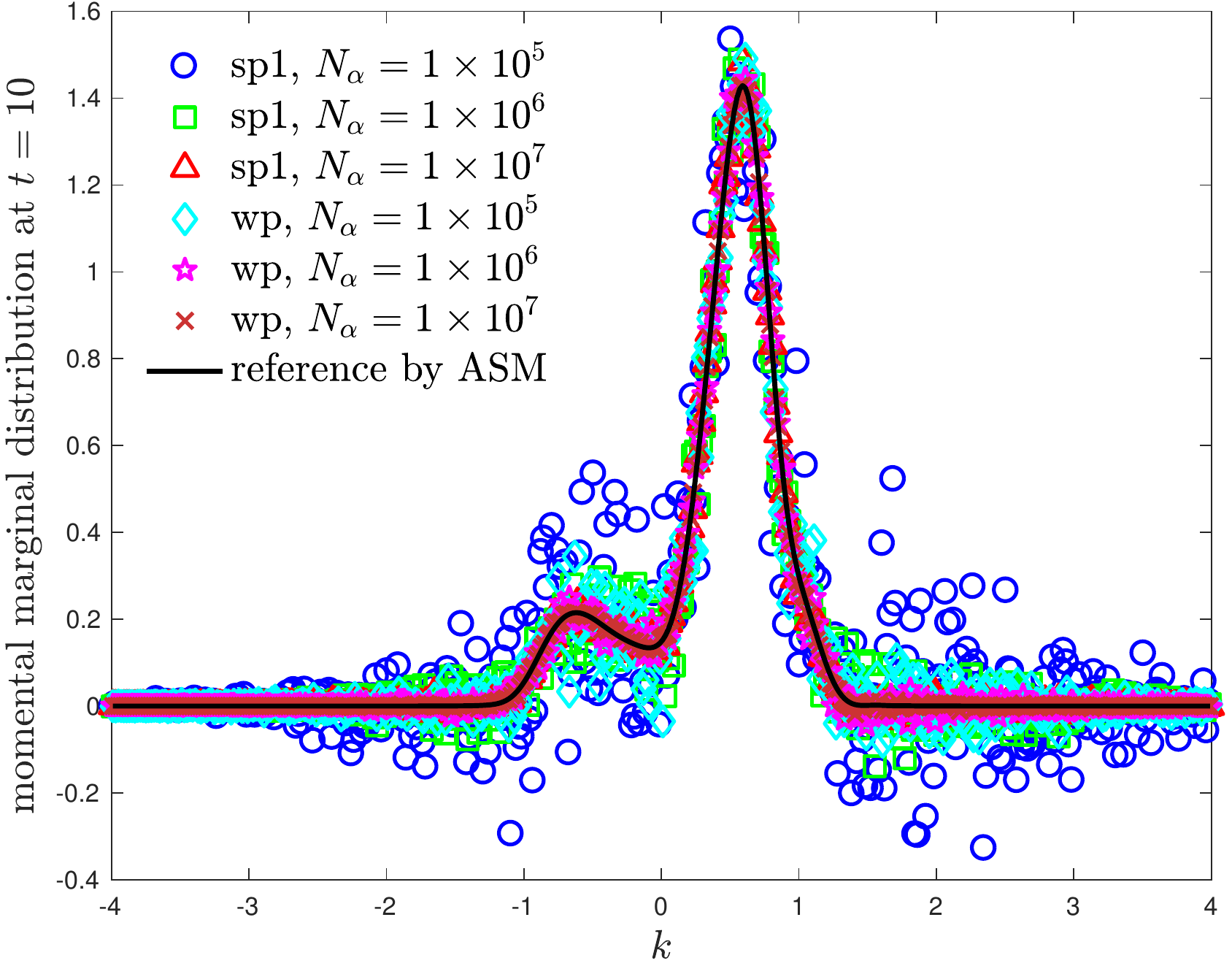}}
{\includegraphics[width=0.49\textwidth,height=0.27\textwidth]{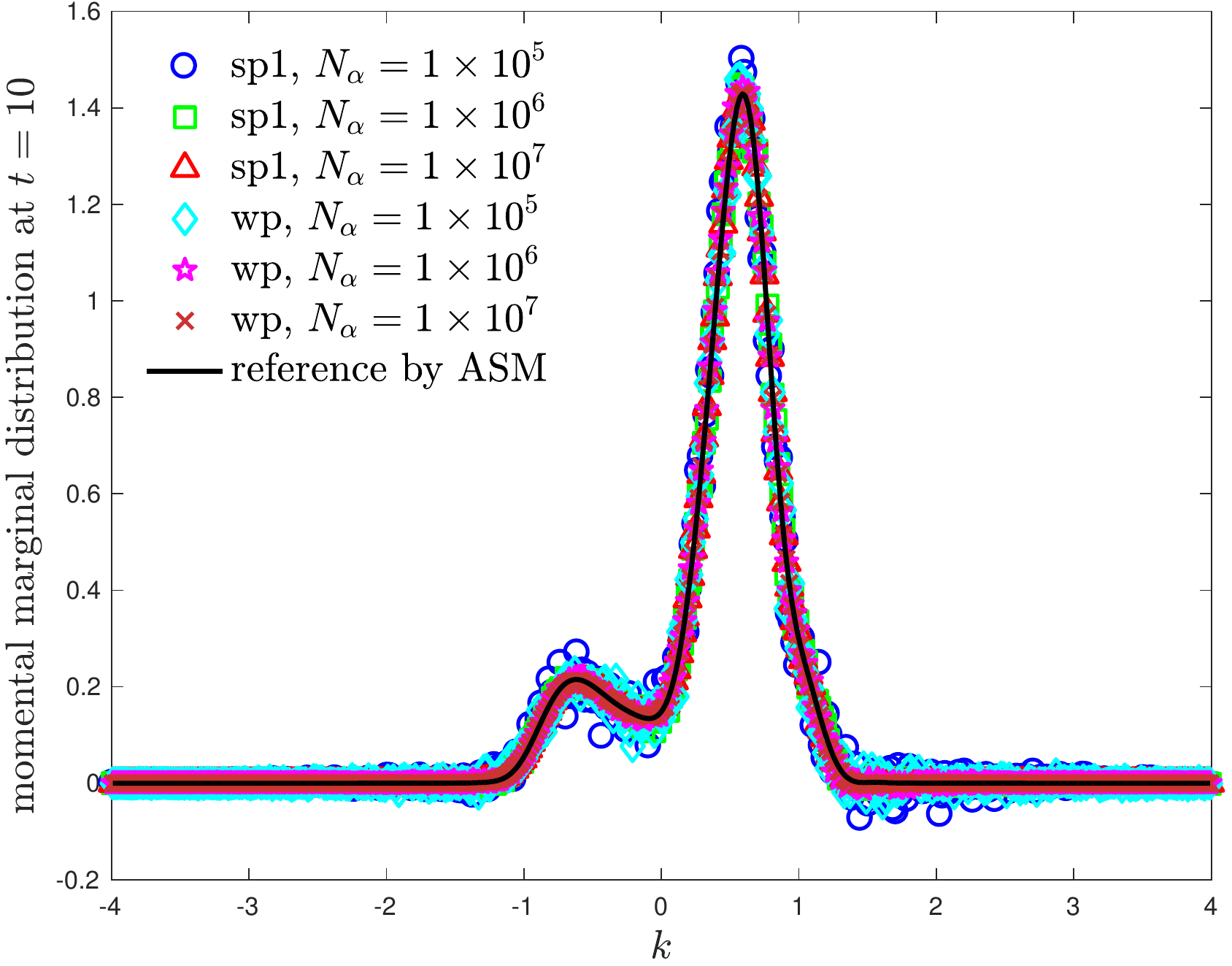}}}
\caption{\small 	A comparison between pure MC simulations (left) and resampling-included simulations (right) for the 2D Gaussian scattering: Time evolution of relative errors and both spatial and momental marginal distributions at $t=10$fs are plotted. Stochastic noises of \textbf{sp1} are larger than those of \textbf{wp}. Nevertheless, they can be significantly suppressed by increasing $N_\alpha$. Here we set $\gamma = 1.5\check{\xi}$ and $T_A = 1$fs.
}
\label{G_ss_error}
\end{figure}

\begin{figure}[!h]
\subfigure[Convergence rate with respect to $N_\alpha$.]{\includegraphics[width=0.49\textwidth,height=0.27\textwidth]{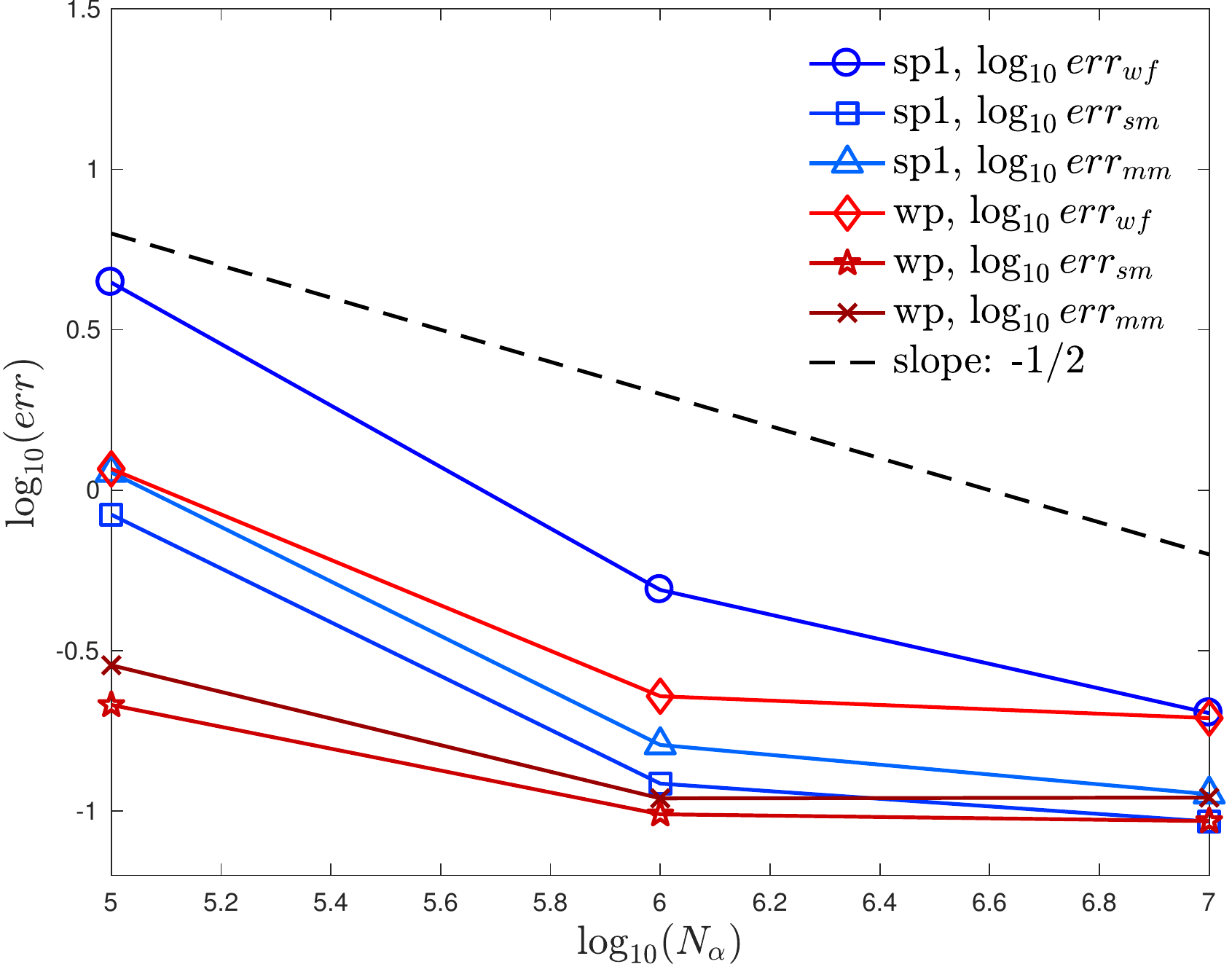}}
\subfigure[$\textup{err}_{wf}$ under different $N_\alpha$.]{\includegraphics[width=0.49\textwidth,height=0.27\textwidth]{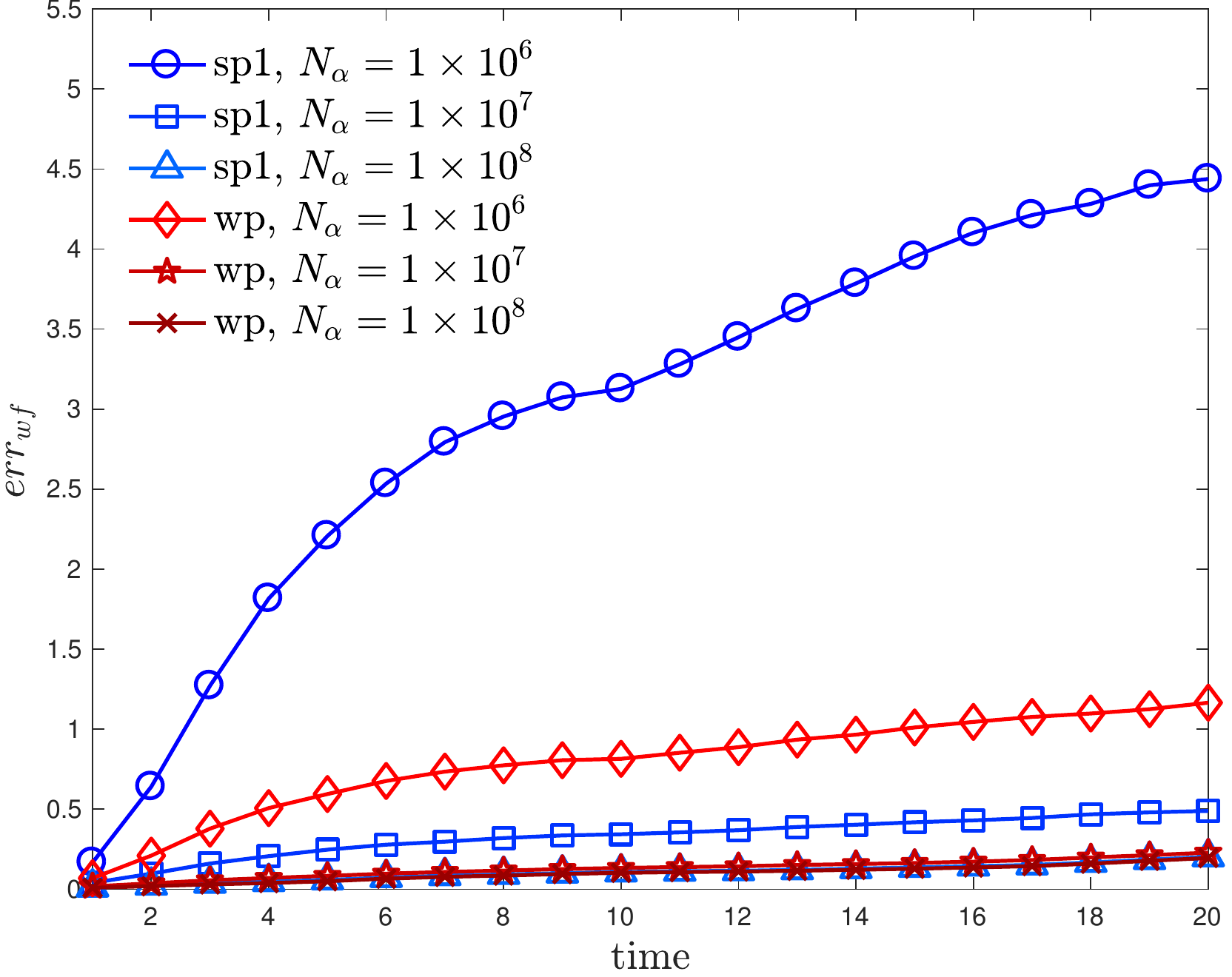}}
\subfigure[$\textup{err}_{sm}$ under different $N_\alpha$.]{\includegraphics[width=0.49\textwidth,height=0.27\textwidth]{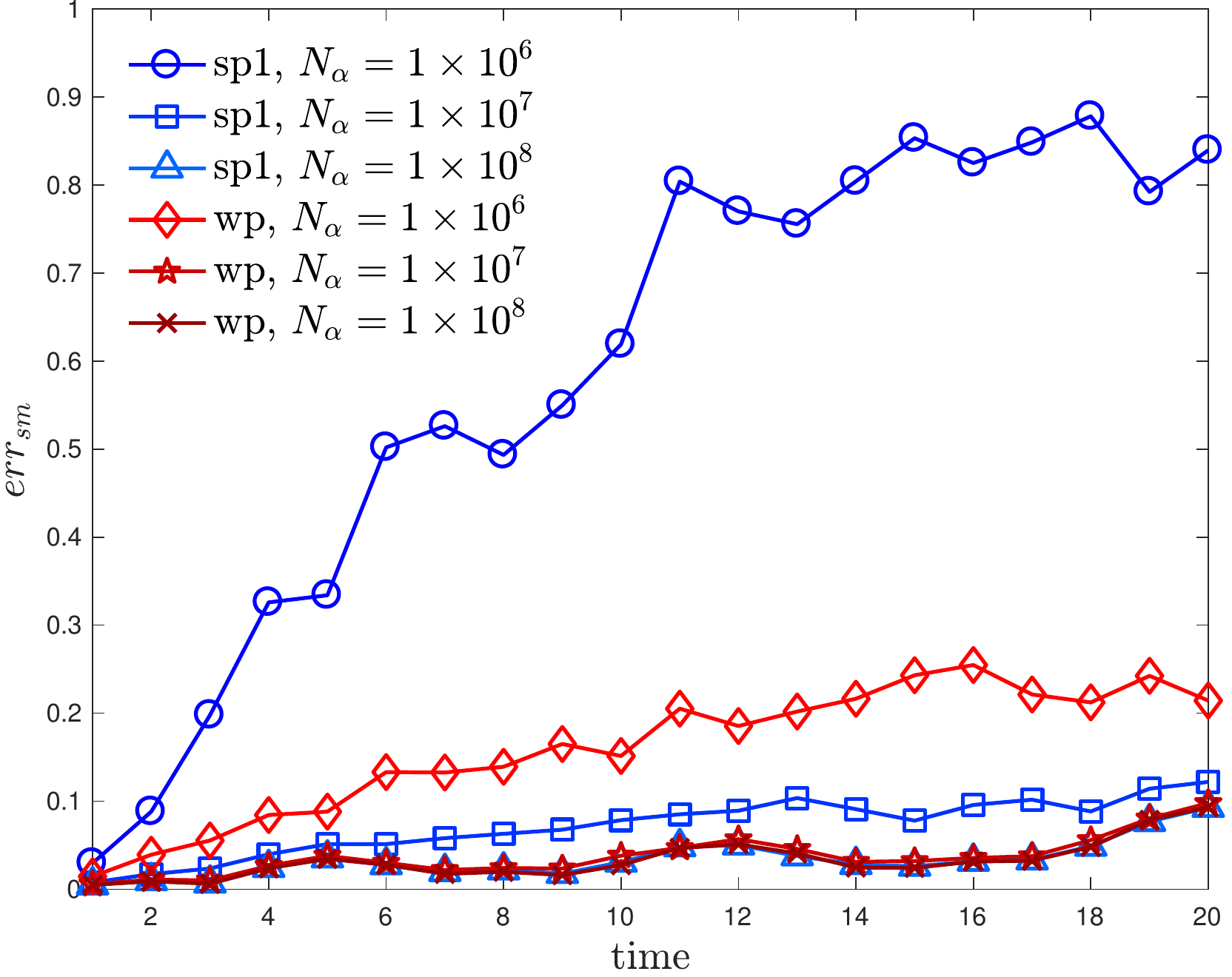}}
\subfigure[$\textup{err}_{mm}$ under different $N_\alpha$.]{\includegraphics[width=0.49\textwidth,height=0.27\textwidth]{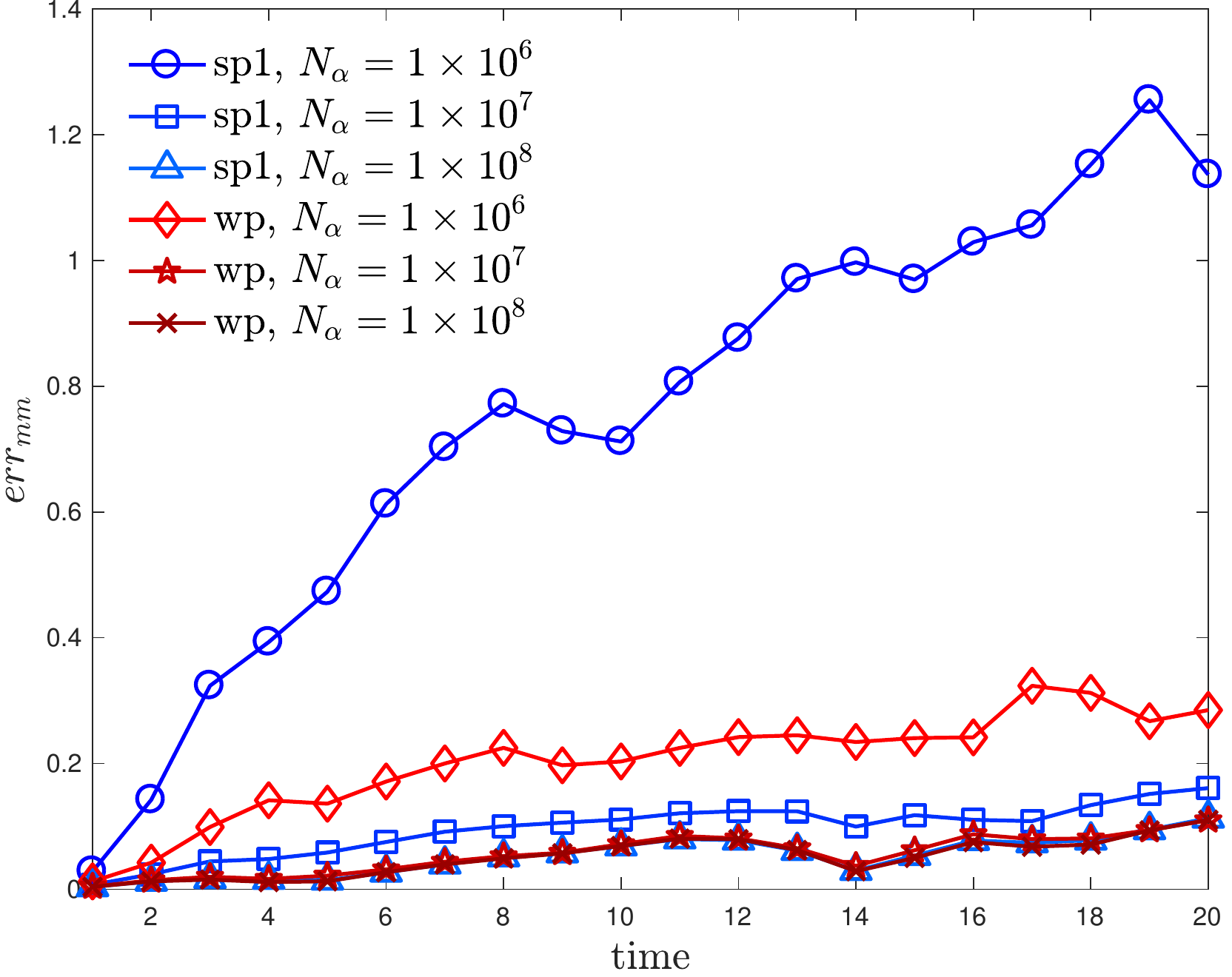}}
\caption{\small  The 4D Helium-like system: Convergence rate and relative errors under different sample sizes $N_\alpha$. The convergence rate slightly deviates from the theoretical prediction due to the error of the resampling. Stochastic noises of \textbf{sp1} are larger than those of \textbf{wp}, but accuracy can be still be improved by increasing $N_\alpha$. Here we set $\gamma = 2$ and $T_A = 2$a.u.
}
\label{He_ss_rate}
\end{figure}

\begin{figure}[!h]
\subfigure[$t=5$a.u.]{{\includegraphics[width=0.49\textwidth,height=0.27\textwidth]{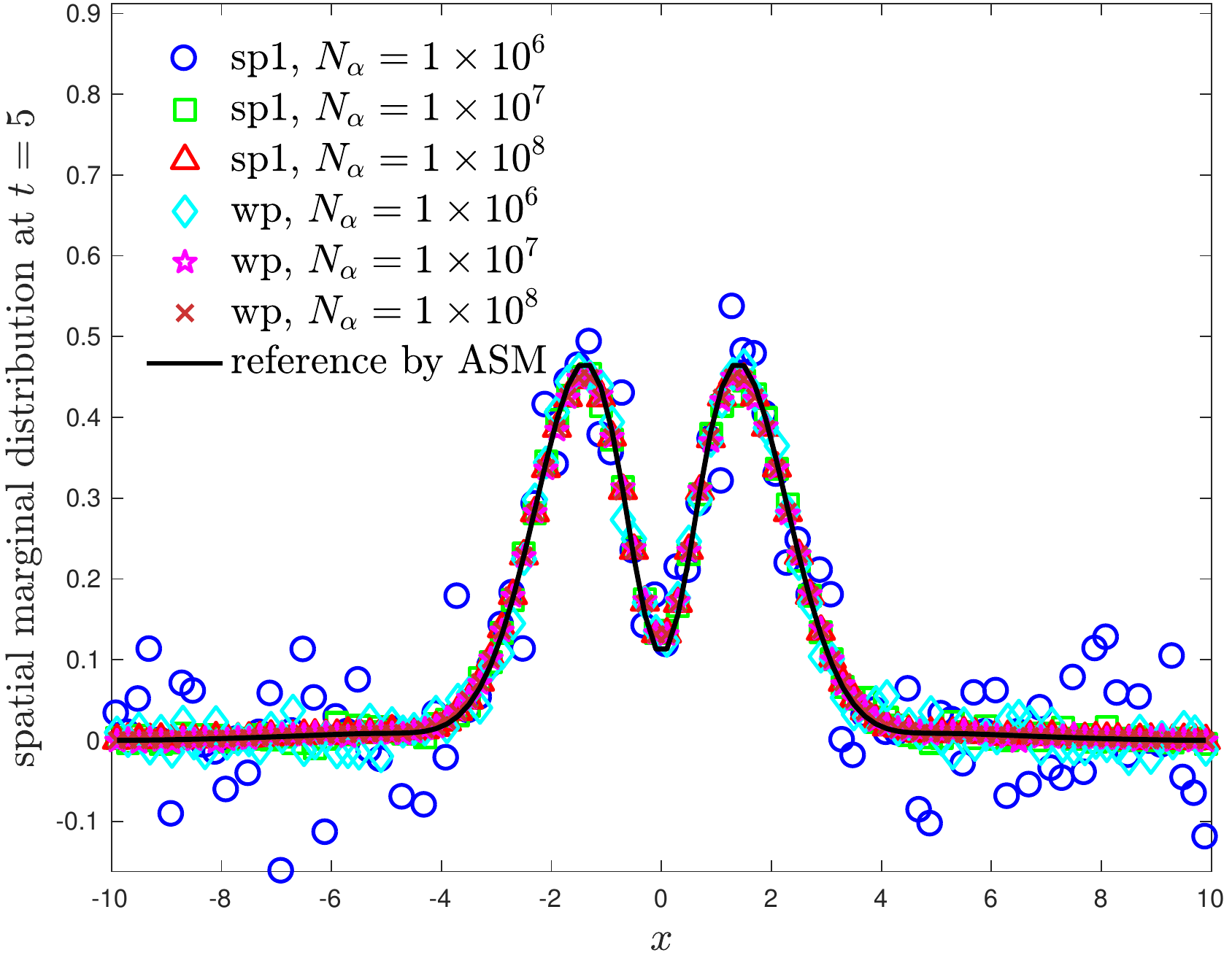}}
{\includegraphics[width=0.49\textwidth,height=0.27\textwidth]{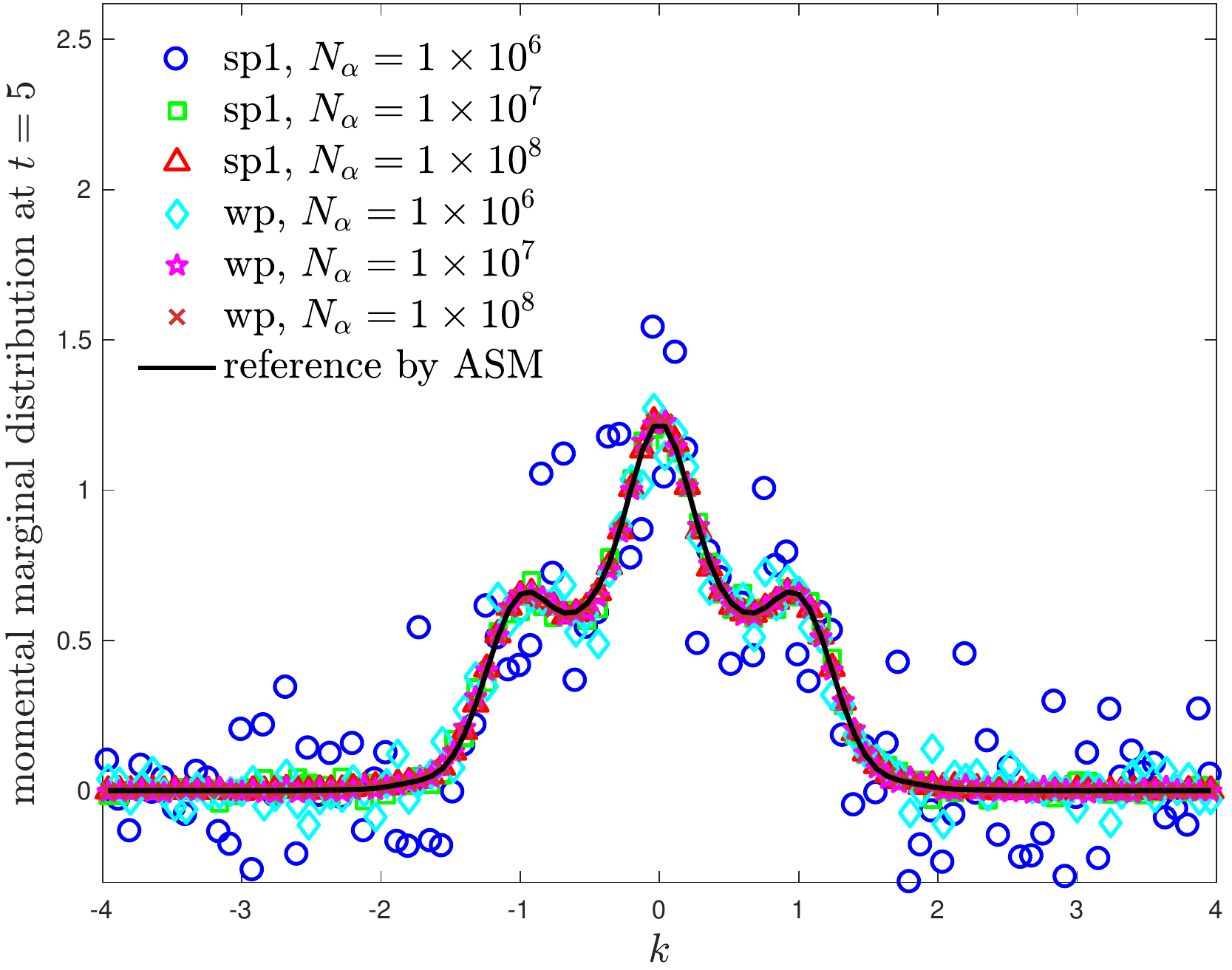}}}
\subfigure[$t=10$a.u.]{{\includegraphics[width=0.49\textwidth,height=0.27\textwidth]{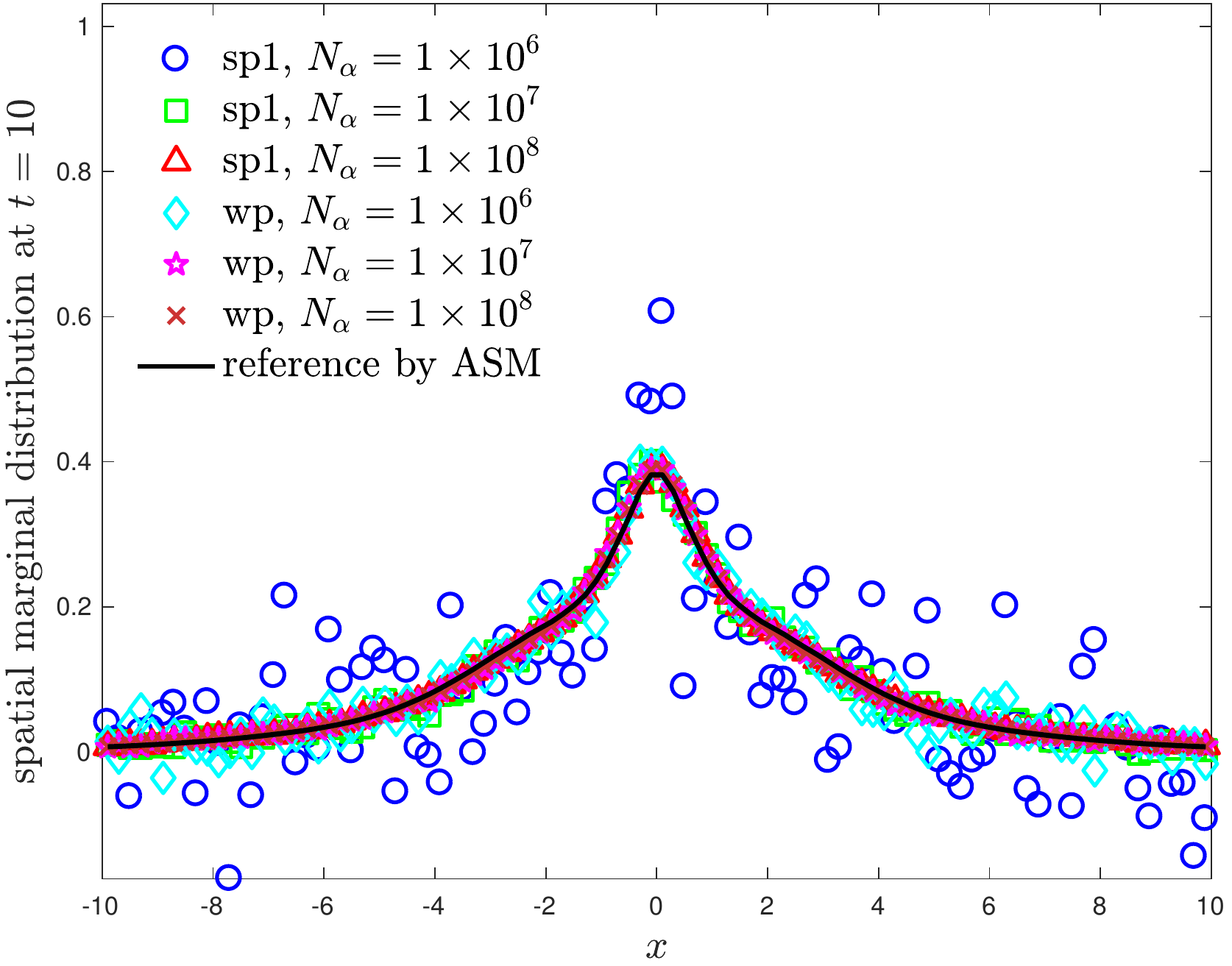}}
{\includegraphics[width=0.49\textwidth,height=0.27\textwidth]{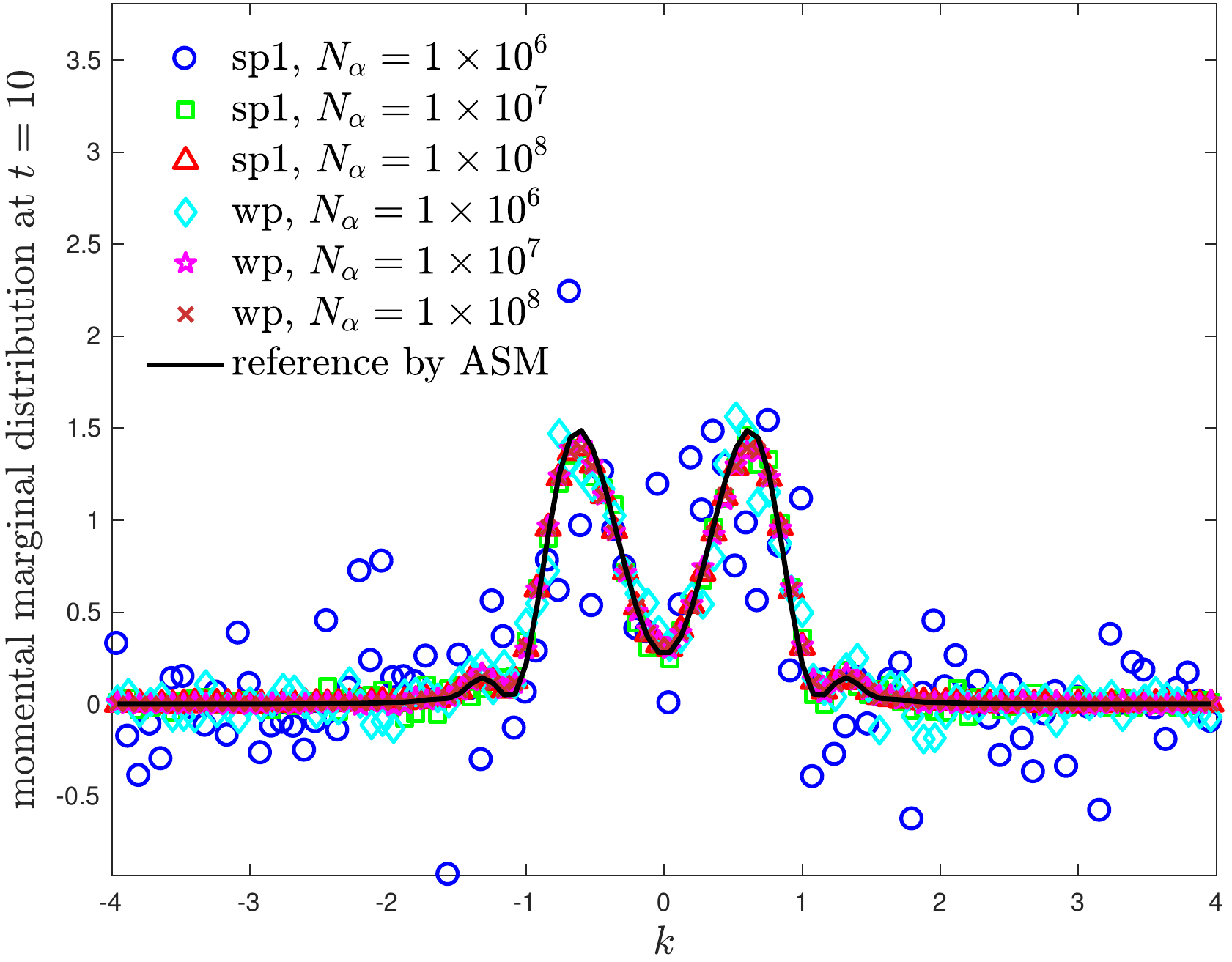}}}
\subfigure[$t=15$a.u.]{{\includegraphics[width=0.49\textwidth,height=0.27\textwidth]{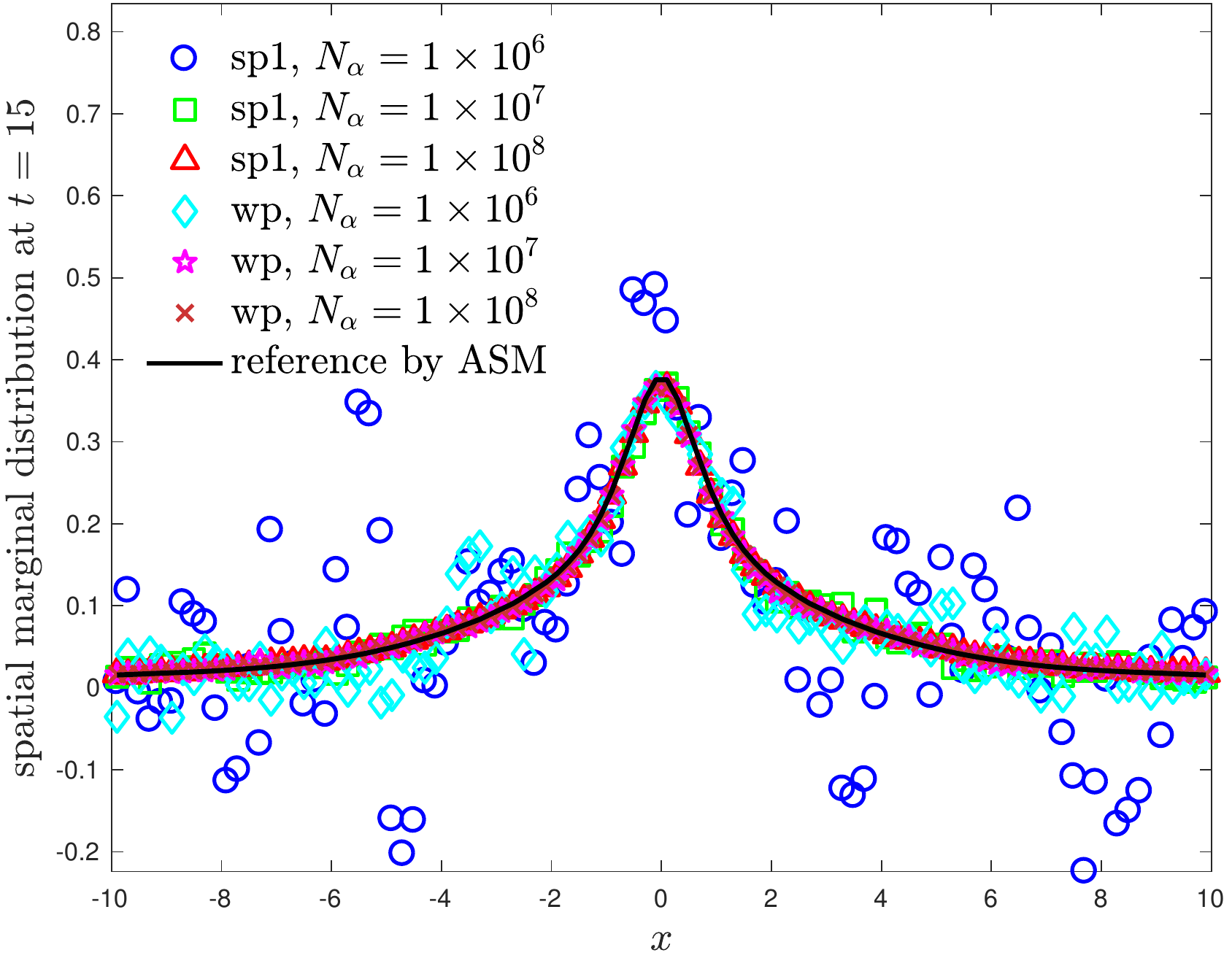}}
{\includegraphics[width=0.49\textwidth,height=0.27\textwidth]{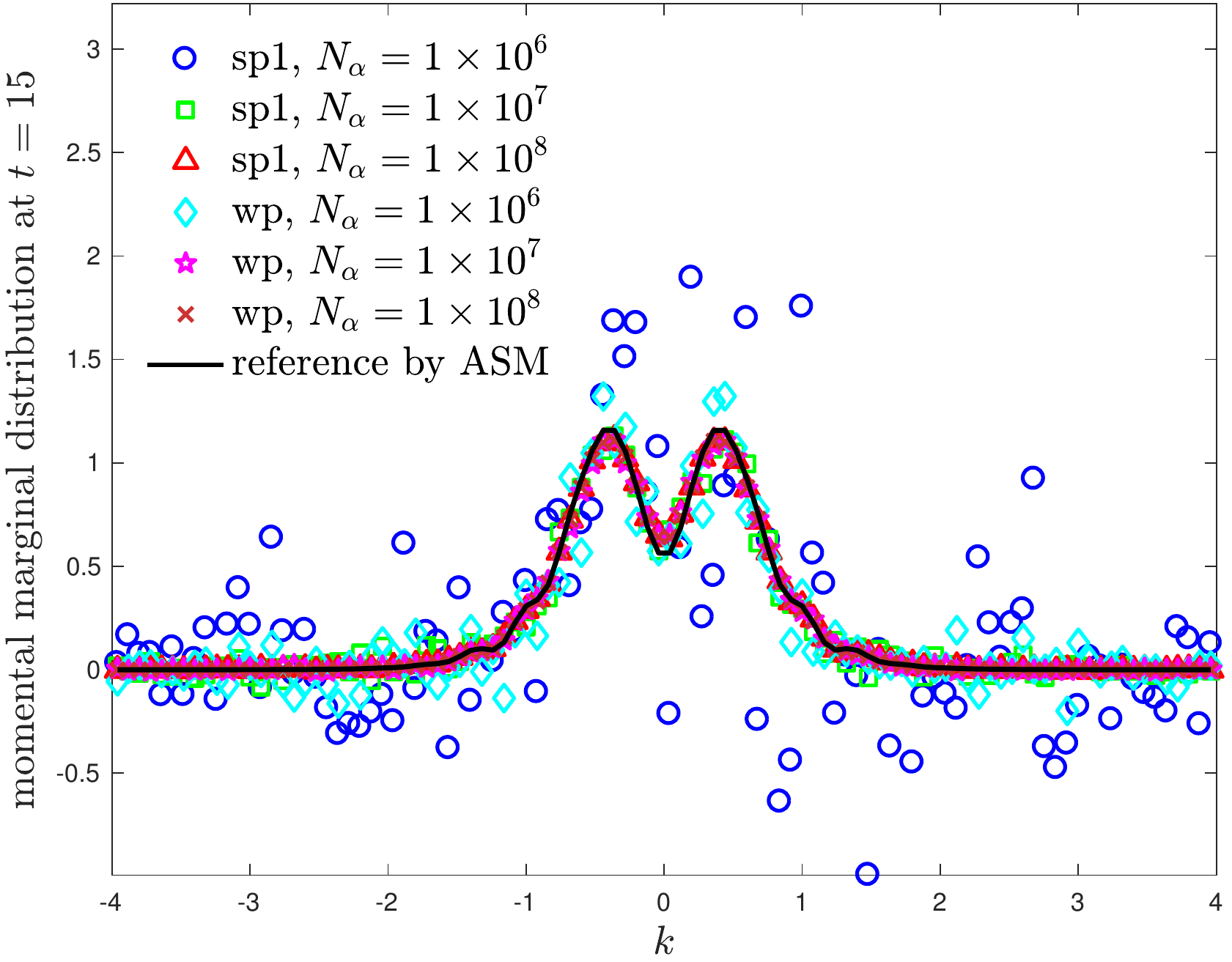}}}
\subfigure[$t=20$a.u.]{{\includegraphics[width=0.49\textwidth,height=0.27\textwidth]{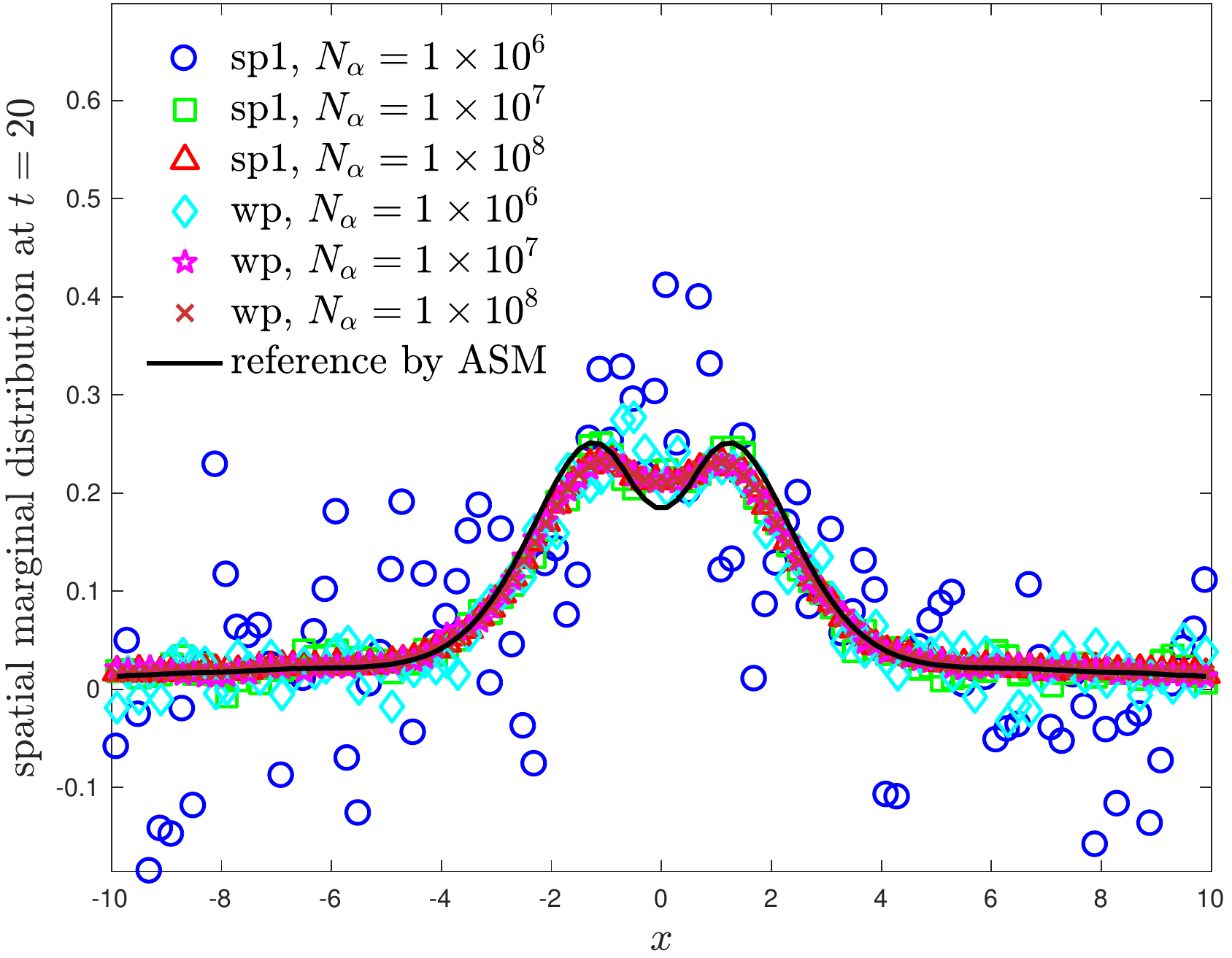}}
{\includegraphics[width=0.49\textwidth,height=0.27\textwidth]{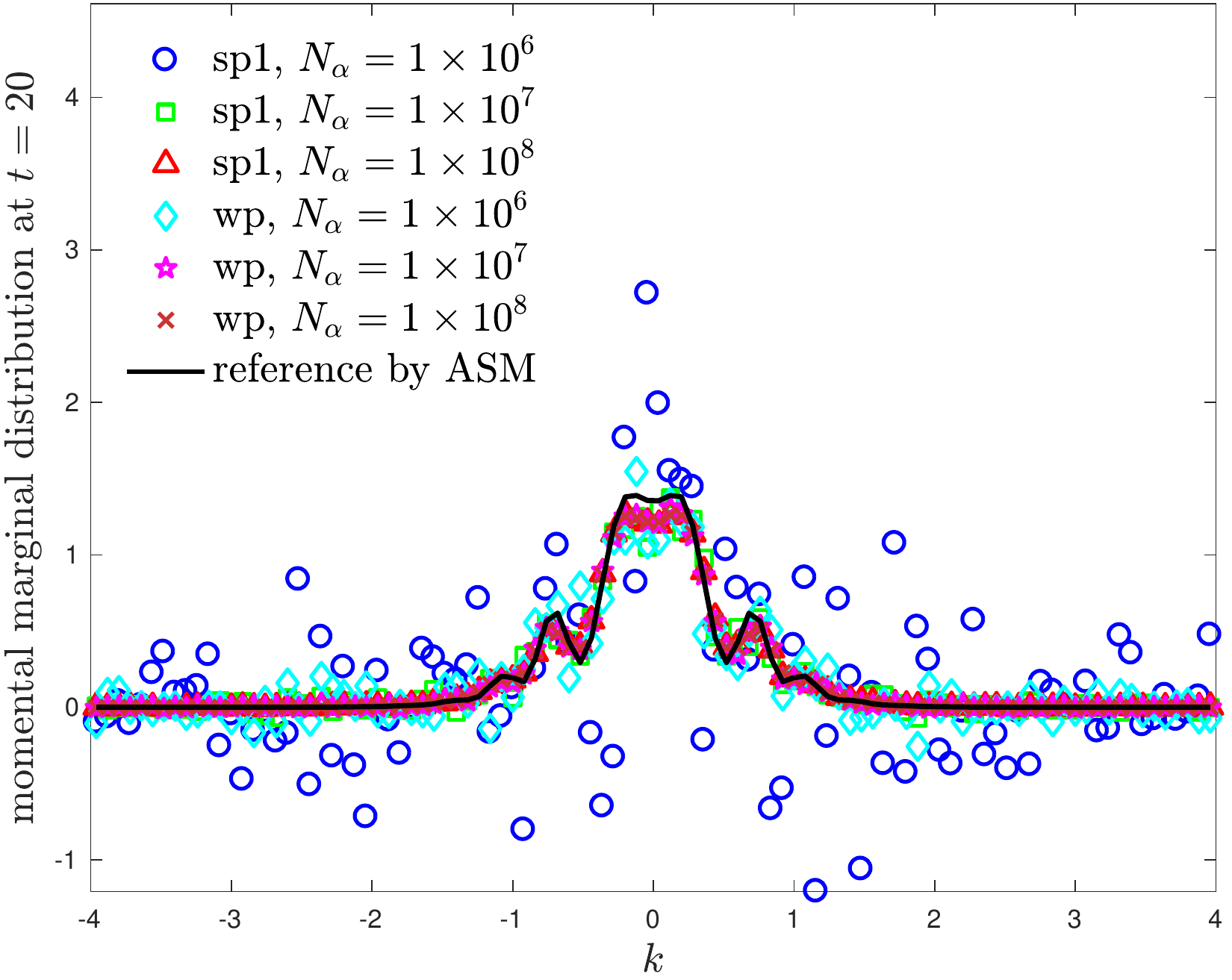}}}
\caption{\small The 4D Helium-like system: Plots of spatial and momental marginal distributions at different instants 5, 10, 15, 20a.u. Too small sample size undermines the accuracy and noises even overshadow the true solutions. For large $N_\alpha$, both \textbf{sp1} and $\textbf{wp}$ perform quite well. Here we set $\gamma = 2$ and $T_A = 0.5$a.u.
}
\label{He_ss_error}
\end{figure}

\begin{figure}[!h]
\subfigure[Reference by ASM.]{\includegraphics[width=0.49\textwidth,height=0.27\textwidth]{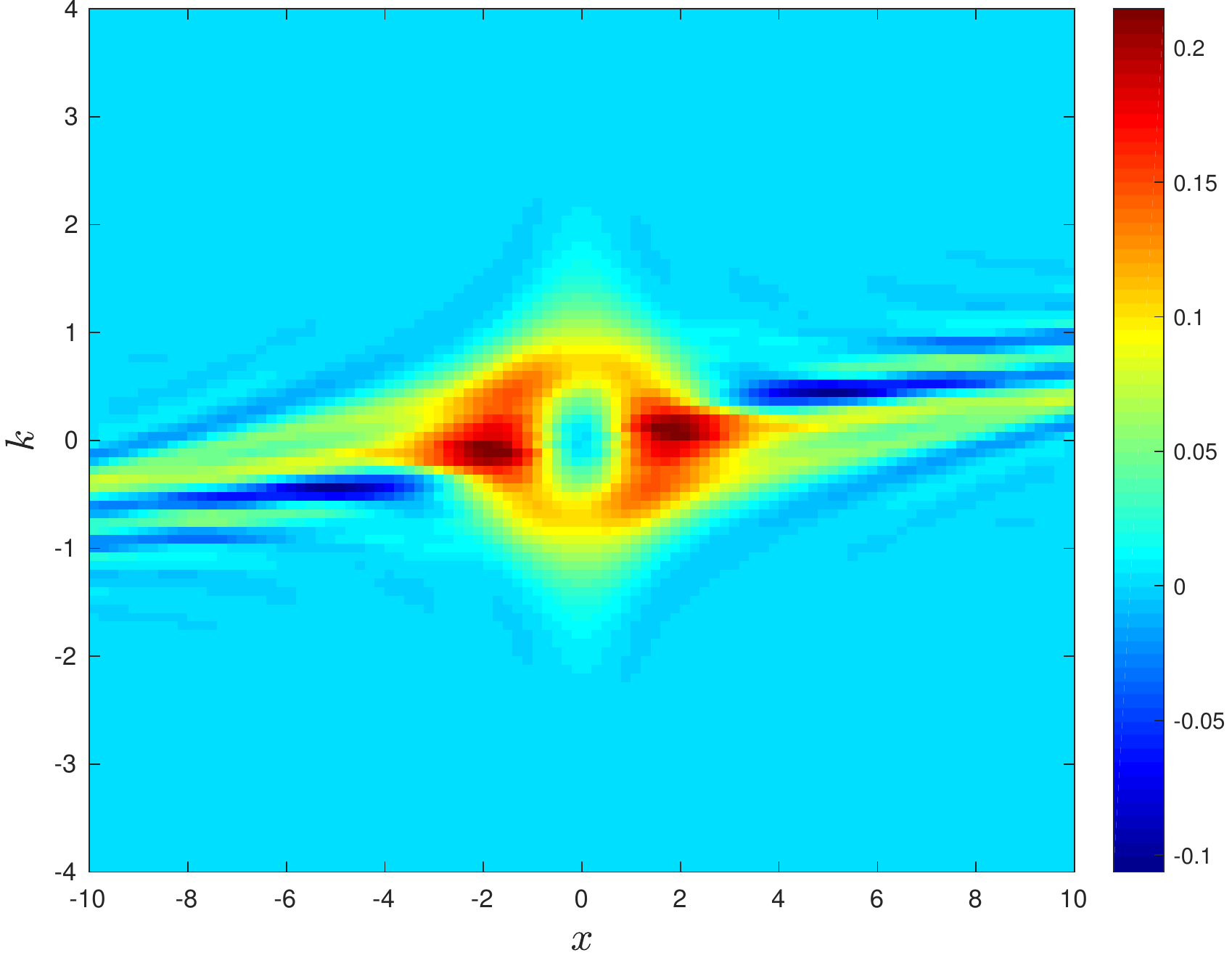}}
\subfigure[$N_\alpha = 1 \times 10^6$.]{\includegraphics[width=0.49\textwidth,height=0.27\textwidth]{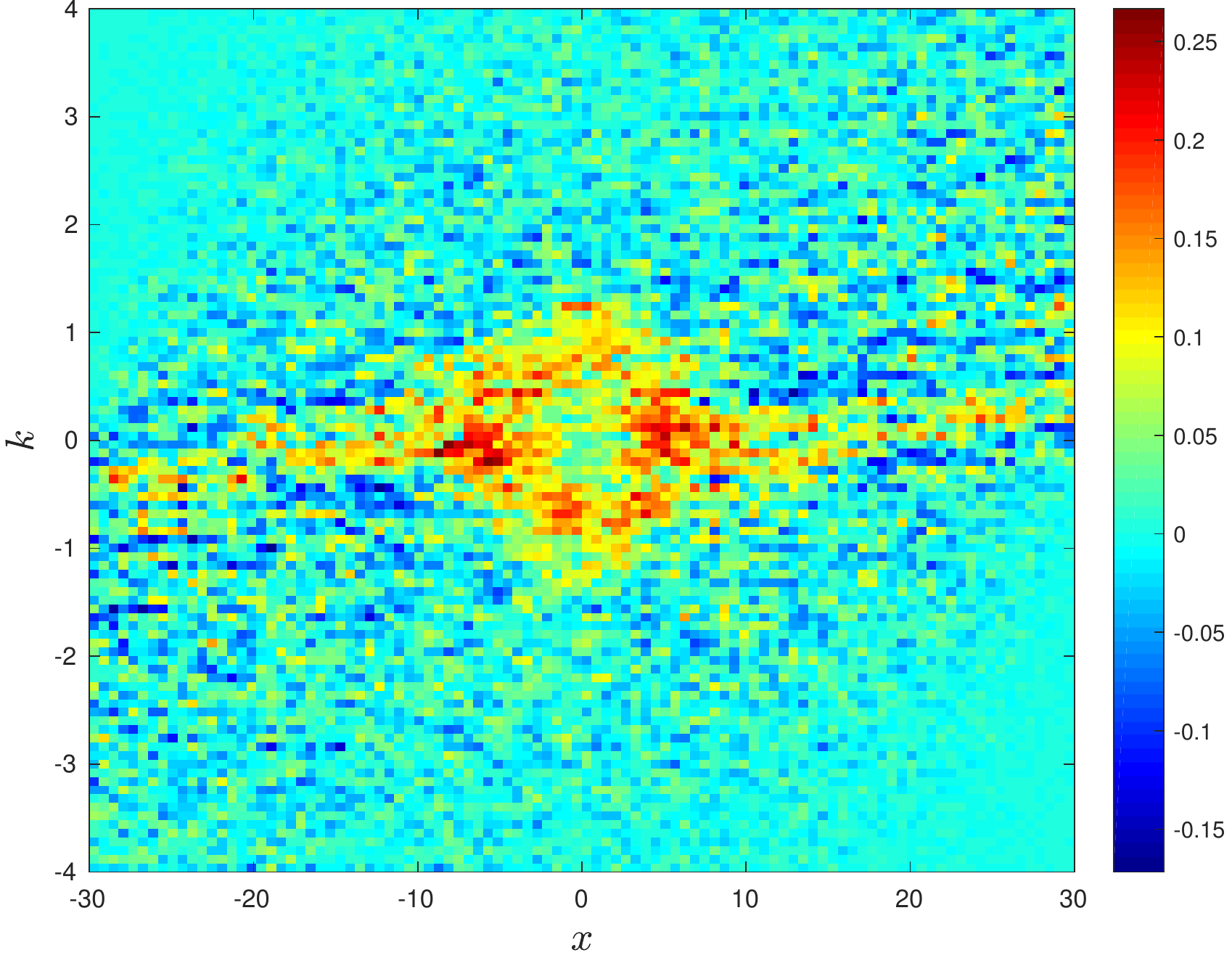}}
\subfigure[$N_\alpha = 1 \times 10^7$.]{\includegraphics[width=0.49\textwidth,height=0.27\textwidth]{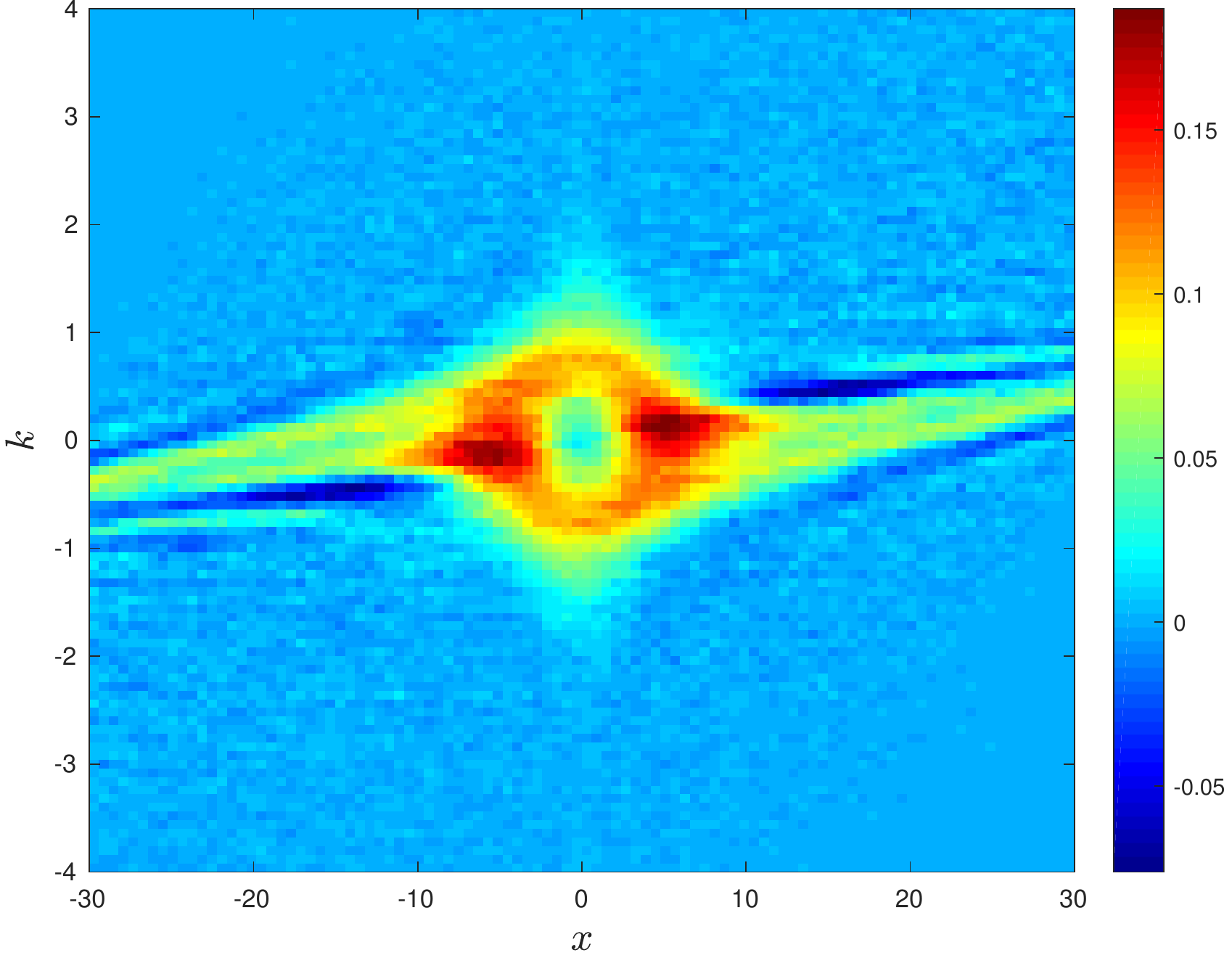}}
\subfigure[$N_\alpha = 1 \times 10^8$.]{\includegraphics[width=0.49\textwidth,height=0.27\textwidth]{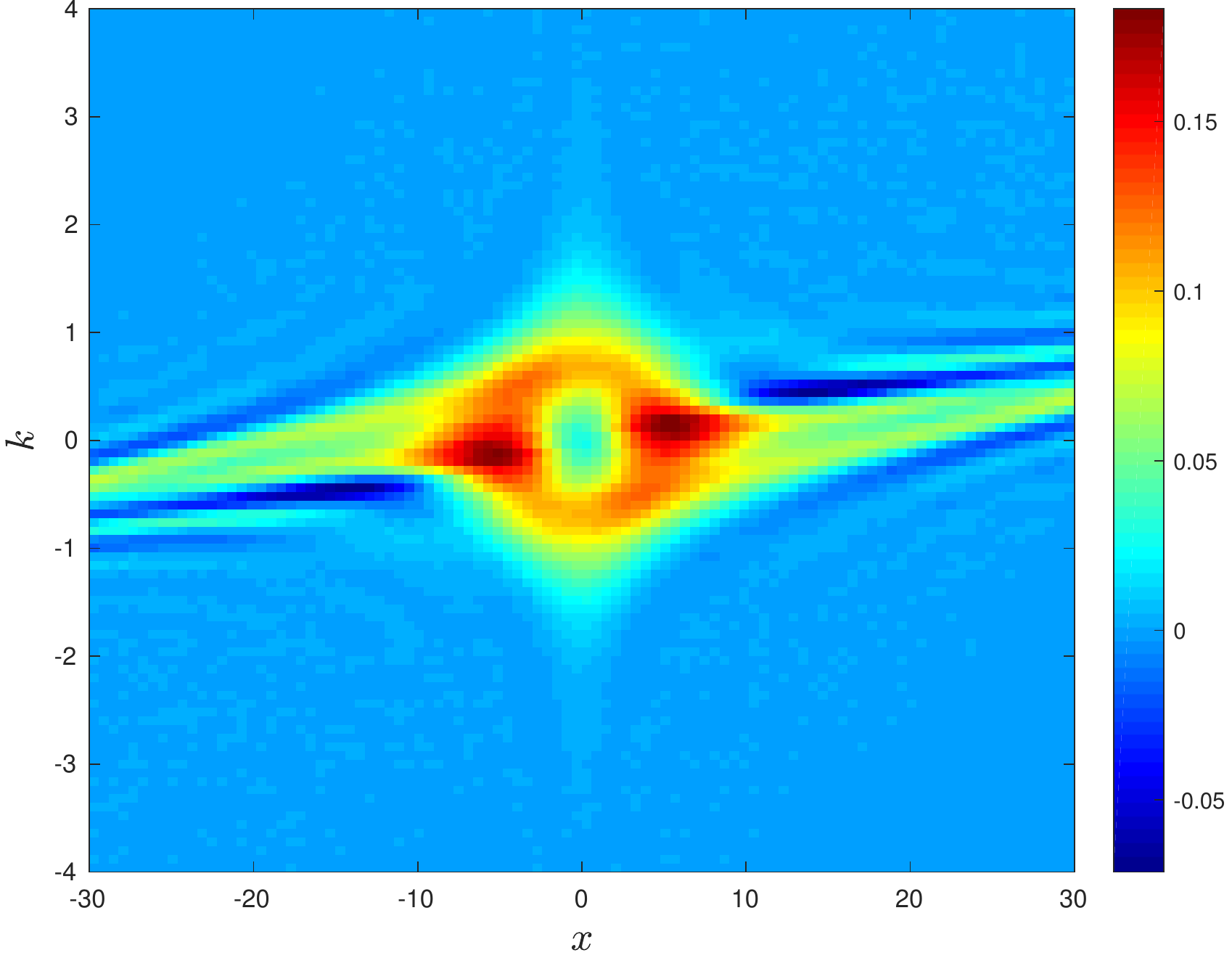}}
\caption{\small The 4D Helium-like system: Reduced Wigner function at $t_{fin}=20$a.u. under different {sample size} $N_\alpha$. The stochastic noises are clearly suppressed by increasing $N_\alpha$.}
\label{He_ss_pic}
\end{figure}

\subsection{Variance and the choice of the auxiliary function $\gamma$}

Another factor that influences the quality of a stochastic estimator is the variance. In this part, we will observe that increasing the auxiliary function $\gamma$ in \textbf{sp1} fails to lead to a systematic improvement on the accuracy, whereas \textbf{wp} allows a variance reduction \cite{ShaoXiong2016}. To study how the choice of a constant auxiliary function $\gamma$ influences the variance, we perform six groups of the 2D Gaussian scattering with the sample size $N_\alpha = 1\times 10^7$ and the final time $t_{fin} =10$fs. The resampling is turned off here in order to get rid of the deterministic errors.  The results are shown in Figs.~\ref{G_boot} and \ref{np_growth}, from which we can figure out the following lessons.

\begin{description}

\item[(1)] For \textbf{sp1}, the choice of $\gamma$ has little influence on the accuracy, which may imply that the variance of \textbf{sp1} doesn't depend on the choice of the auxiliary function. By contrast, \textbf{wp} clearly allows a variance reduction as the stochastic errors decrease sharply when $\gamma$ increases. Furthermore, the numerical convergence rate with respect to $\gamma$ is about $\mathcal{O}(\gamma^{-2})$.

\item[(2)] The growth rate of particle number $N/N_{\alpha}$ in \textbf{sp1} is also independent of $\gamma$, as we have predicted in Theorem \ref{th:exp}. Actually, the growth rate is slightly smaller than the theoretical upper bound $\me^{2\check{\xi} t}$, and in \textbf{wp} the growth rate is very close to the theoretical prediction of $\me^{2\gamma_0 t}$. 

\end{description}

\begin{figure}[!h]
\subfigure[Convergence rate with respect to $\gamma$.]{\includegraphics[width=0.49\textwidth,height=0.27\textwidth]{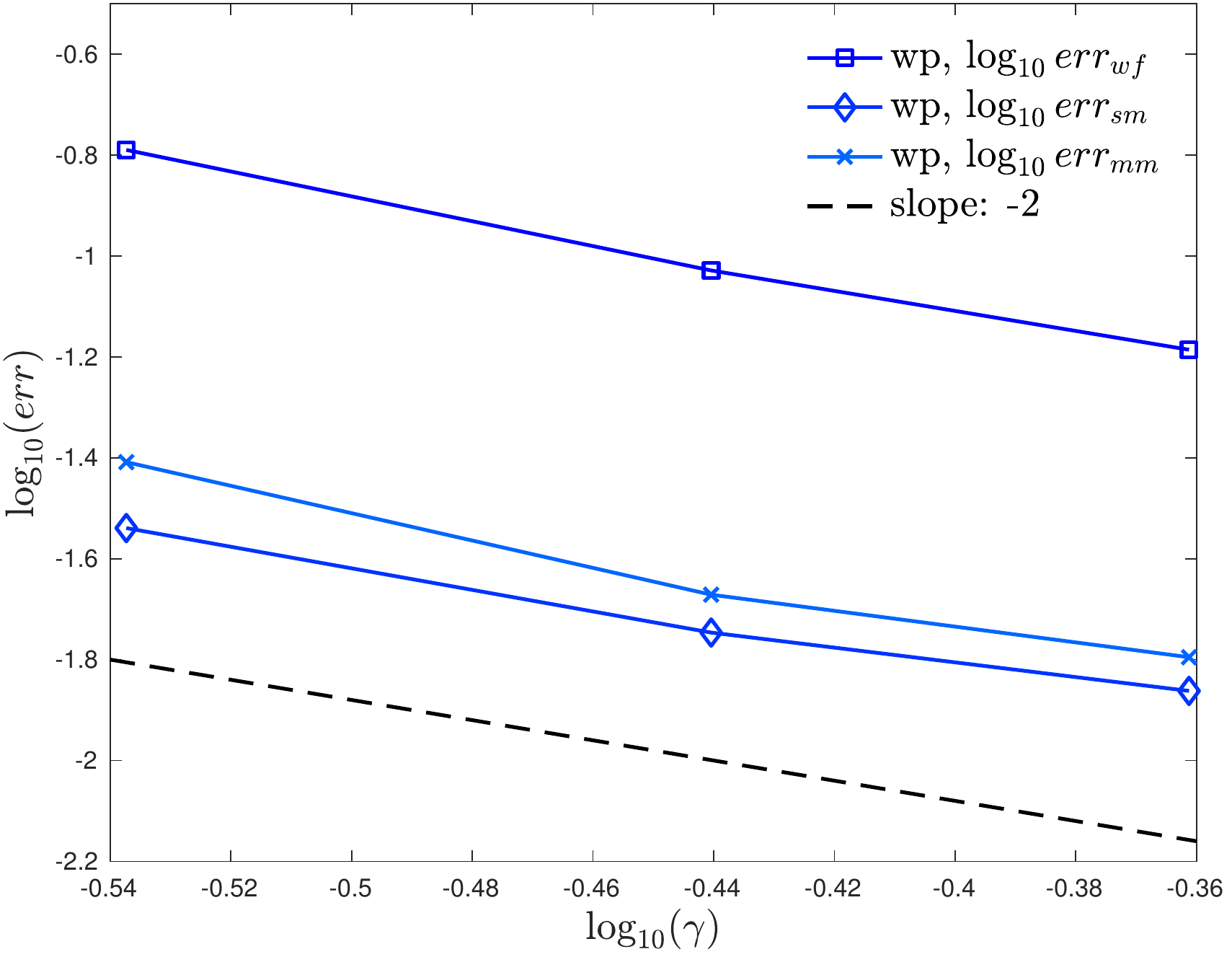}}
\subfigure[$\textup{err}_{wf}$ under different $\gamma$.]{\includegraphics[width=0.49\textwidth,height=0.27\textwidth]{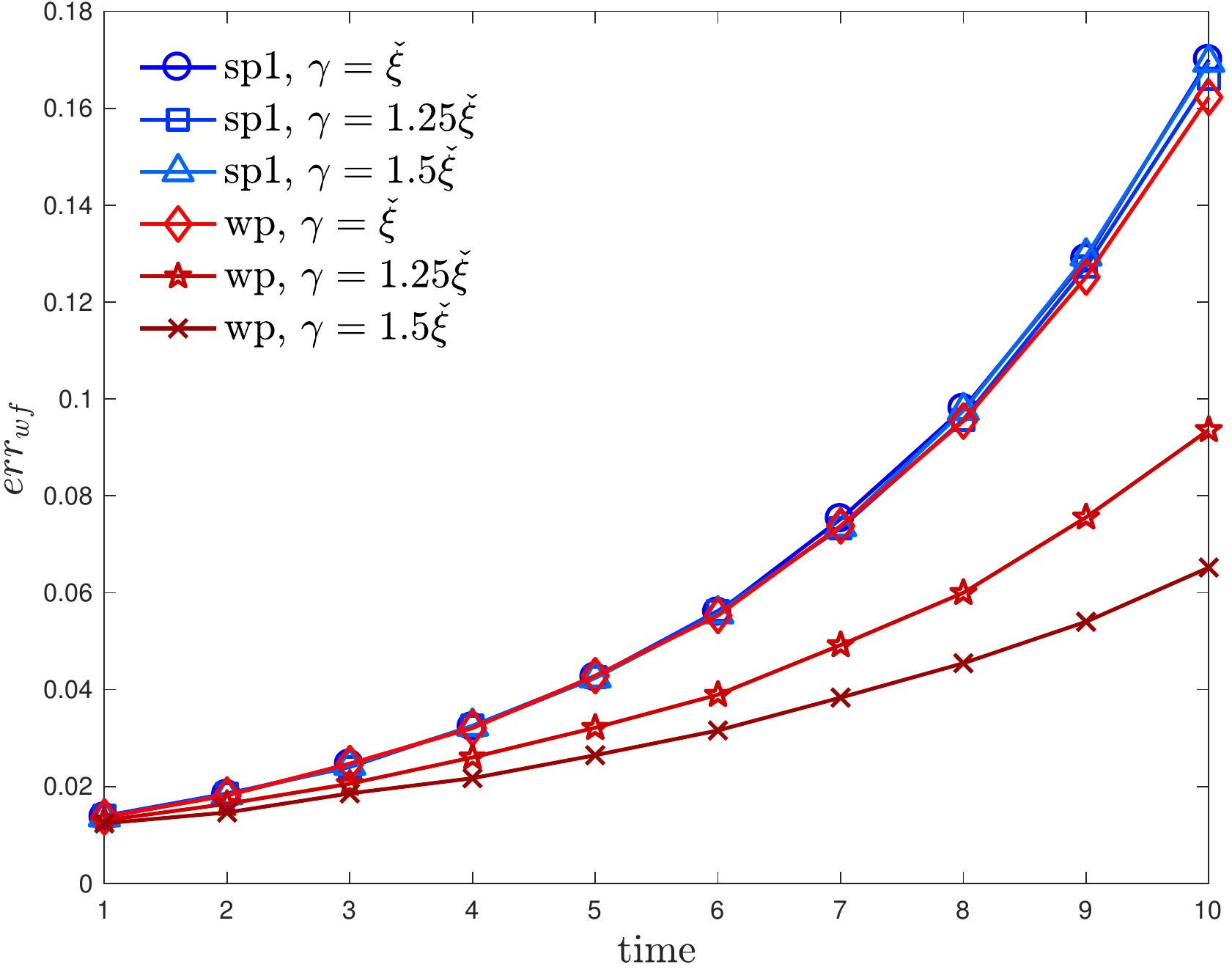}}
\subfigure[$\textup{err}_{sm}$ under different $\gamma$.]{\includegraphics[width=0.49\textwidth,height=0.27\textwidth]{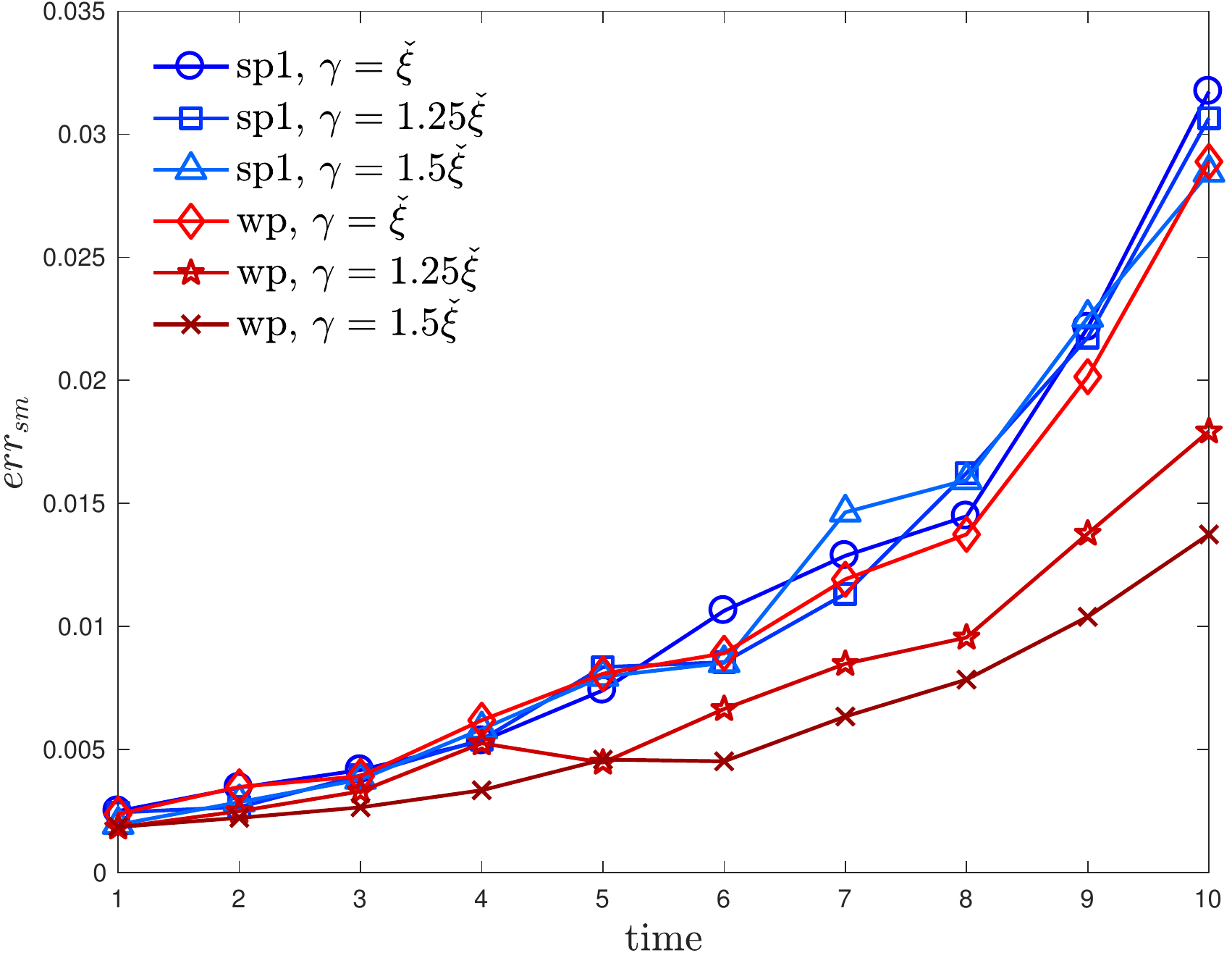}}
\subfigure[$\textup{err}_{mm}$ under different $\gamma$.]{\includegraphics[width=0.49\textwidth,height=0.27\textwidth]{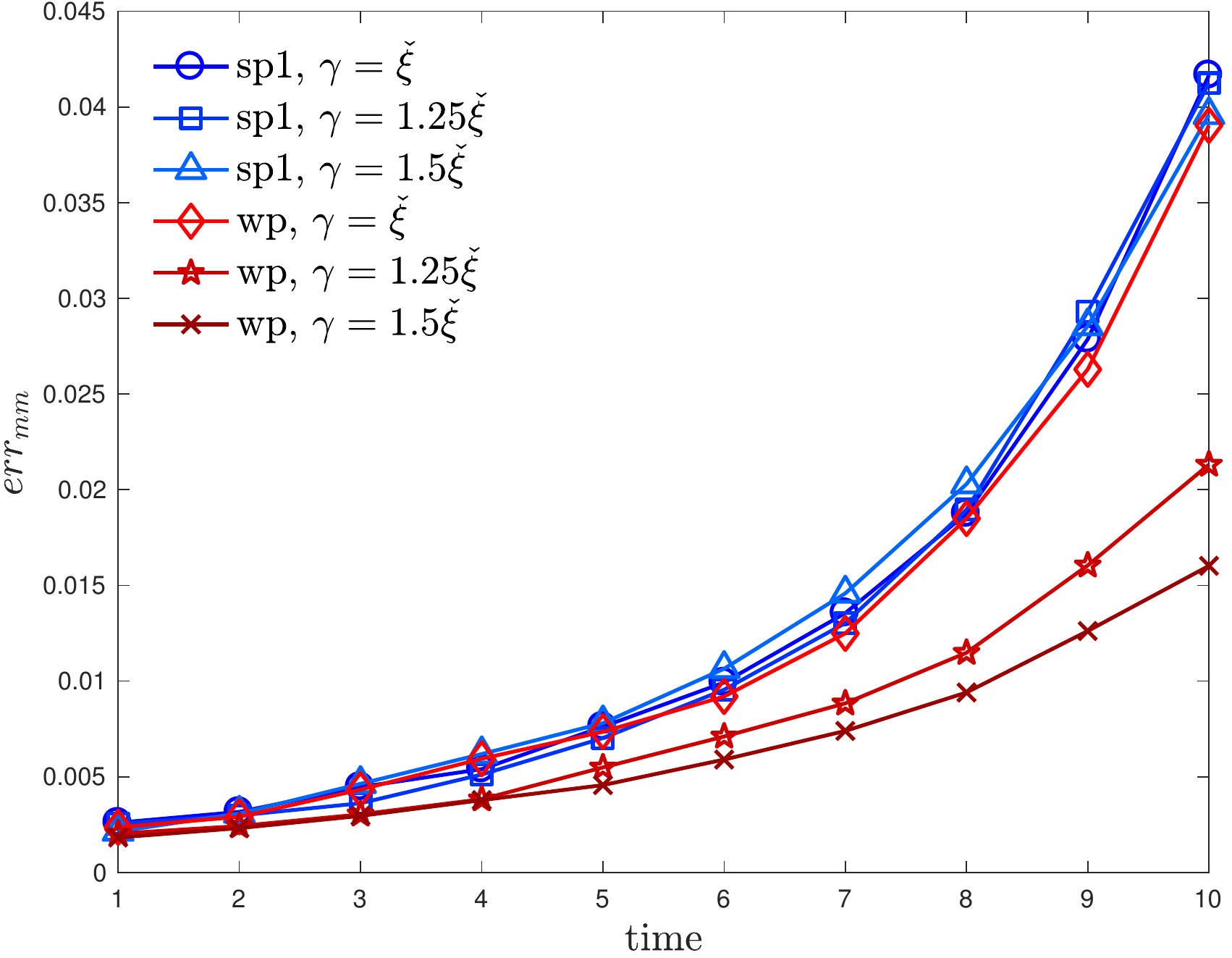}}
\caption{\small  The 2D Gaussian scattering: Relative errors under different constant auxiliary function $\gamma$ and the convergence rate at $t_{fin}=10$fs. The accuracy of \textbf{wp} is improved as $\gamma$ increases, and the numerical convergence rate is $\mathcal{O}(\gamma^{-2})$. On the contrary, the accuracy of \textbf{sp1} cannot be improved by adjusting $\gamma$. Here we set the sample size $N_\alpha = 1\times 10^7$ and the resampling is not used.
}
\label{G_boot}
\end{figure}

\subsection{Accuracy and efficiency of resampling}

The accuracy of resampling is closely related to the resampling frequency $1/T_A$, while the efficiency depends on the sample size $N_{\alpha}$ and the partition size $N_h$.  We simulate six groups of the 2D Gaussian scattering with $N_\alpha = 10^7$ and $\gamma = 1.5\check{\xi}$, as well as four groups of the 4D Helium-like system with $N_\alpha = 1\times 10^8$ and $\gamma = 2$. The results under different $T_A$ are compared in Fig.~\ref{G_He_NA}. In addition, we collect the data on the growth rate of particle number $N/N_{\alpha}$  and the particle number after resampling $\#_P^a$ in the previous simulations. The numerical results uncover the following facts.

\begin{description}

\item[(1)] We have evidence that the resampling procedure is indispensible for two reasons. First, the particle number in pure Monte Carlo simulations will inevitably grow exponentially and soon exceeds the limit of memory storage, while the resampling procedure can significantly kill the redundant particles and suppress the particle growth, as shown in Fig.~\ref{np_growth}. Second, although the resampling introduces an additional error term, it suppresses the stochastic errors in the Monte Carlo simulation, thereby improving the accuracy as presented in Fig.~\ref{G_ss_error}.

\item[(2)] Too frequent resampling leads to a smaller $\#_{P}^{a}$, at the cost of the decline in the accuracy for both 2D and 4D simulations. On the other hand, too low frequency does not necessarily improve the accuracy, as presented in Fig.~\ref{G_He_NA}, due to the accumulation of stochastic errors. Therefore, there may be an optimal choice of $T_A$ to strike the balance between the accuracy and the efficiency as already proposed in \cite{ShaoXiong2016}. 

\item[(3)] The efficiency of resampling is measured by $\#_{P}^a$. We can clearly see that the resampling performs quite well when $N_\alpha \ge N_h$. However, according to Fig.~\ref{particle_number}, it works not so efficiently when $N_\alpha < N_h$. For instance, when $N_h= 10^{8}$ and $N_\alpha = 10^6$ in the 4D problem, $\#_{P}^a$ soon goes to the magnitude of the partition size and finally the particle number exceeds $10^9$, in which the rate $\#_{P}^a / N_{\alpha}$ is almost $728$ (\textbf{sp1}). Even for 2D problem, when $N_\alpha = 1\times 10^4$ and $N_h= 8\times 10^4$, we also find that $\#_{P}^a$ increases dramatically and exceeds $7.7\times 10^4$ at the final stage.  Therefore, the ``bottom line' structure described in \cite{ShaoSellier2015} depends on not only the oscillating structure of the Wigner function, but also the partition size $N_h$. These observations reveal the potential weakness of the uniform histogram,  as also pointed out in the statistical community. Actually, to maintain the same accuracy for the high dimensional histogram, the number of samples must also grow exponentially, otherwise a severe `overfitting' will be observed \cite{bk:HastieTibshiraniFriedman2009}. In practice, $N_{\alpha} \approx 10 N_h$ is a balanced choice in consideration of the efficiency of the resampling. But this condition limits the applicability of the uniform histogram in high dimensional problems.

\item[(4)] $\#_{P}^{a}$ is usually larger in \textbf{sp1} than that in \textbf{wp}, although the converse is true for the growth rate of particle number. In fact, the variance reduction in \textbf{wp} helps suppress more random noises and kill more redundant particles. Fortunately, the advantage in storing signed weights saves the performance of \textbf{sp1}.

\item[(5)] We record the total running time and the average value of $\#_{P}^a$ (denoted by $\bar{\#}_{P}^a)$ in Table \ref{Table_1}, with the computing platform: Intel(R) Xeon(R) CPU E5-2680 v4 (2.40GHz, 35MB Cache, 9.6GT/s QPI Speed, 14 Cores, 28 Threads) and 256GB Memory (we use 14 threads for each task). It is readily seen there that the computational time is indeed proportional to $\bar{\#}_P^a$, instead of $N_\alpha$. This means the efficiency of WBRW depends largely on the quality of resampling. Although $\bar{\#}_P^{a}$ is usually larger, \textbf{sp1} is significantly faster because of the moderate growth of particle number.

\begin{table}[!h]
 \centering
 \caption{The 4D Helium-like system: $\bar{\#}_{P}^{a}$ at $t = 20$a.u. and the total running time.}
 \label{Table_1}
 \begin{tabular}{c | c | c | c |  c | c | c}
  \toprule
  \toprule
  & \multicolumn{3}{c|}{\textbf{sp1}} &   \multicolumn{3}{|c} {\textbf{wp}}\\
    \midrule
  $N_\alpha$ &  $1 \times 10^6$  &  $1 \times 10^7$ & $1 \times 10^8$ &  $1 \times 10^6$  &  $1 \times 10^7$ & $1 \times 10^8$\\
  \midrule
  $\bar{\#}_{P}^{a}$  & $9.03 \times 10^{8}$ & $9.40 \times 10^{8}$ &   $1.19 \times 10^{9}$ & $2.42\times 10^8$ & $2.67\times 10^8$ & $4.97\times 10^8$\\ 
  time(h)  & 44.949 & 46.837  &    59.666 & 51.792 &  64.033 & 162.313\\
  \bottomrule
  \bottomrule
  \end{tabular}
\end{table}

\end{description}

\begin{figure}[!h]
\subfigure[$\textup{err}_{wf}$ under different $T_A$ (left: 2D Gaussian, right: 4D Helium).]{{\includegraphics[width=0.49\textwidth,height=0.27\textwidth]{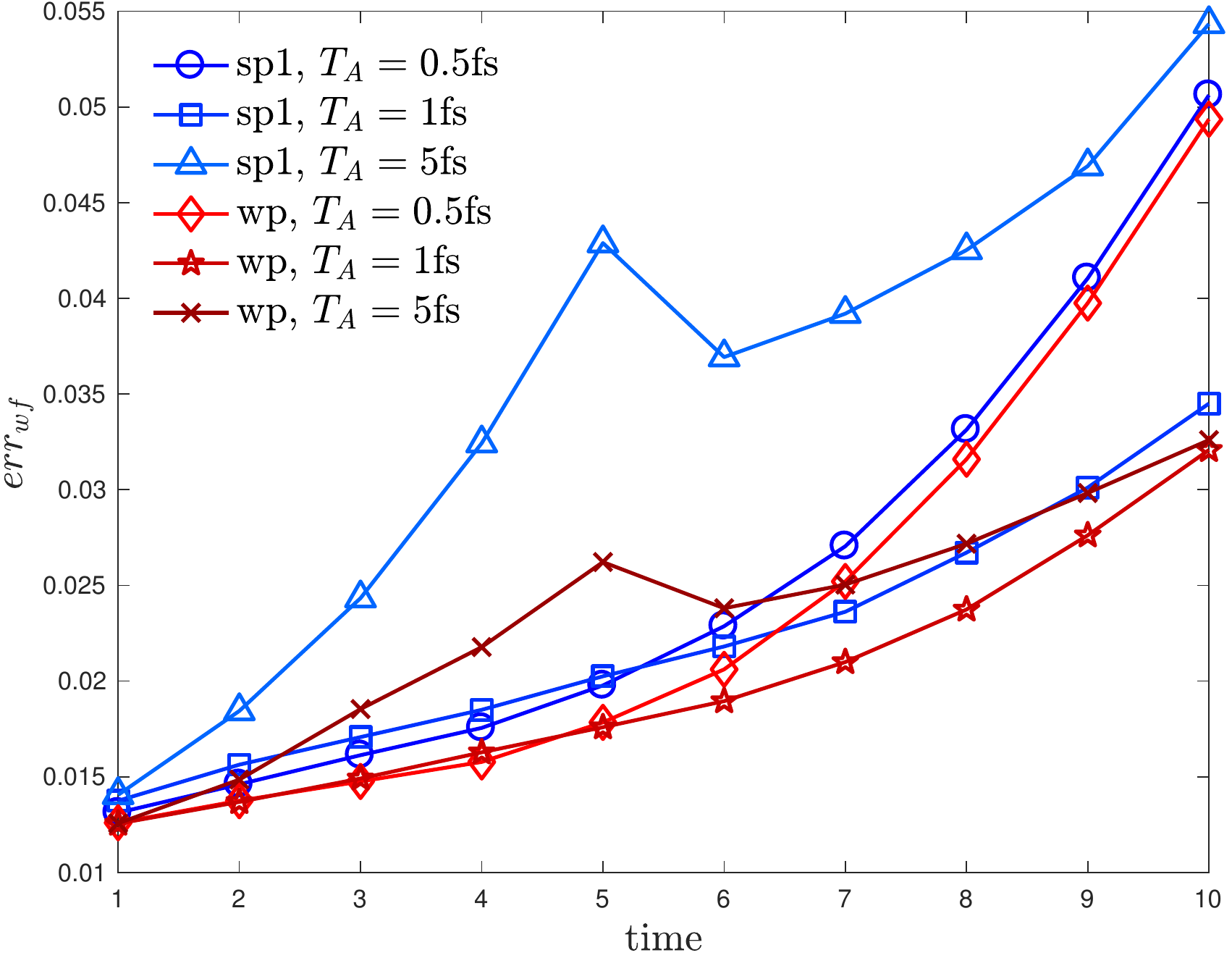}}
{\includegraphics[width=0.49\textwidth,height=0.27\textwidth]{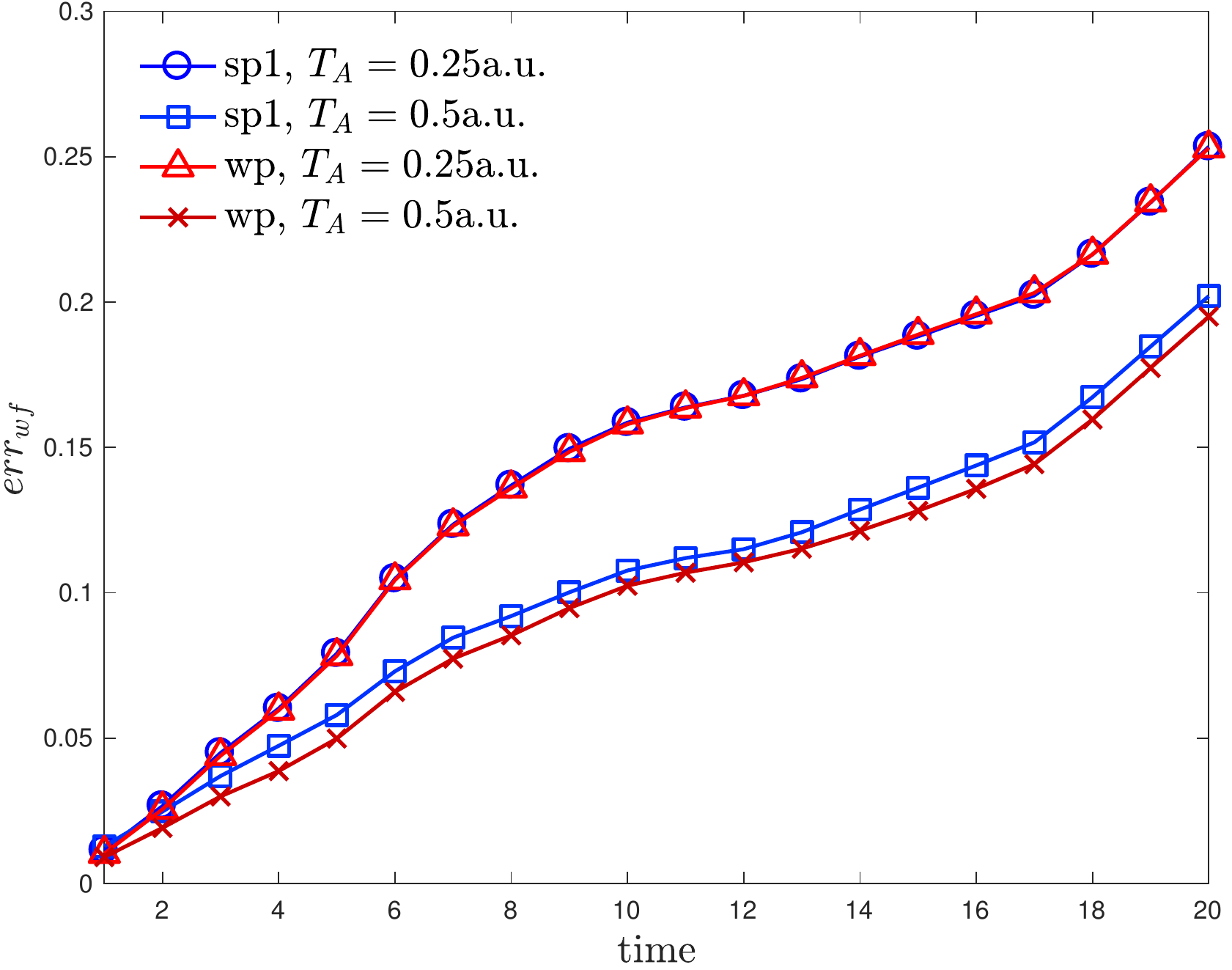}}}
\subfigure[$\textup{err}_{sm}$ under different $T_A$ (left: 2D Gaussian, right: 4D Helium).]
{{\includegraphics[width=0.49\textwidth,height=0.27\textwidth]{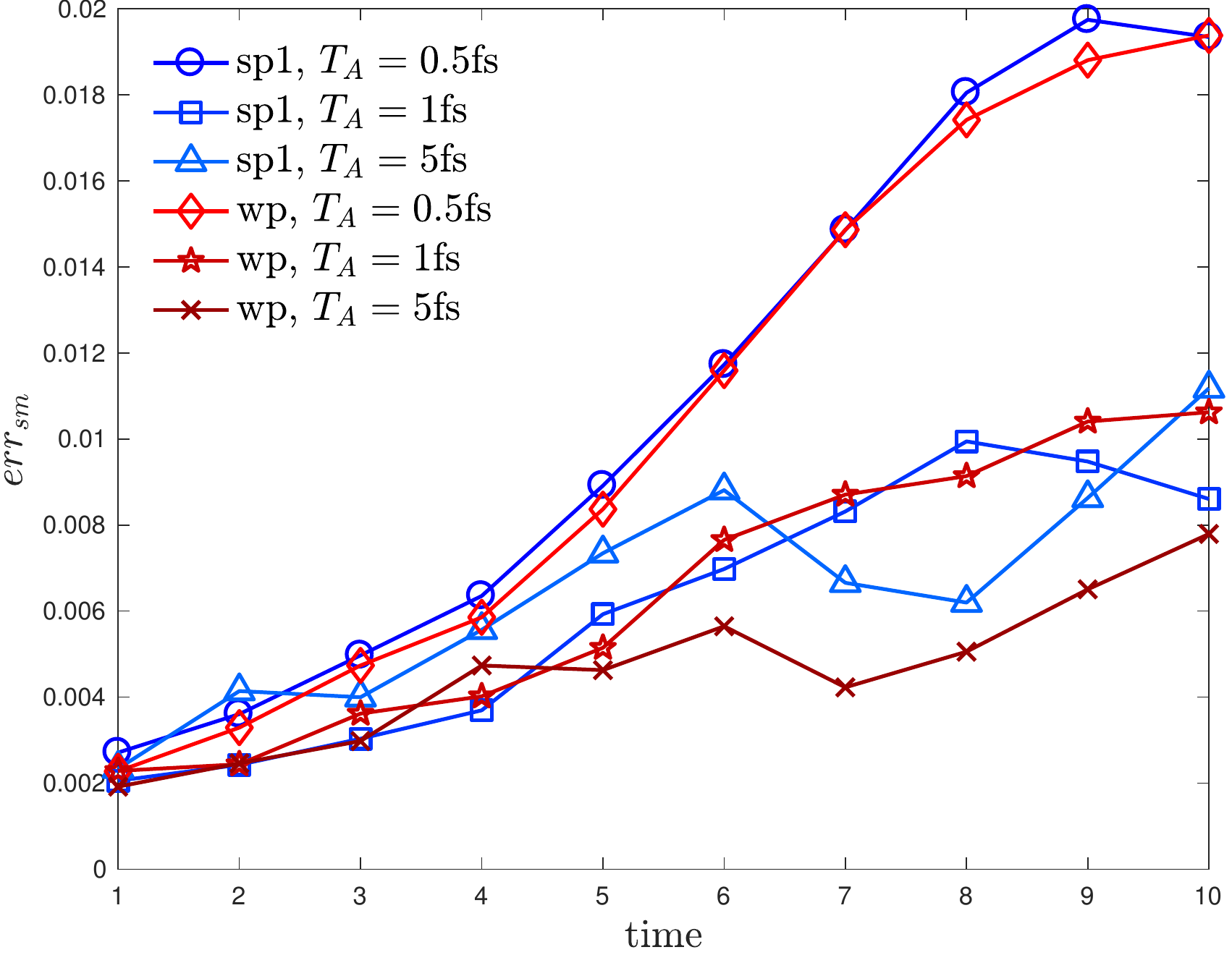}}
{\includegraphics[width=0.49\textwidth,height=0.27\textwidth]{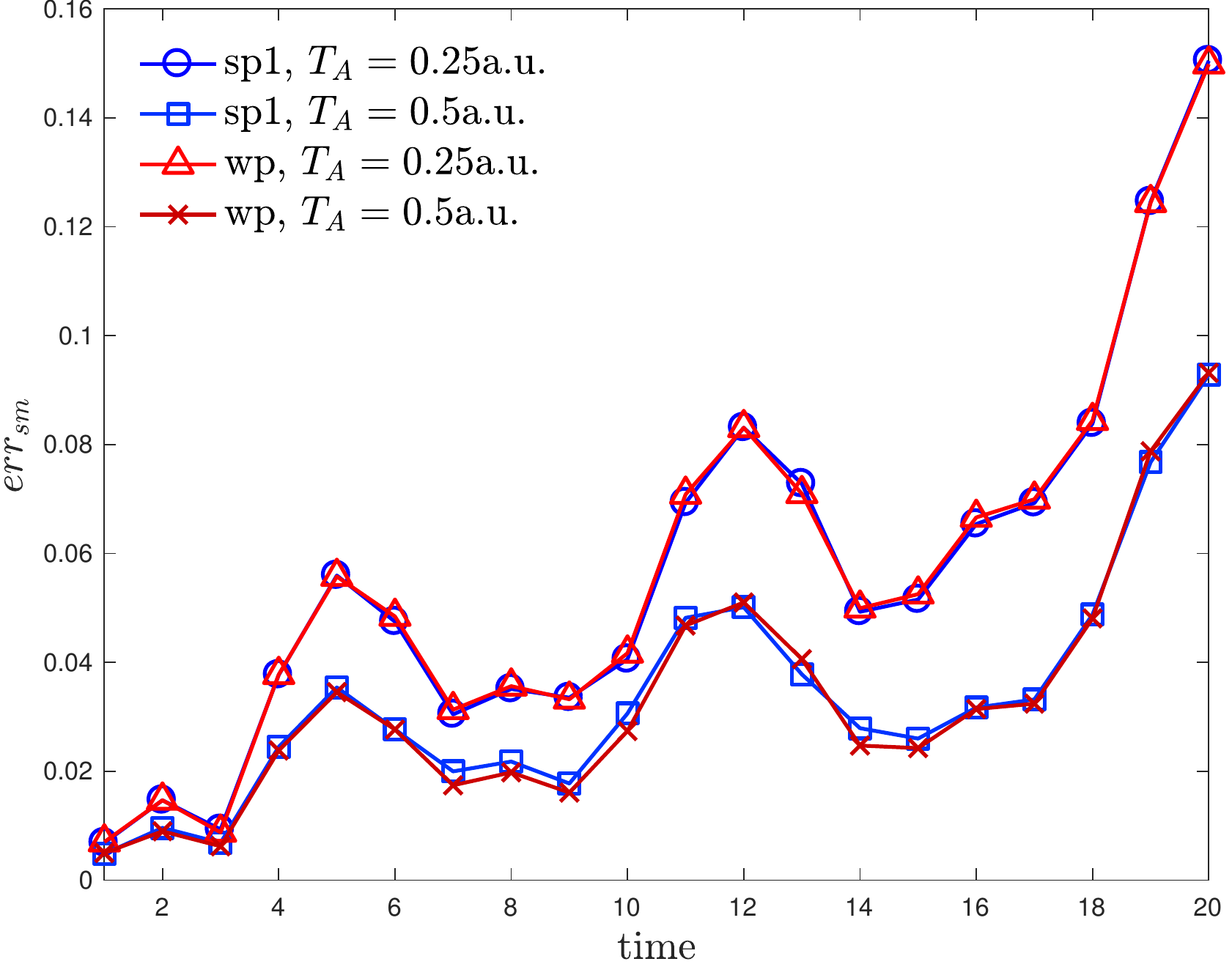}}}
\subfigure[$\textup{err}_{mm}$ under different $T_A$ (left: 2D Gaussian, right: 4D Helium).]{{\includegraphics[width=0.49\textwidth,height=0.27\textwidth]{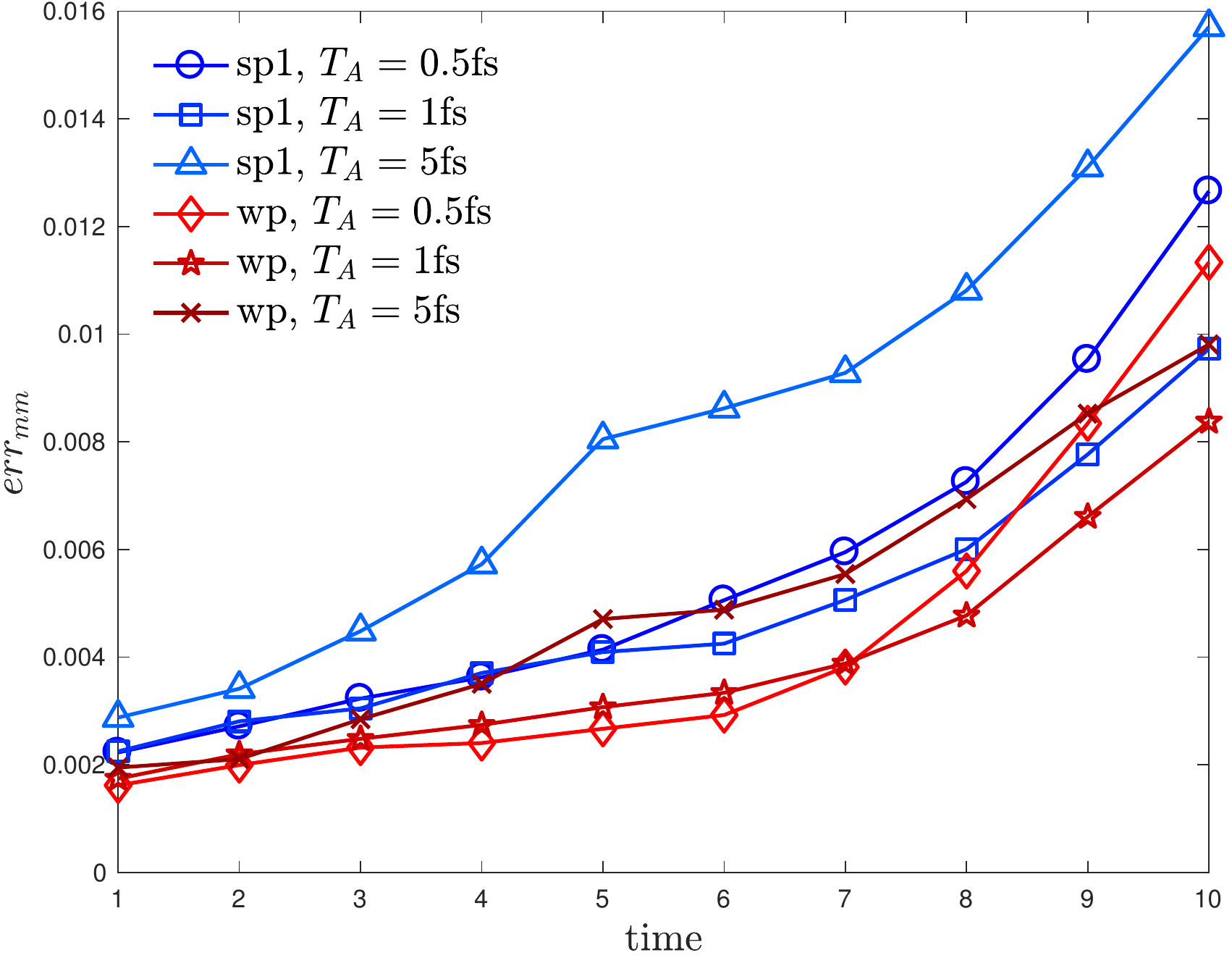}}
{\includegraphics[width=0.49\textwidth,height=0.27\textwidth]{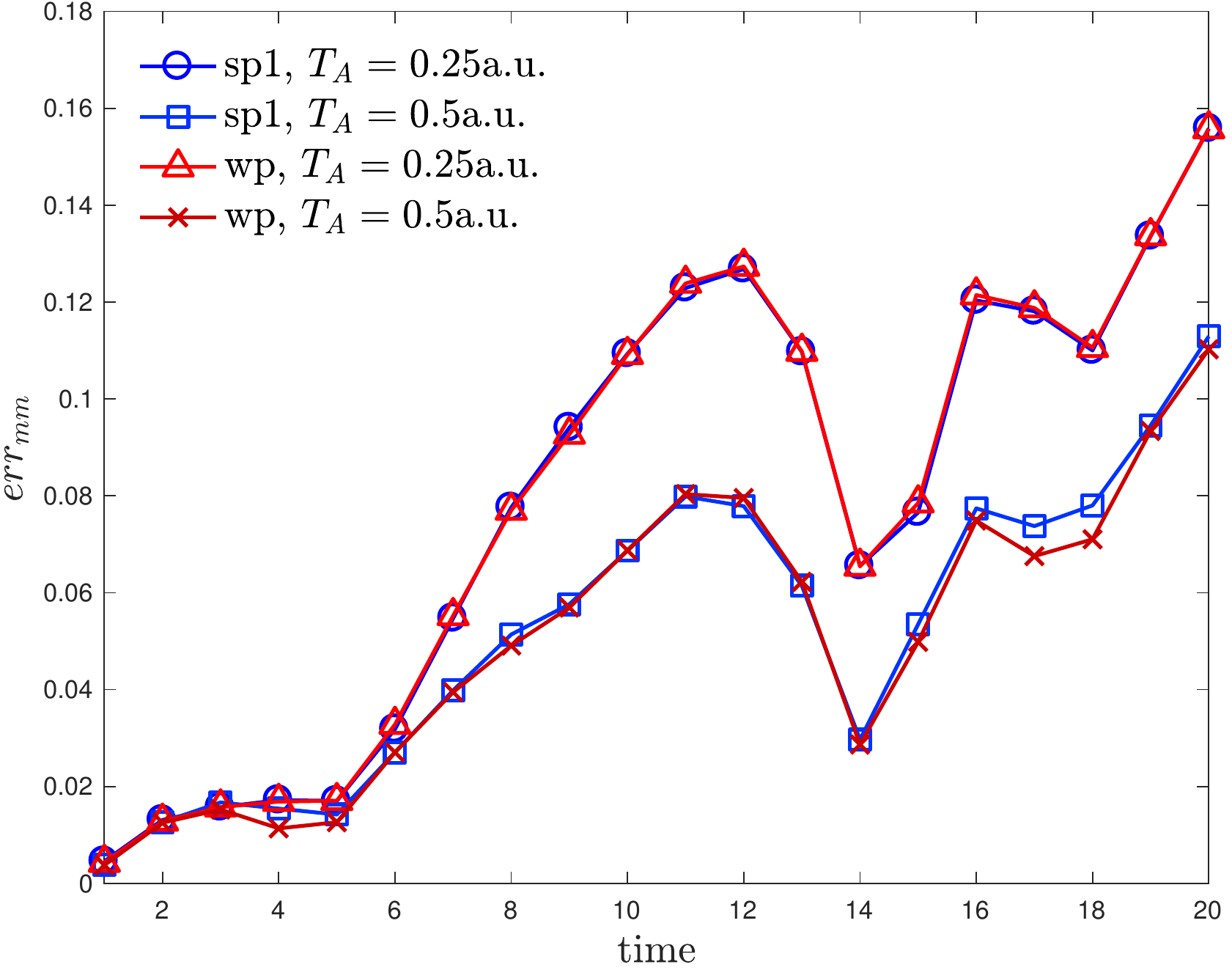}}}
\caption{\small The relative errors of the Wigner function, spatial and marginal distributions for the 2D Gaussian scattering (left) and the 4D Helium-like system (right). The relation between the relative errors and the resampling frequency $1/T_A$ is revealed. Too frequent resampling leads to a decline in the accuracy, while too low frequency also leads to an accumulation of stochastic errors and diminish the accuracy as a result. $\gamma = 1.5\check{\xi}$, $N_\alpha = 1 \times 10^7$, and $\gamma = 2$, $N_\alpha = 1\times 10^8$ are set for the 2D and 4D simulations, respectively. }
\label{G_He_NA}
\end{figure}

\begin{figure}[!h]
\subfigure[Growth rate of particle number (2D Gaussian). \label{np_growth}]
{\includegraphics[width=0.49\textwidth,height=0.27\textwidth]{{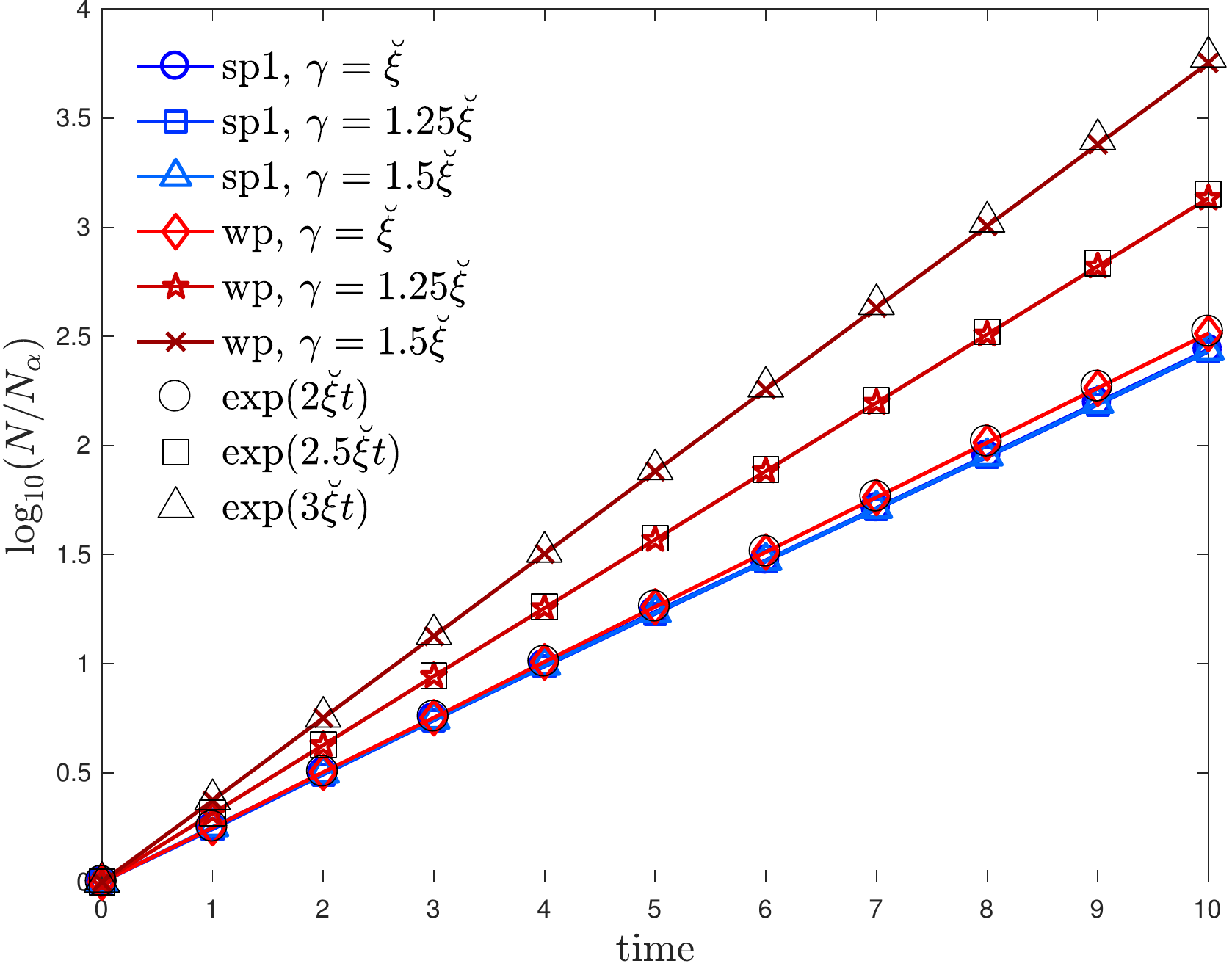}}
{\includegraphics[width=0.49\textwidth,height=0.27\textwidth]{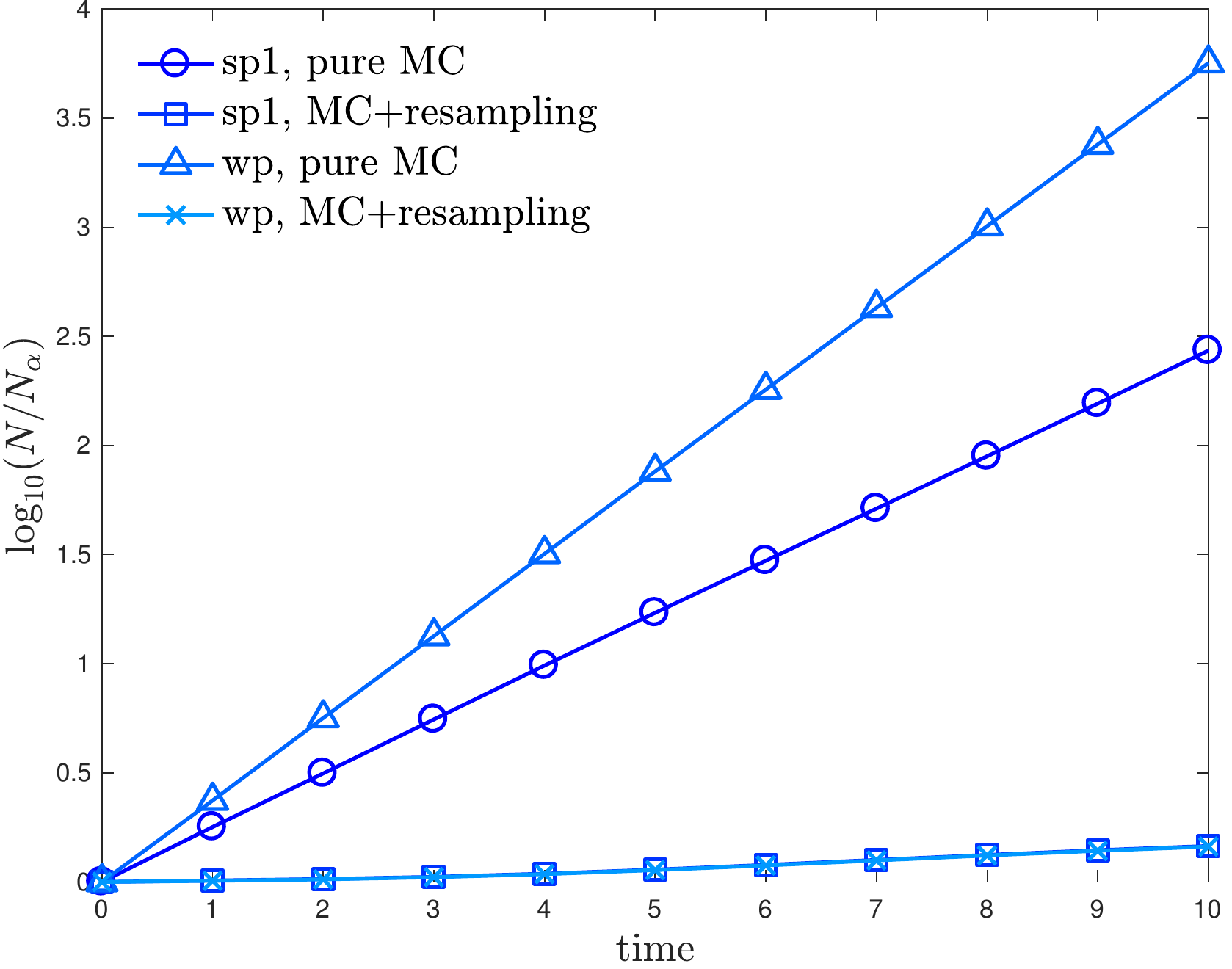}}}
\subfigure[$\#_P^a$ under different $T_A$(left: 2D Gaussian, right: 4D Helium).]
{{\includegraphics[width=0.49\textwidth,height=0.27\textwidth]{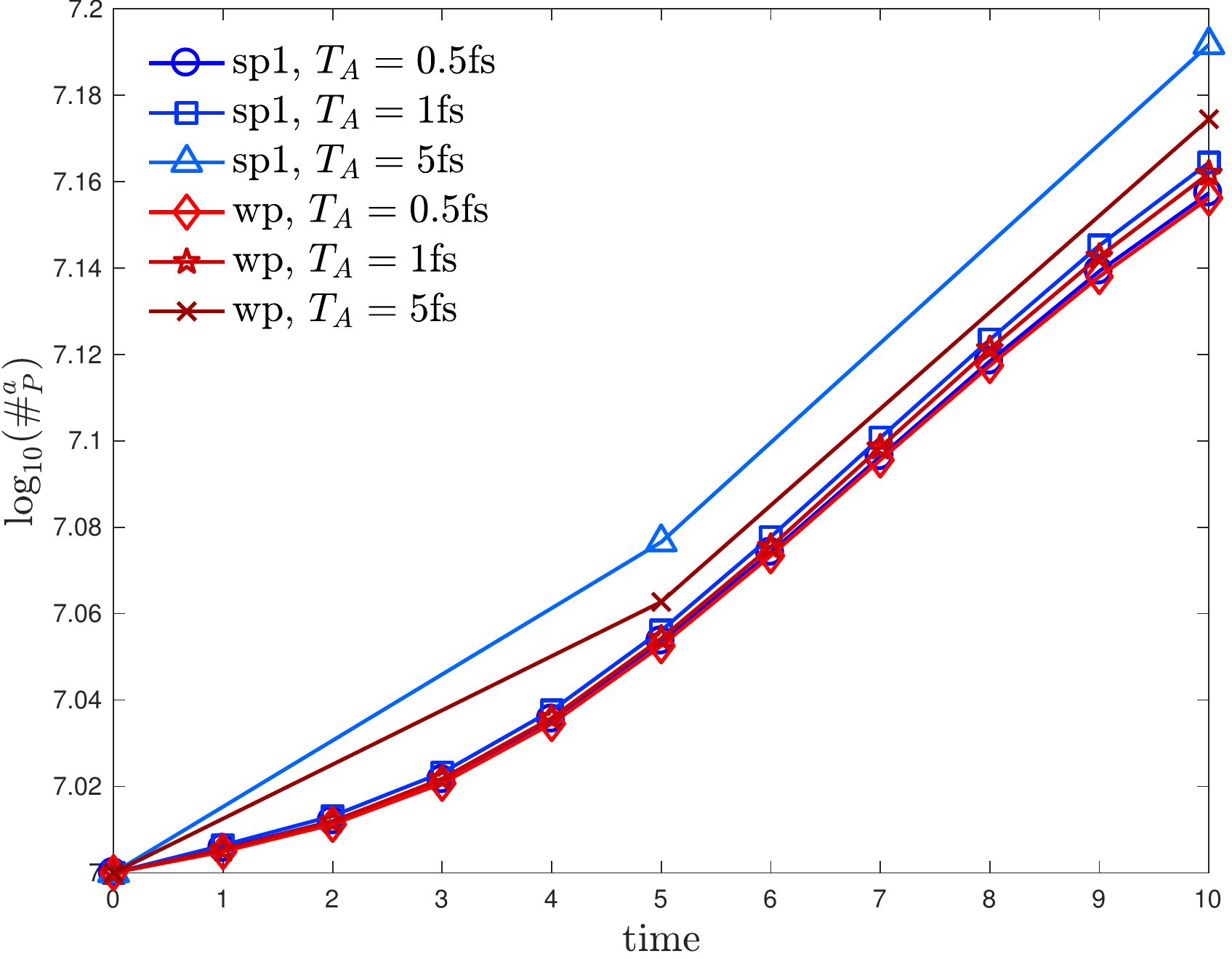}}
{\includegraphics[width=0.49\textwidth,height=0.27\textwidth]{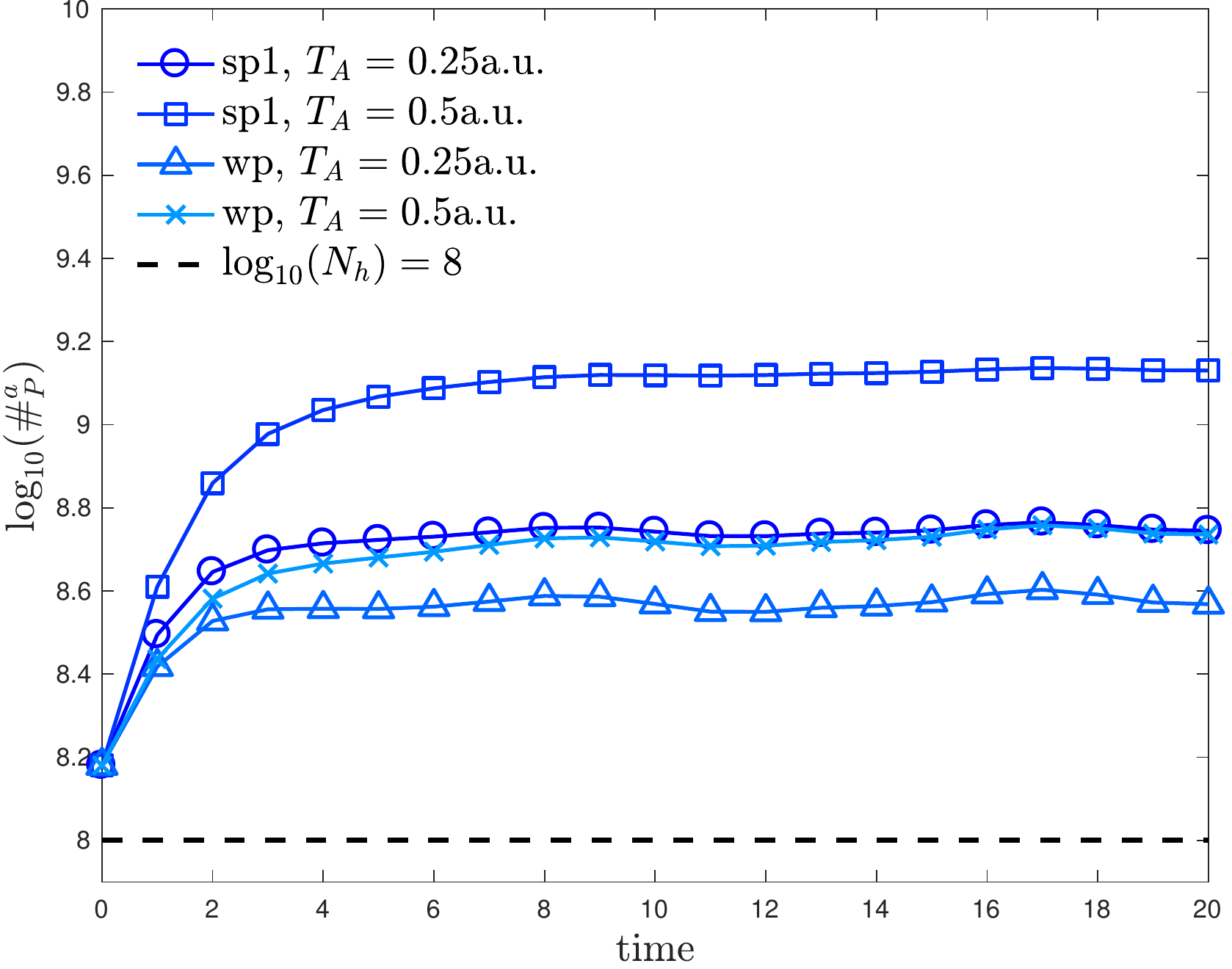}}}
\subfigure[$\#_P^a$ under different $N_\alpha$ (left: 2D Gaussian, right: 4D Helium).]{{\includegraphics[width=0.49\textwidth,height=0.27\textwidth]{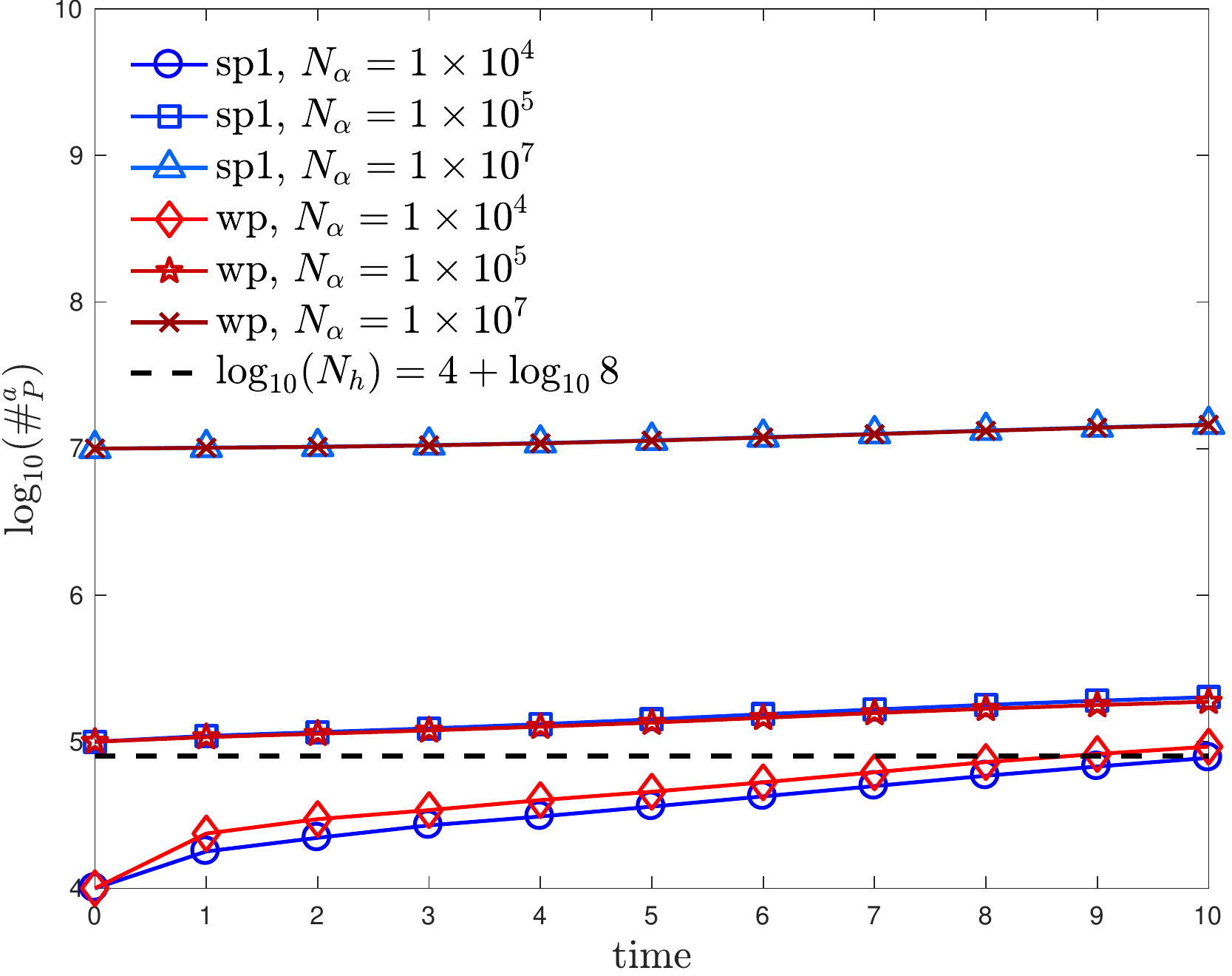}}
{\includegraphics[width=0.49\textwidth,height=0.27\textwidth]{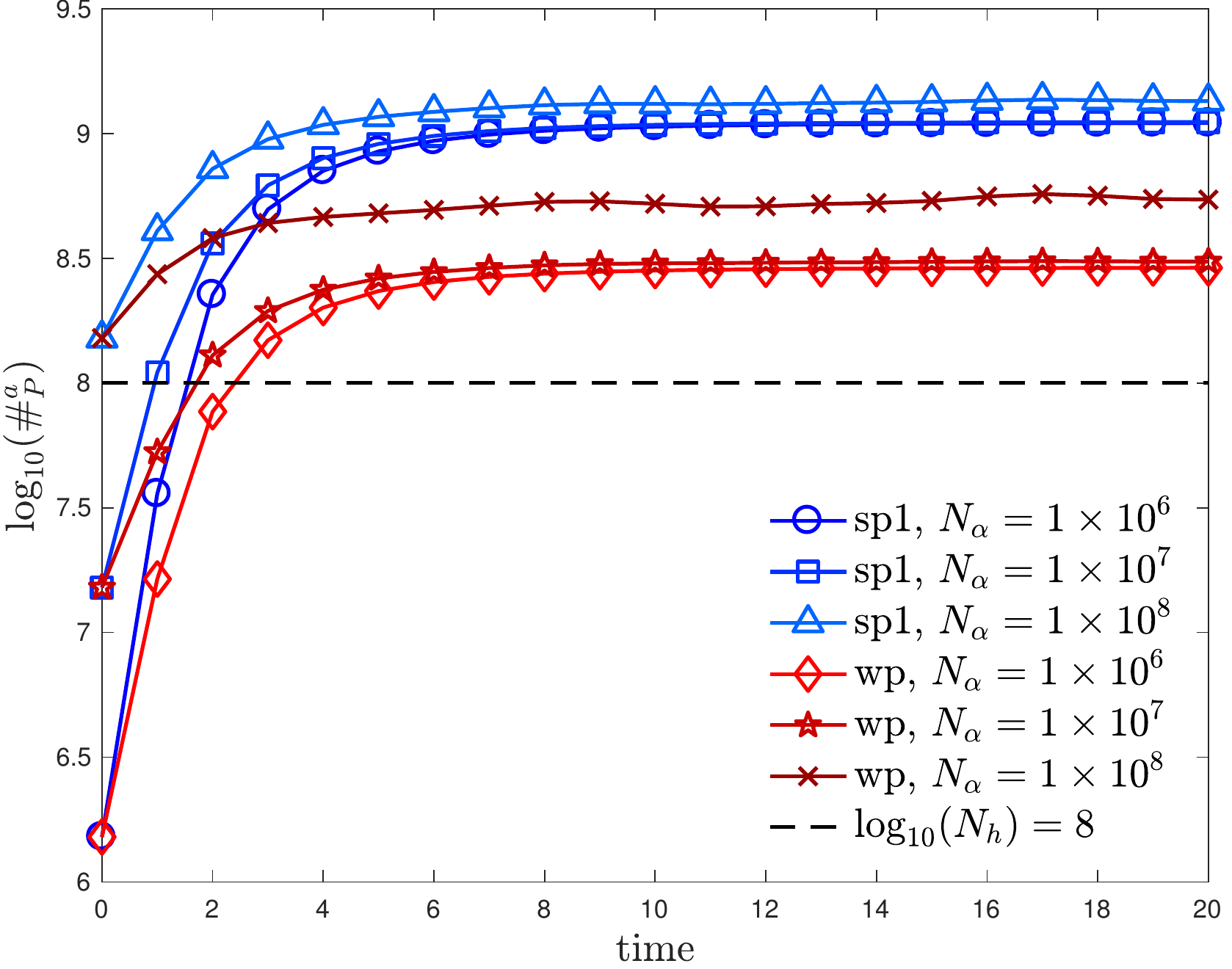}}}
\caption{\small  Relation between the growth of particle number and several factors: $\gamma$, sample size $N_\alpha$, resampling frequency $1/T_A$ and the partition size $N_h$. The particle number grows exponentially in time and is suppressed by resampling. More frequent resampling leads to a smaller $\#_{P}^a$. However, the efficiency of resampling deteriorates when $N_h$ is very large, and $\#_P^a$ finally becomes stable in a level that is comparable to $N_h$.}
\label{particle_number}
\end{figure}

\subsection{Bootstrap filtering}
\label{sec.bootstrap_filtering}

The bootstrap filtering is introduced to adjust the weights to $\{-1, 1\}$, thereby permitting us to store the histogram in two integer-valued matrices (for positive and negative particles). On the other hand, it brings some stochastic noises which diminish as $N_\alpha \to \infty$. To compare the accuracy of \textbf{sp1}, \textbf{sp2} and \textbf{wp}, we perform four groups of the 2D Gaussian scattering with $\gamma = 3\check{\xi}$ and $T_A=1$fs. Three kinds of the sample size $N_\alpha = 10^5, 10^6$ and $10^7$ are considered. Then to show the relation between the variance in \textbf{sp2} and the auxiliary function $\gamma$, we also do four groups of simulations with $N_\alpha = 1\times 10^7$ and $T_A=1$fs. Four constant auxiliary functions, ranging from $\check{\xi}$ to $4\check{\xi}$, are adopted.  All the numerical results are shown in Fig.~\ref{bootstrap}, from which we can figure out the following observations.

\begin{description}

\item[(1)] As the bootstrap filtering introduces some additional random errors, the accuracy of \textbf{sp2} lies between those of  \textbf{sp1} and \textbf{wp}.  \textbf{sp2} also converges as the sample size $N_\alpha$ increases, although the converge rate is slightly deviated from the order $-1/2$.

\item[(2)] \textbf{sp2} inherits the advantages in \textbf{wp} as increasing $\gamma$ leads to an evident improvement on the accuracy. Therefore, it successfully seizes the property of variance reduction of \textbf{wp}. However, the convergence of \textbf{sp2} with respect to $\gamma$ is slower than that of $\textbf{wp}$ due to the additional stochastic noises. 

\item[(3)] \textbf{sp2} utilizes two integer-valued matrices, instead of a double-valued one, for storing the histogram. This enables us to establish the histogram by simply counting the particle number in each bin and recording the coefficients $\lambda^\pm$ defined in Eq.~\eqref{eq:lambda}. The memory requirement of \textbf{sp2} is the same as \textbf{wp}, and twice more than \textbf{sp1}. For example, when $N_h = 100^4$, the memory for storing the histogram is $763$MB for \textbf{wp} or \textbf{sp2}, and $382$MB for \textbf{sp1}.

\end{description}

The bootstrap filtering indeed achieves the target as we have expected. Although it seems difficult to compare the deterministic errors induced by the histogram approximation and stochastic errors induced by the bootstrap filtering, there may be an optimal choice of $\gamma$, $N_\alpha$ and $\varepsilon$ to achieve a better convergence according to the theoretical error bound \eqref{histogram_bootstrap_error}.

\begin{figure}[!h]
\subfigure[Convergence rate with respect to $N_\alpha$.]
{{\includegraphics[width=0.49\textwidth,height=0.27\textwidth]{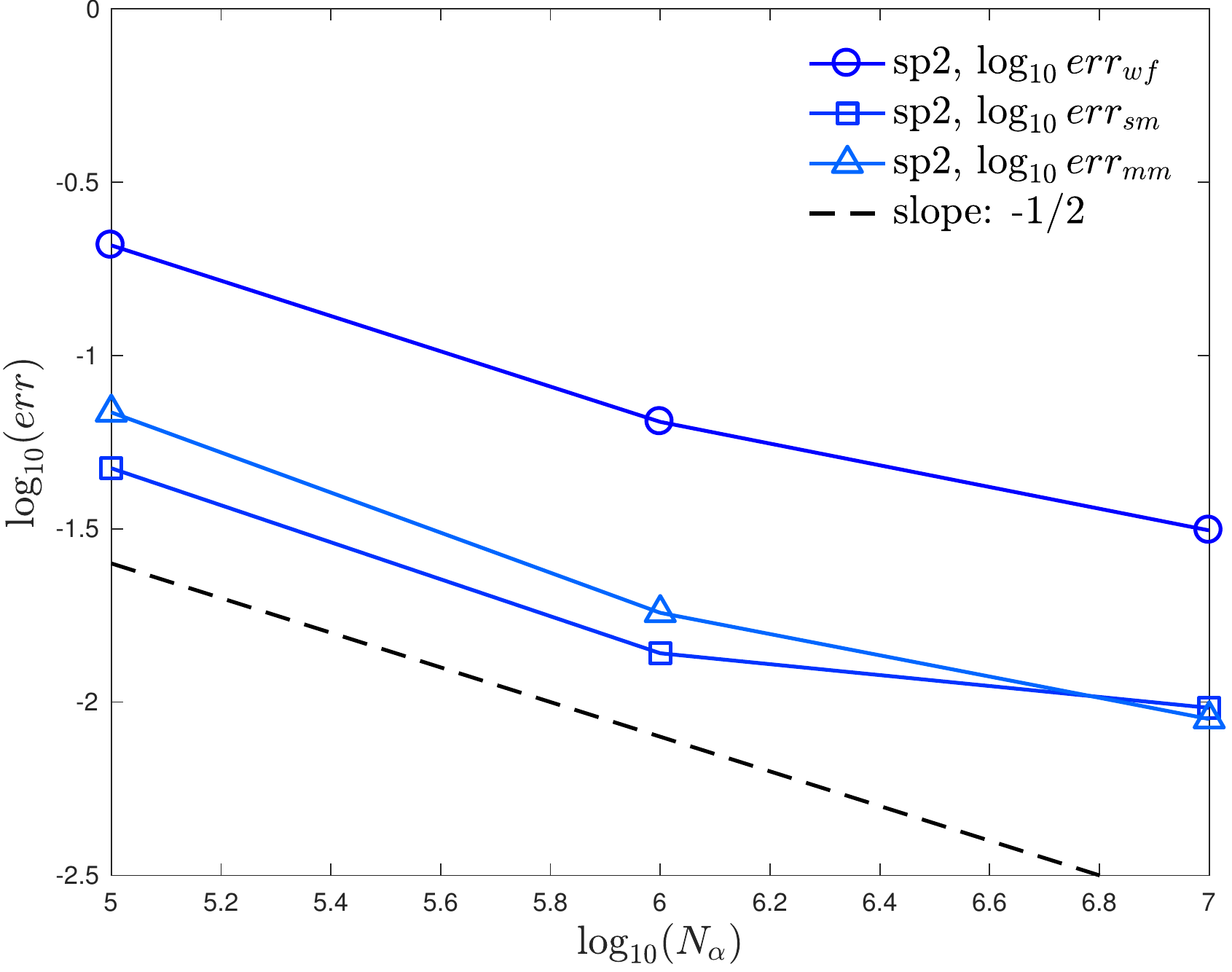}}}
\subfigure[Convergence rate with respect to $\gamma$.]
{{\includegraphics[width=0.49\textwidth,height=0.27\textwidth]{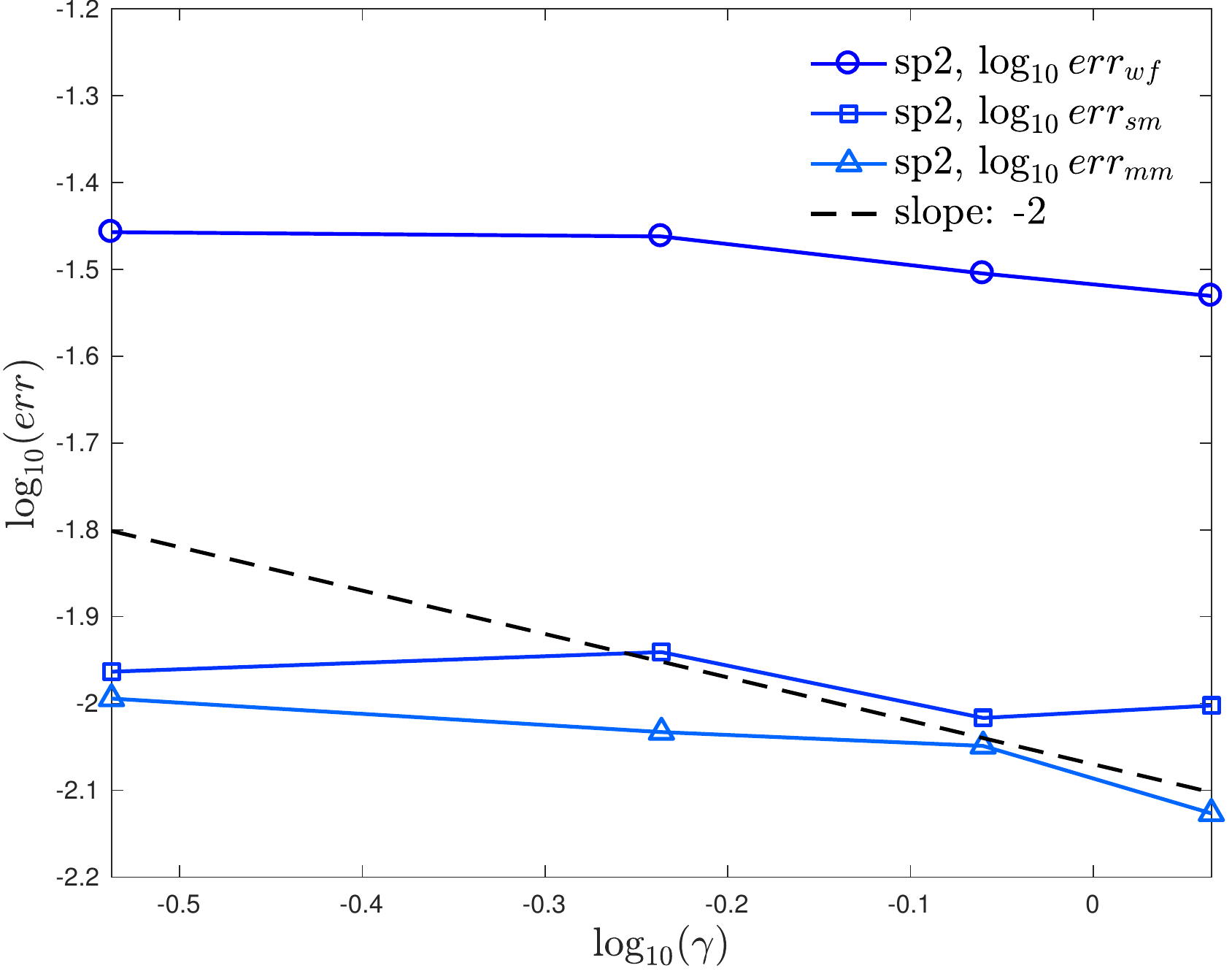}}}
\subfigure[$\textup{err}_{wf}$ under different $N_\alpha$ (left) or different $\gamma$ (right).]
{{\includegraphics[width=0.49\textwidth,height=0.27\textwidth]{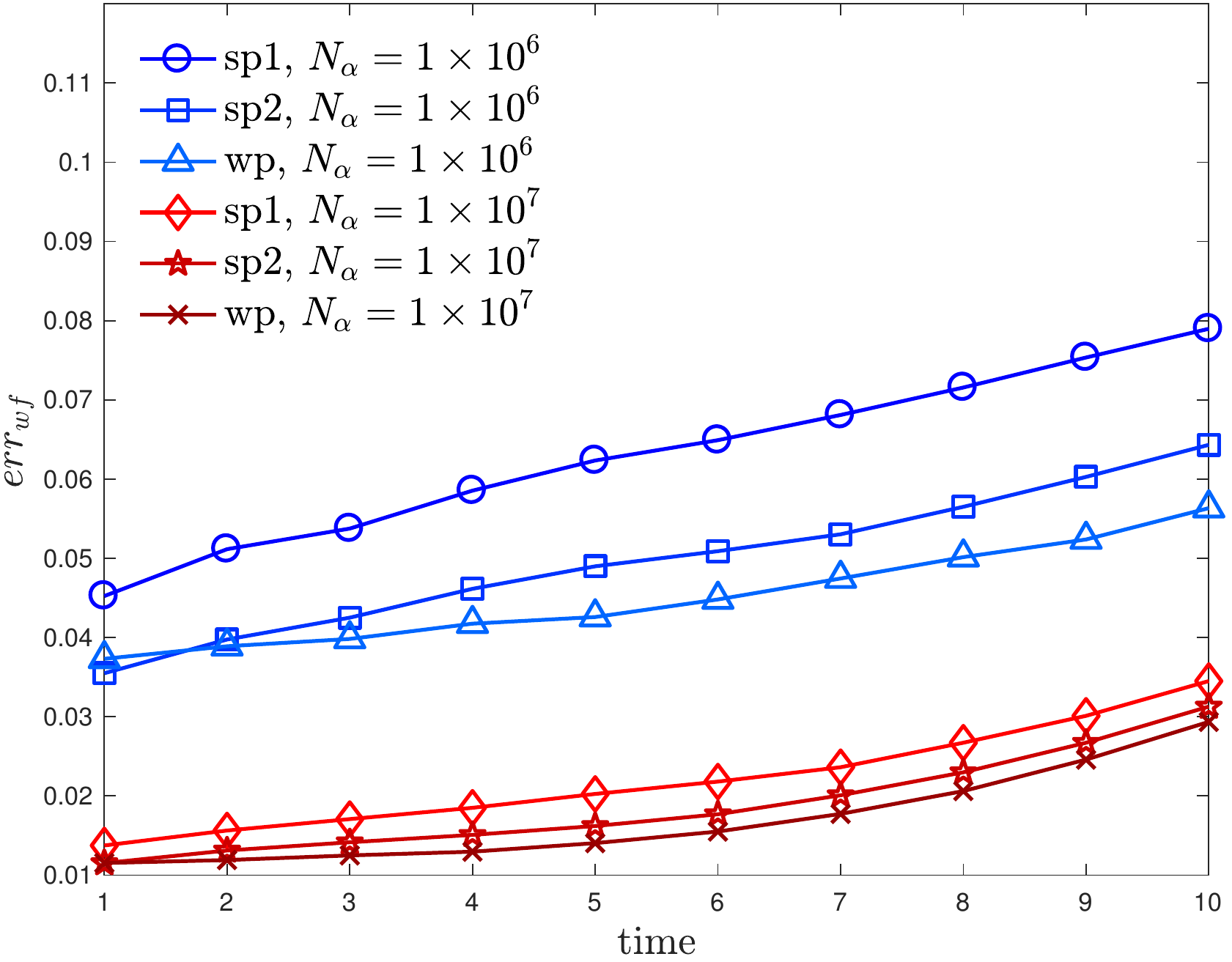}}
{\includegraphics[width=0.49\textwidth,height=0.27\textwidth]{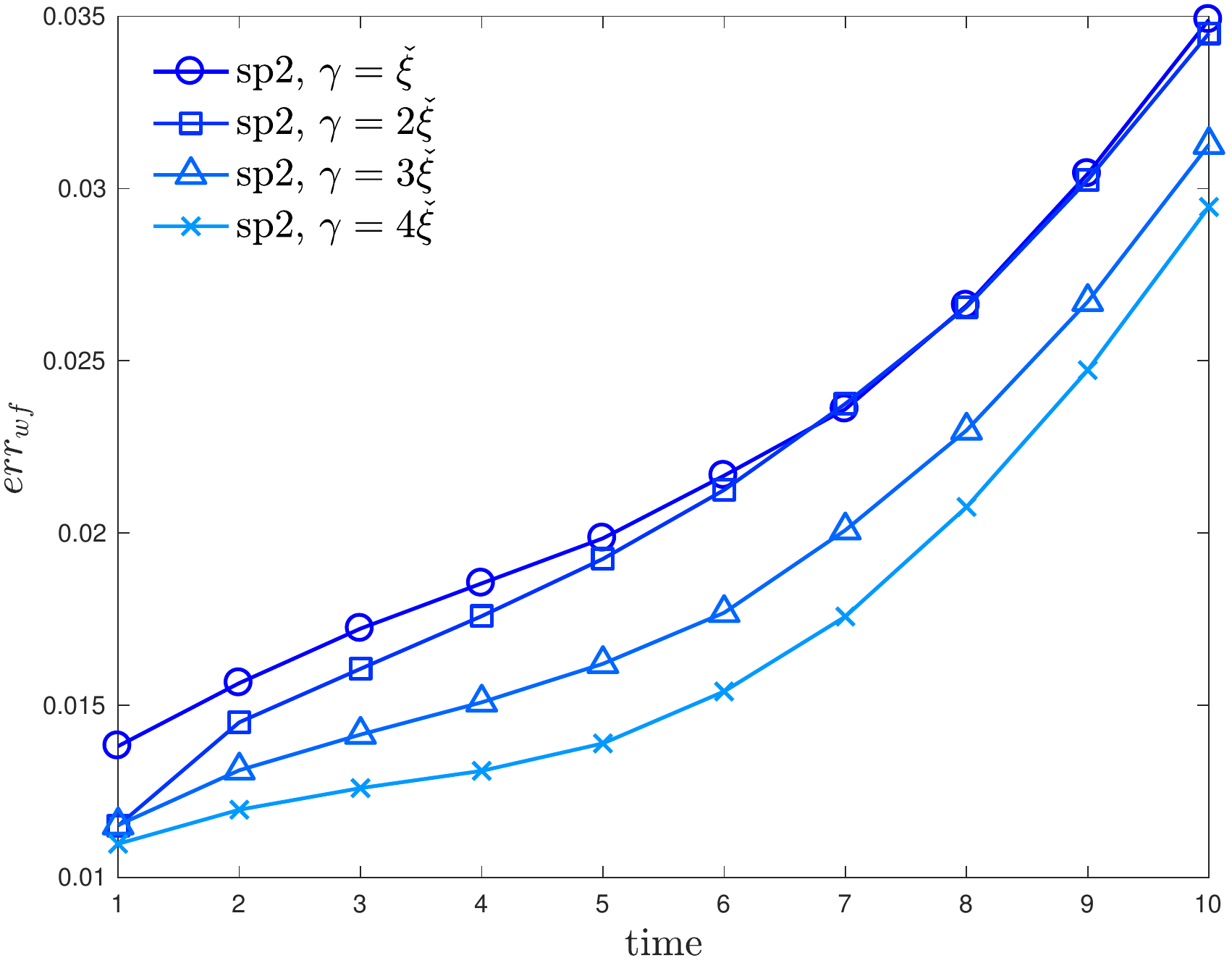}}}
\subfigure[$\textup{err}_{sm}$ under different $N_\alpha$ (left) or different $\gamma$ (right).]{{\includegraphics[width=0.49\textwidth,height=0.27\textwidth]{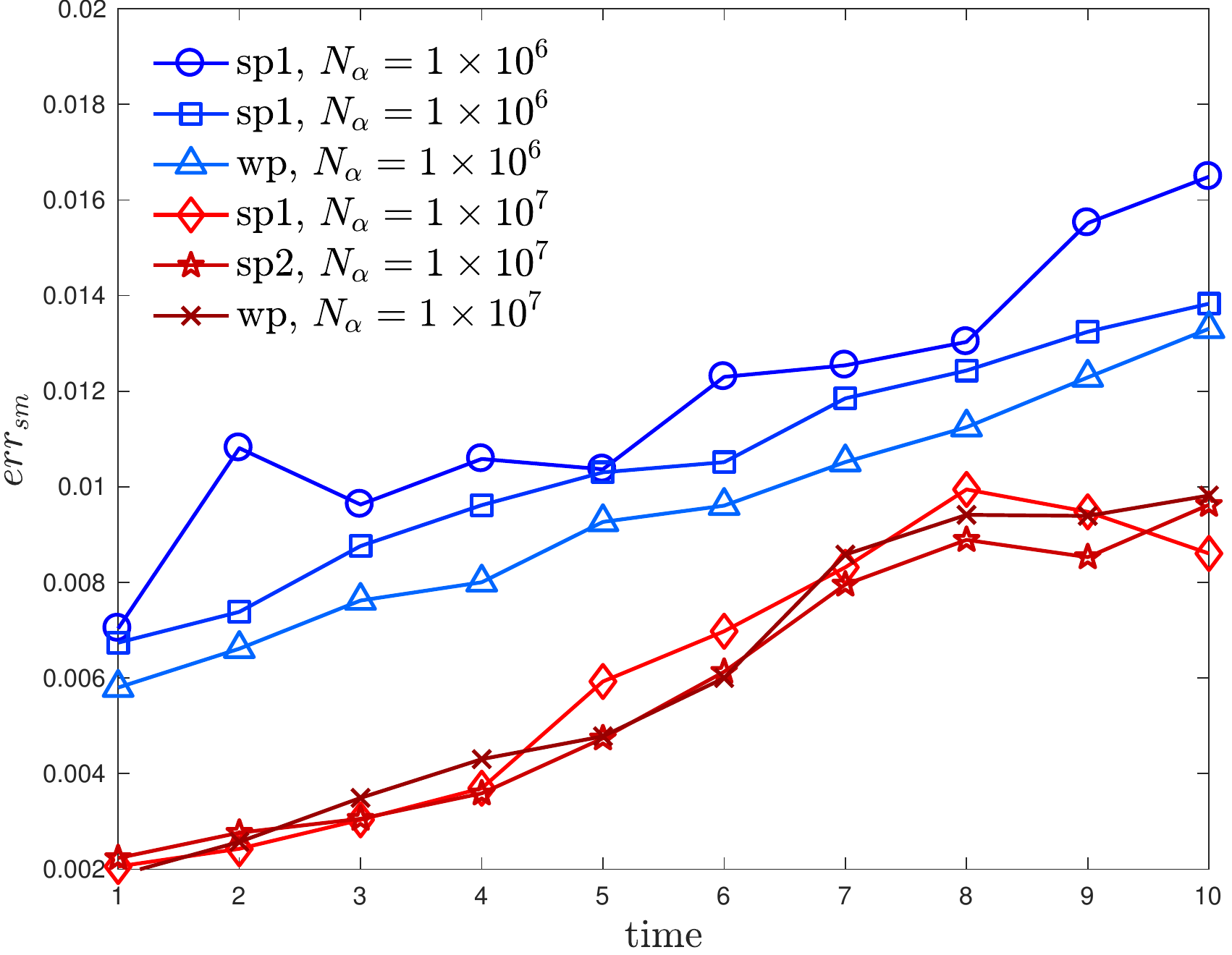}}{\includegraphics[width=0.49\textwidth,height=0.27\textwidth]{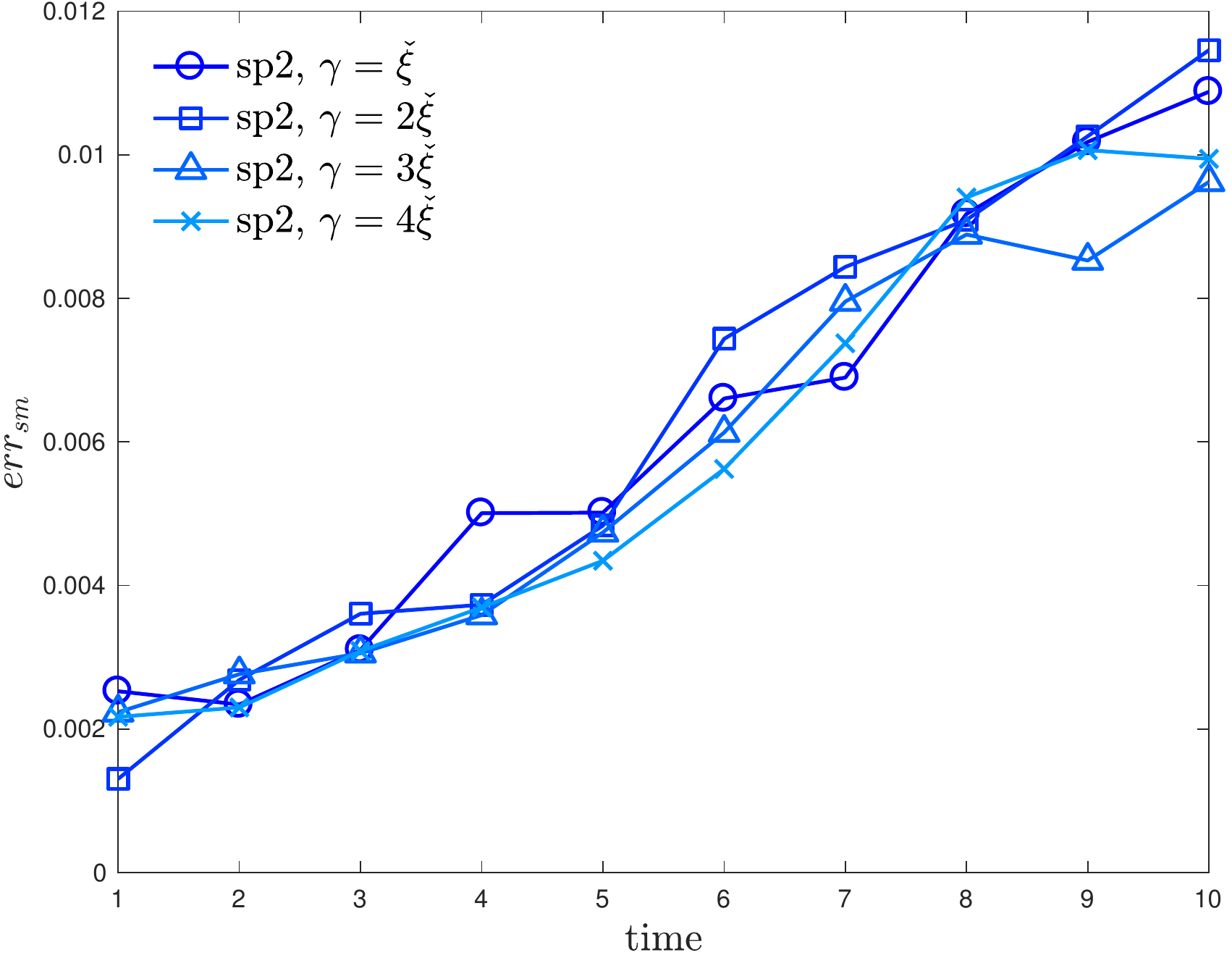}}}
\subfigure[$\textup{err}_{mm}$ under different $N_\alpha$ (left) or different $\gamma$ (right).]{{\includegraphics[width=0.49\textwidth,height=0.27\textwidth]{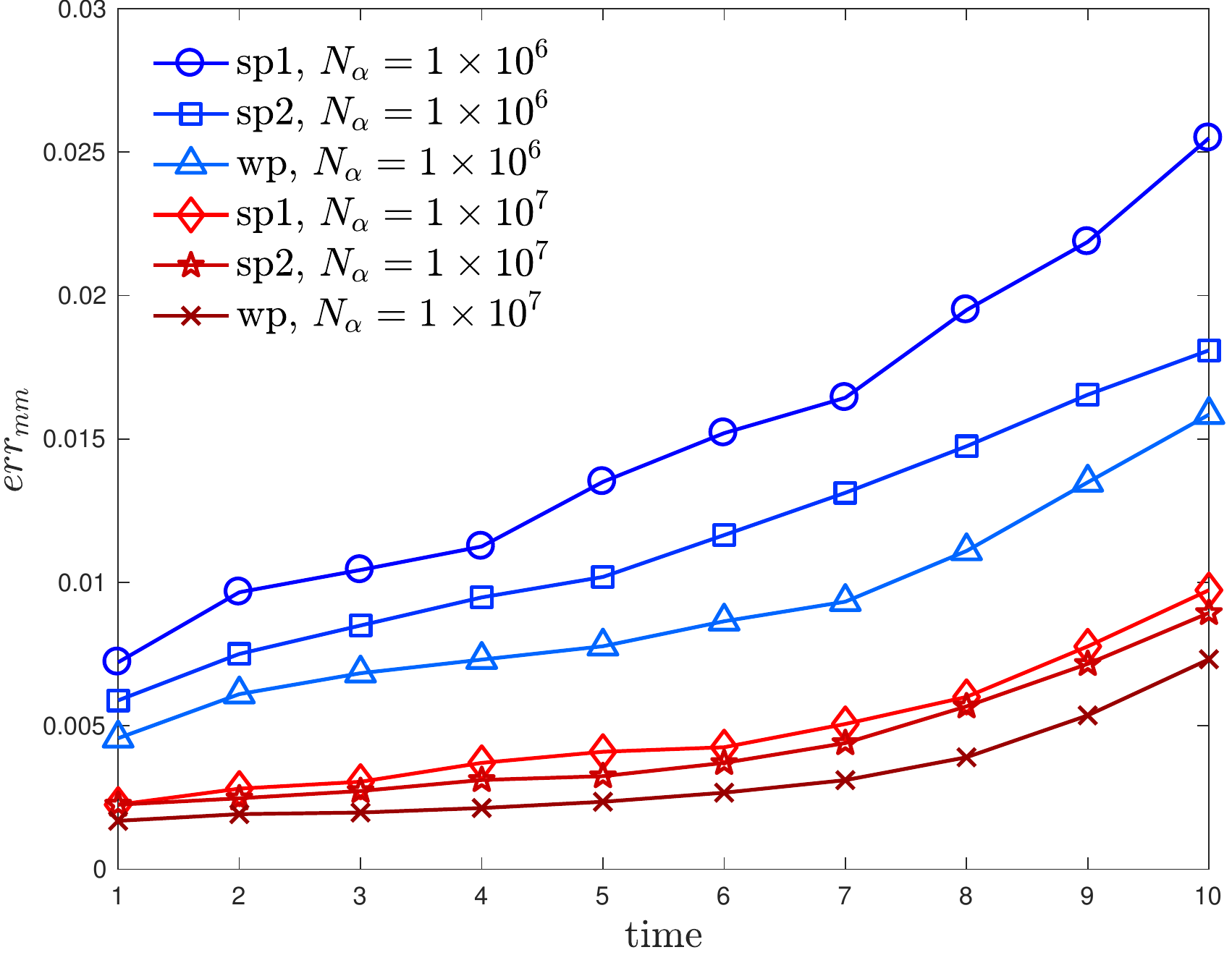}}
{\includegraphics[width=0.49\textwidth,height=0.27\textwidth]{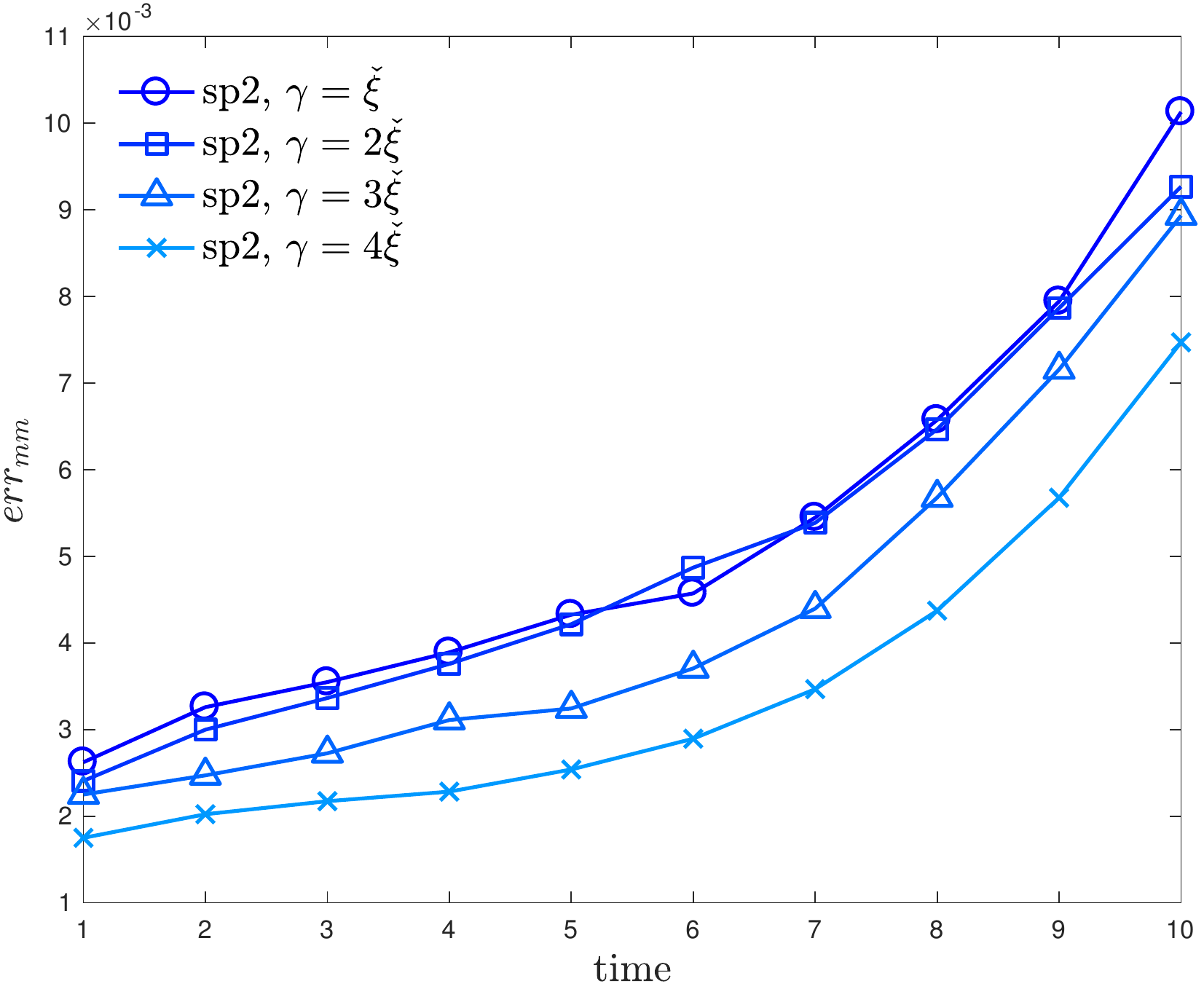}}}
\caption{\small The 2D Gaussian scattering: Comparison among \textbf{sp2}, \textbf{sp1} and \textbf{wp}. \textbf{sp2} is slightly less accurate than the \textbf{wp} due to the random noises induced by the bootstrap filtering, but it is still better than $\textbf{sp1}$. $\textbf{sp2}$ allows a reduction in variance as $\gamma$ increases. Here we set the auxiliary function $\gamma = 3\check{\xi}$ and $T_A = 1$fs.
}
\label{bootstrap}
\end{figure}

\subsection{A comparison among signed-particle implementations}
\label{sec:comparison}

After a thorough study of \textbf{sp1} and $\textbf{sp2}$, we would like to make a comparison among our proposed strategies and the existing ones \textbf{sp0} and \textbf{RC}. 
As we have already demonstrated in Section \ref{sec:gamma}, the major advantage of WBRW, as well as \textbf{sp0-I}, over \textbf{sp0} is that they are able to alleviate the strict restriction on the time step $\Delta t$.  To clarify it, we simulate the 2D Gaussian scattering under the potential
\begin{equation}
V(x) = 0.6\me^{-{(x-30)^2}/{8}},
\end{equation}
with the average position of the initial data $x_0 = -15$, $N_{\alpha} = 1 \times 10^7$, $T_A = 1$fs under five kinds of $\Delta t$: $0.01, 0.1, 0.25, 0.5, 1$fs. The numerical solutions are compared with those produced by \textbf{sp1} and \textbf{sp2} with $\gamma = 1.5\check{\xi}\approx 0.8706$. The implementations of \textbf{sp0} and \textbf{RC} follow the procedures in the articles \cite{SellierNedjalkovDimov2014} and \cite{MuscatoWagner2016}, respectively.  
\begin{figure}[!h]
\subfigure[Convergence with respect to $\Delta t$ (left: \textbf{sp0} and \textbf{sp0-I}, right:\textbf{RC}).]{
{\includegraphics[width=0.49\textwidth,height=0.27\textwidth]{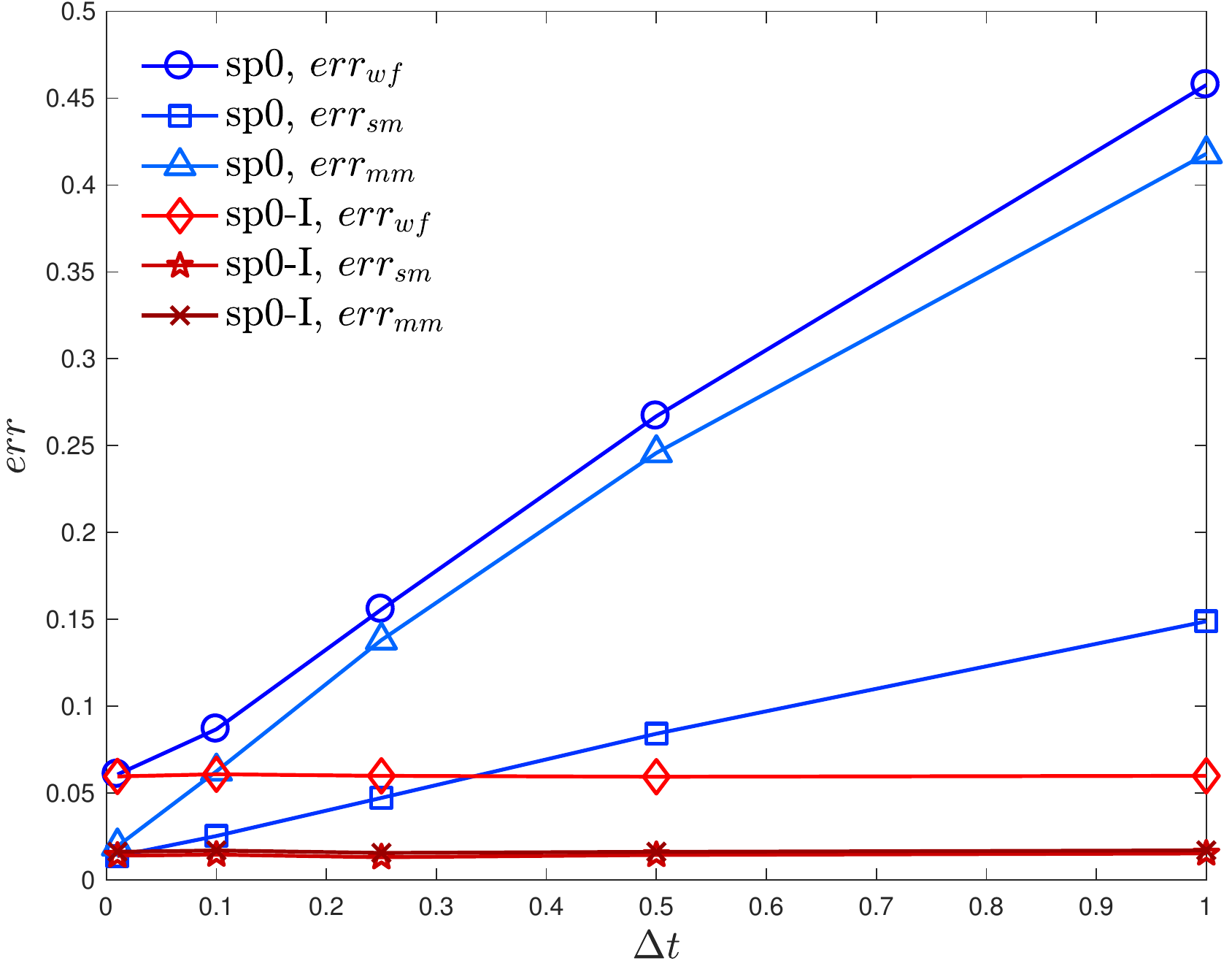}}{\includegraphics[width=0.49\textwidth,height=0.27\textwidth]{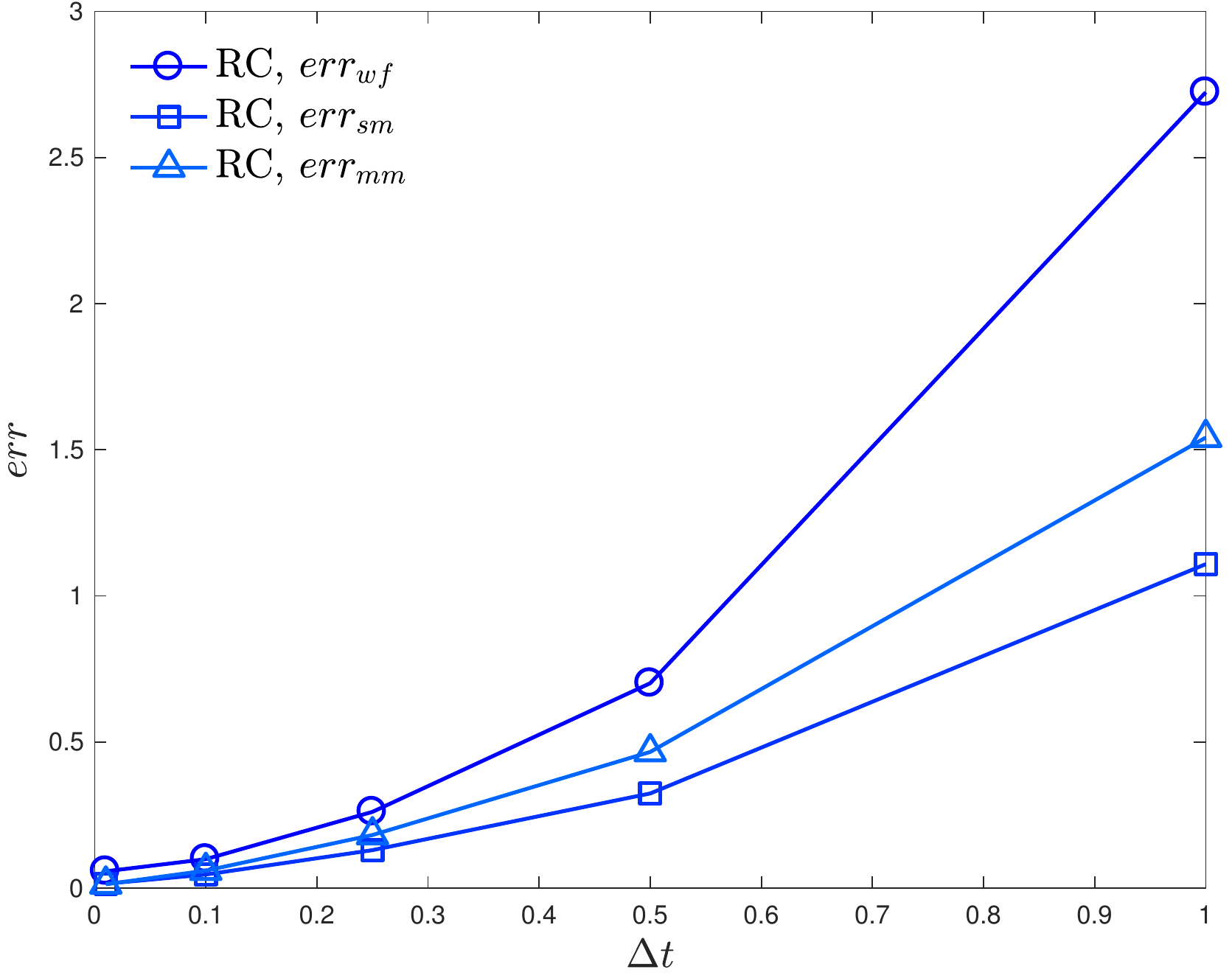}}}
\subfigure[$\textup{err}_{wf}$ under different $\Delta t = 0.01$fs (left) and $\Delta t = 0.1$fs (right).]{
{\includegraphics[width=0.49\textwidth,height=0.27\textwidth]{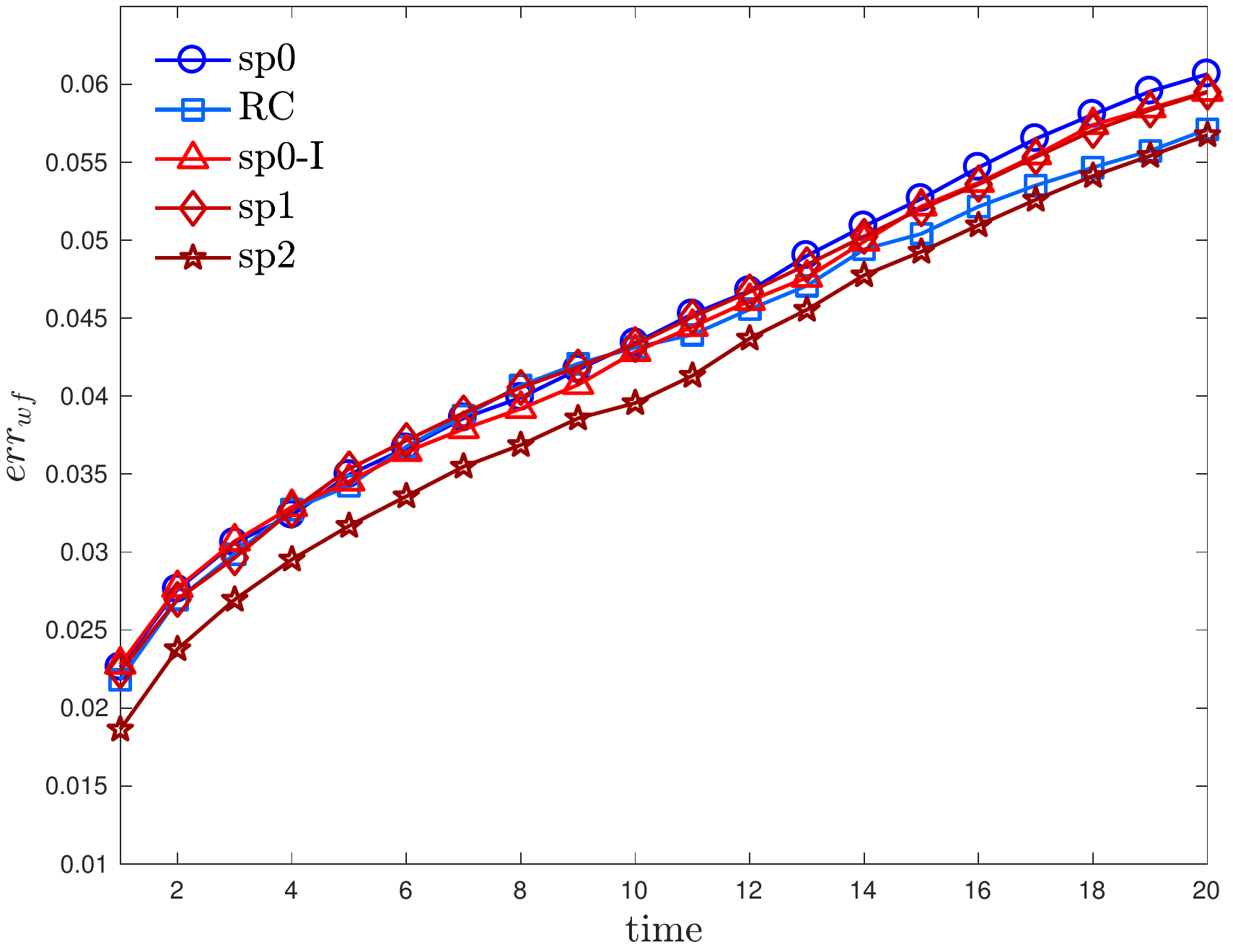}}
{\includegraphics[width=0.49\textwidth,height=0.27\textwidth]{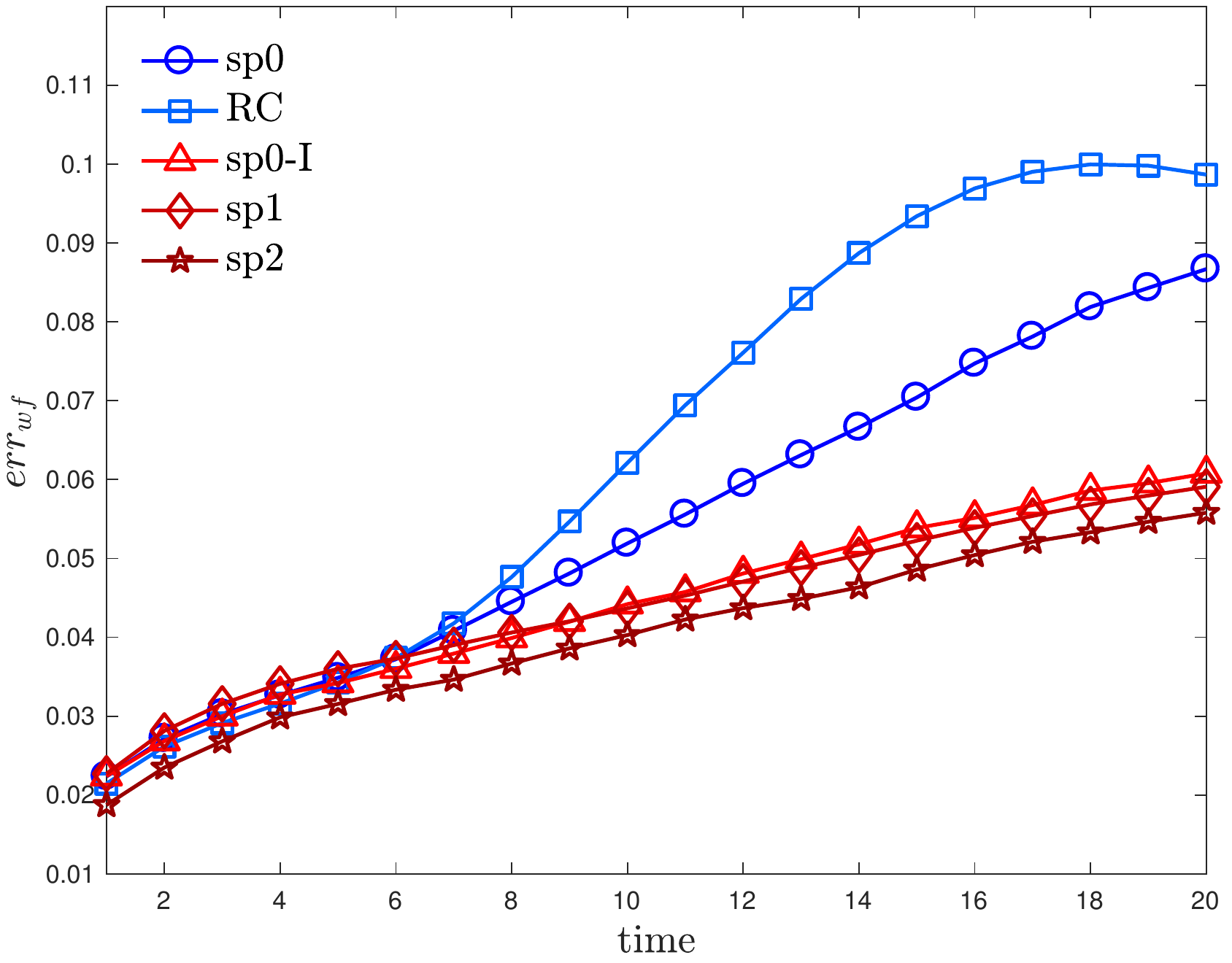}}}
\subfigure[$\textup{err}_{sm}$ under different $\Delta t = 0.01$fs (left) and $\Delta t = 0.1$fs (right).]{
{\includegraphics[width=0.49\textwidth,height=0.27\textwidth]{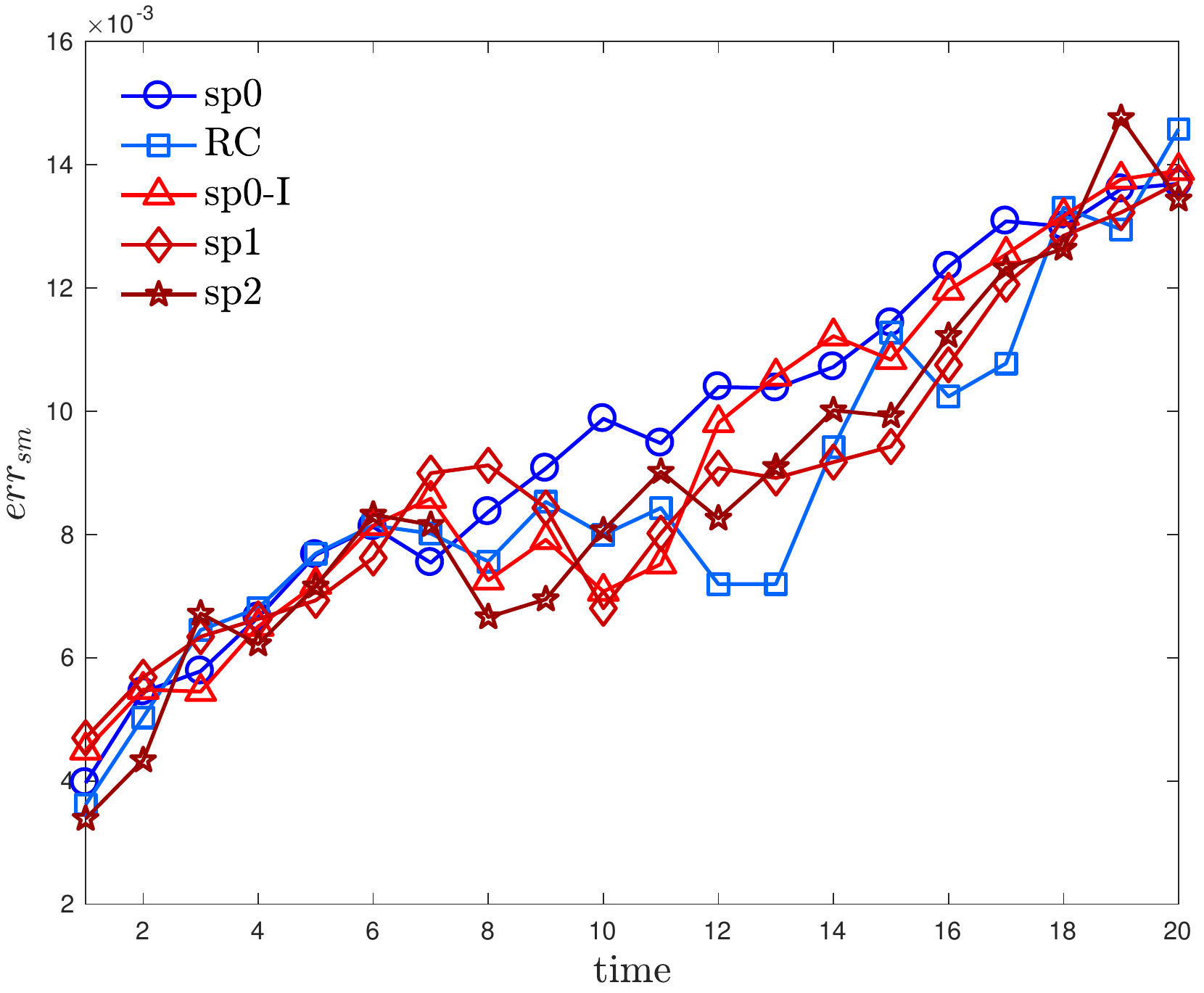}}{\includegraphics[width=0.49\textwidth,height=0.27\textwidth]{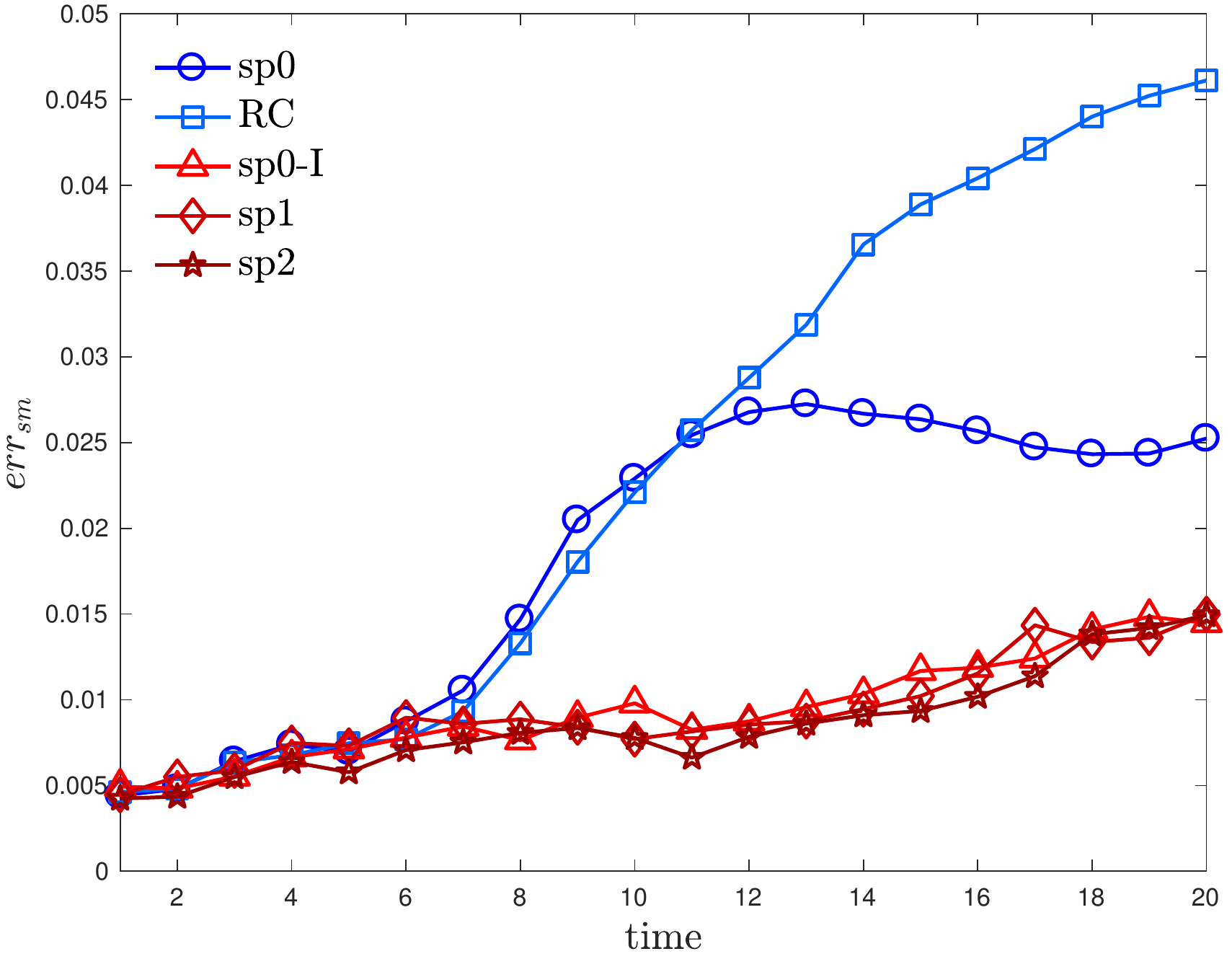}}}
\subfigure[$\textup{err}_{mm}$ under different $\Delta t = 0.01$fs (left) and $\Delta t = 0.1$fs (right).]{
{\includegraphics[width=0.49\textwidth,height=0.27\textwidth]{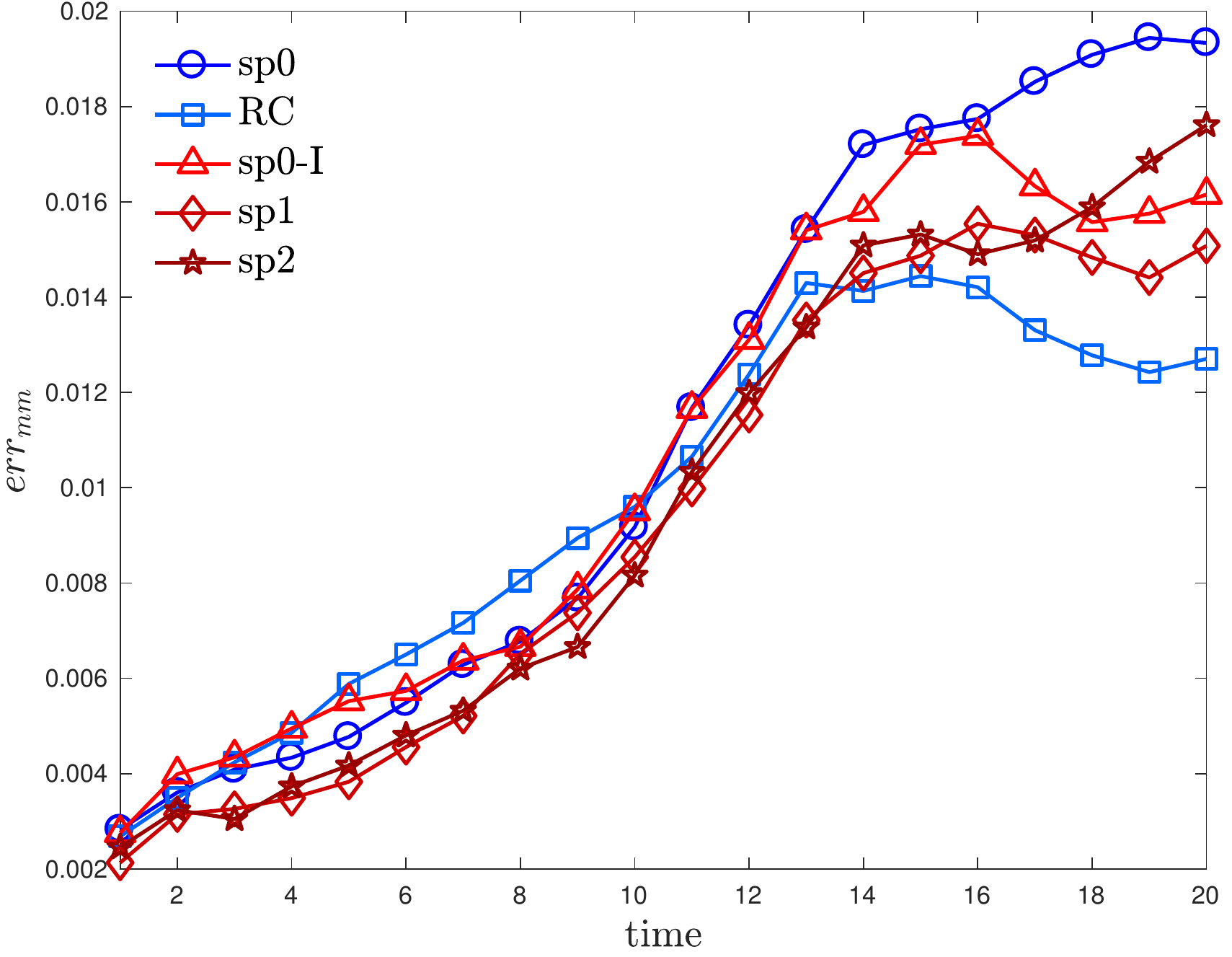}}
{\includegraphics[width=0.49\textwidth,height=0.27\textwidth]{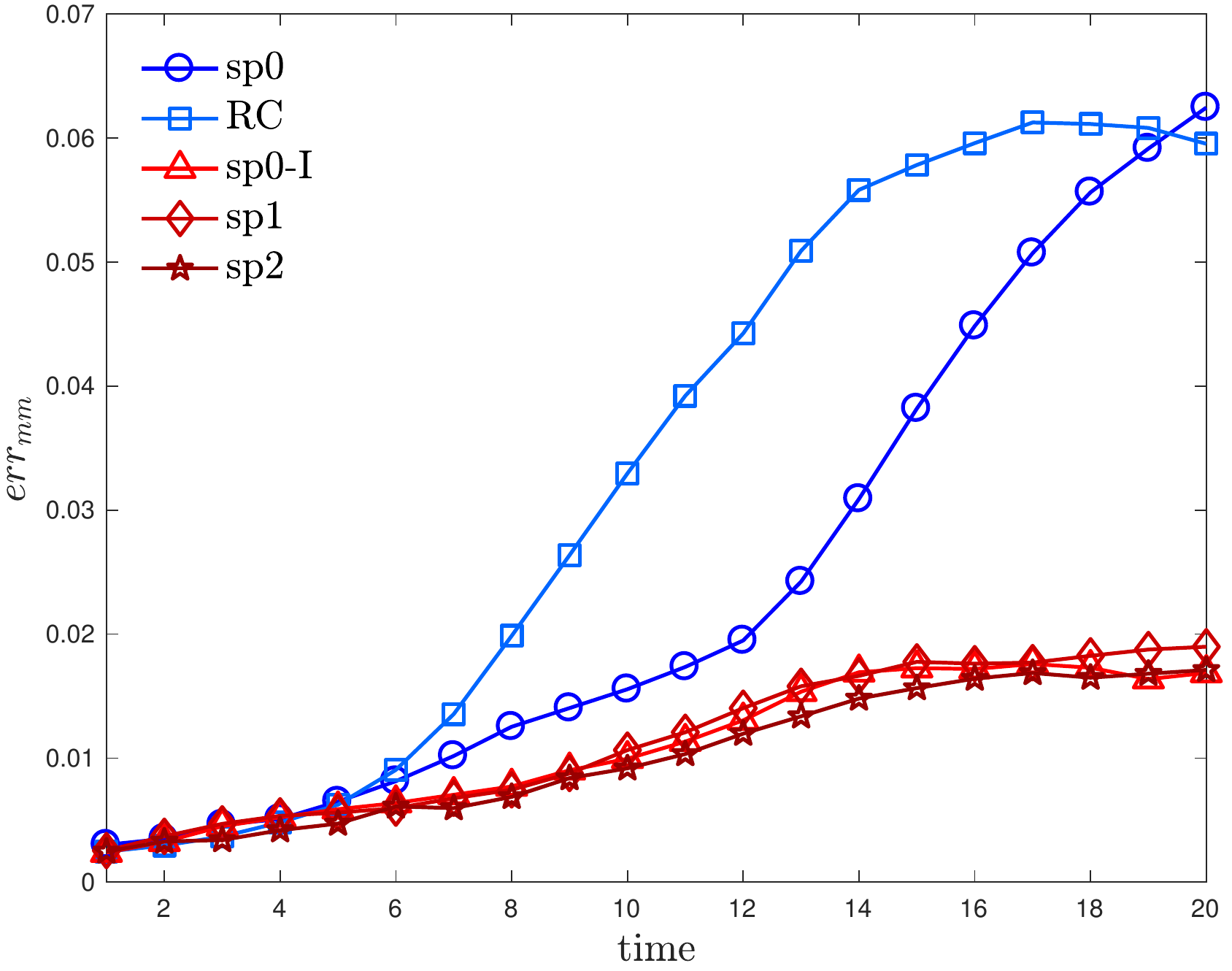}}}
\caption{\small The 2D Gaussian scattering: Comparison among the signed-particle implementations \textbf{sp0}, \textbf{sp0-I}, \textbf{sp1}, \textbf{sp2} and \textbf{RC}. Both \textbf{sp0} and \textbf{RC} require the time step $\Delta t$ sufficiently small to ensure the accuracy, whereas \textbf{sp0-I}, \textbf{sp1} and \textbf{sp2} alleviate such restriction. $\textbf{sp2}$ produce the most accurate results as it seizes the variance reduction property of \textbf{wp}. Here we set the auxiliary function $\gamma = 1.5\check{\xi}$ and $T_A = $1fs.}
\label{G_sp}
\end{figure}

Several facts below observed from the numerical results are presented in Fig.~\ref{G_sp}.  
\begin{description}

\item[(1)] On one hand, the accuracy of $\textbf{sp0}$ and $\textbf{RC}$ depends on the choice of time step $\Delta t$, due to the errors induced by the time discretizations. For $\textbf{sp0}$, the errors seem linearly dependent on $\Delta t$, while for $\textbf{RC}$, a superlinear convergence is observed, although the order is slightly deviated from the theoretical one $\mathcal{O}((\Delta t)^2)$. By contrast, the accuracy of  \textbf{sp0-I}, \textbf{sp1} or \textbf{sp2} is clearly independent of $\Delta t$. Actually, such observation also manifests the accuracy of the self-scattering technique because it is equivalent to $\textbf{sp1}$ in 1D single-body problem. 

\item[(2)]  On the other hand, too small $\Delta t$ leads to a dramatic increase in computational time. 
For instance, the running time of \textbf{sp0} is 86.00s for $\Delta t = 1$fs, 425.53s for $\Delta t = 0.1$fs, and 3738.89s for $\Delta t = 0.01$fs on the computational platform: Intel(R) Core(TM) i7-6700 CPU (3.40GHz, 8MB Cache, 4 Cores, 8 Threads) and 16GB Memory (we use one thread for each task). 

\item[(3)] A simple improvement \textbf{sp0-I} is able to alleviate the restriction on the choice of $\Delta t$, as the bias term is properly tackled. Moreover, the accuracy of \textbf{sp0-I} is comparable to that of \textbf{sp1}.

\item[(4)] A comparison shows that \textbf{sp2} achieves the best accuracy, as it seizes the variance reduction property of \textbf{wp}. When $\Delta t$ is sufficiently small, \textbf{RC} also performs quite well as the choice of majorant $\hat{V}_W(\bx, \bk) = \mathcal{F}[V](2\bk)$ avoids the calculation of $\int_{2\mathcal{K}} V^\pm_W(\bx, \bk) \D \bk$, thereby getting rid of the errors induced by numerical integrations and interpolations.


\end{description}


\begin{figure}
  \subfigure[$t=10$fs (left: potential, middle: ASM, right: \textbf{sp1}).]{
  {\includegraphics[width=0.32\textwidth,height=0.26\textwidth]{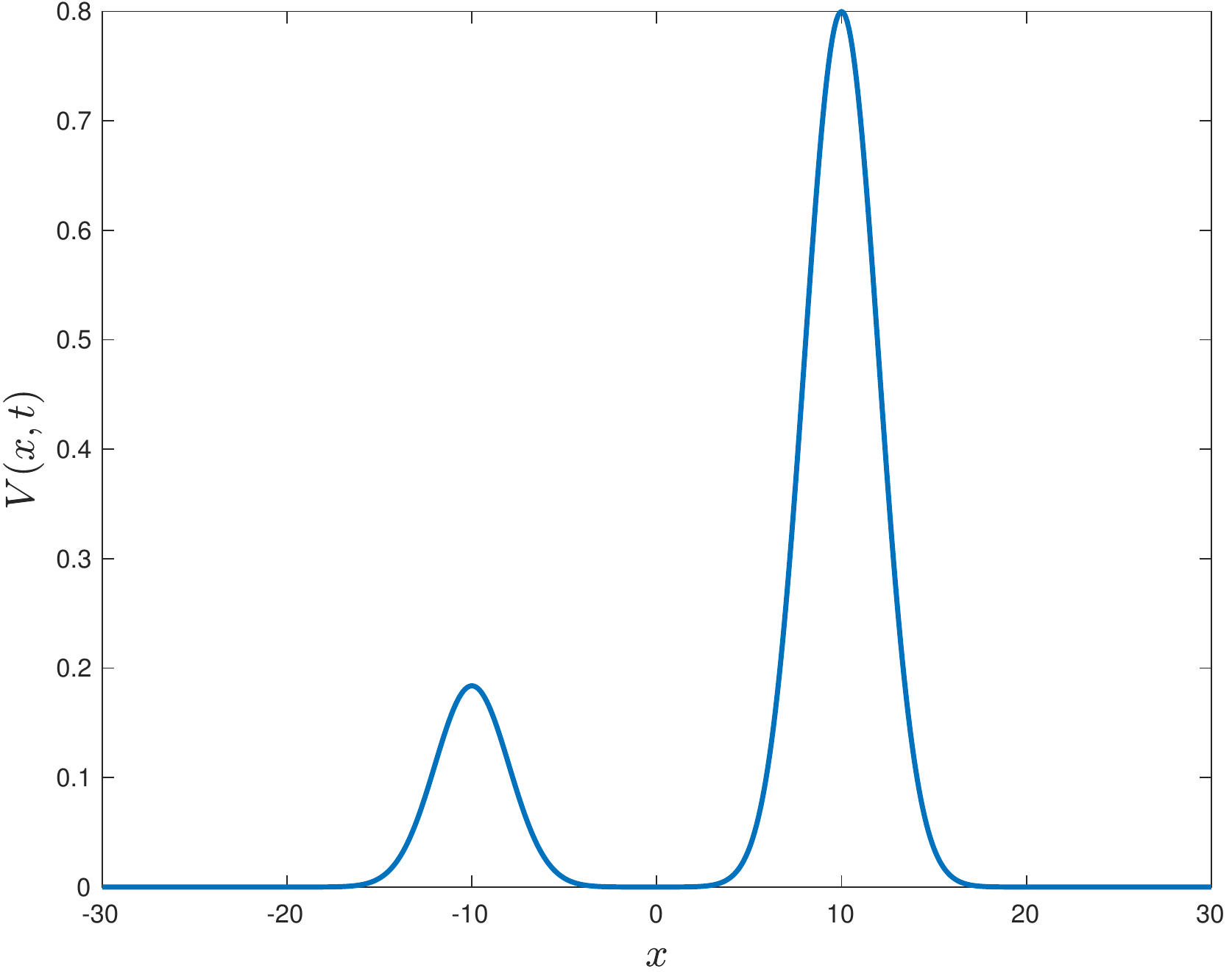}}
  {\includegraphics[width=0.32\textwidth,height=0.26\textwidth]{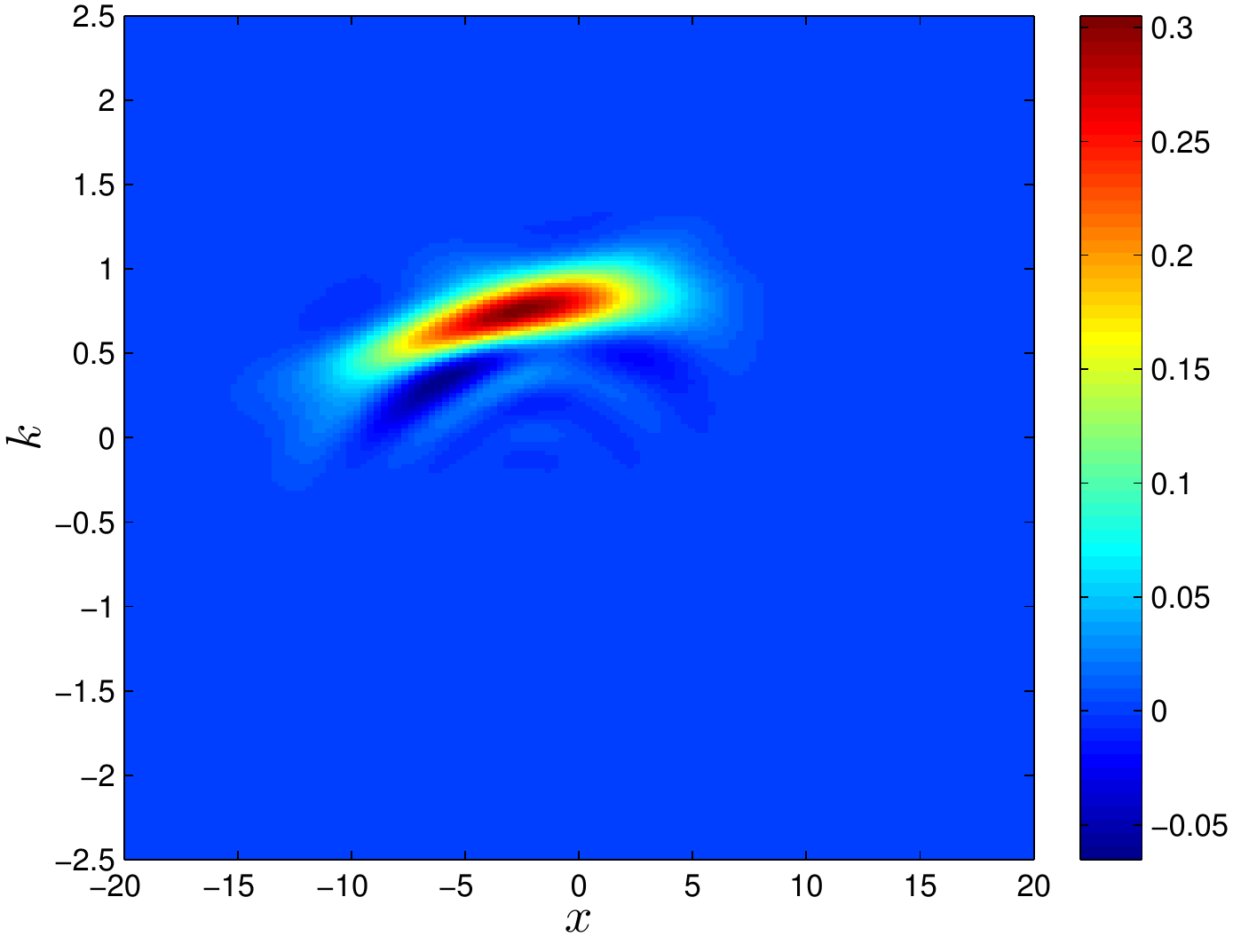}}
  {\includegraphics[width=0.32\textwidth,height=0.26\textwidth]{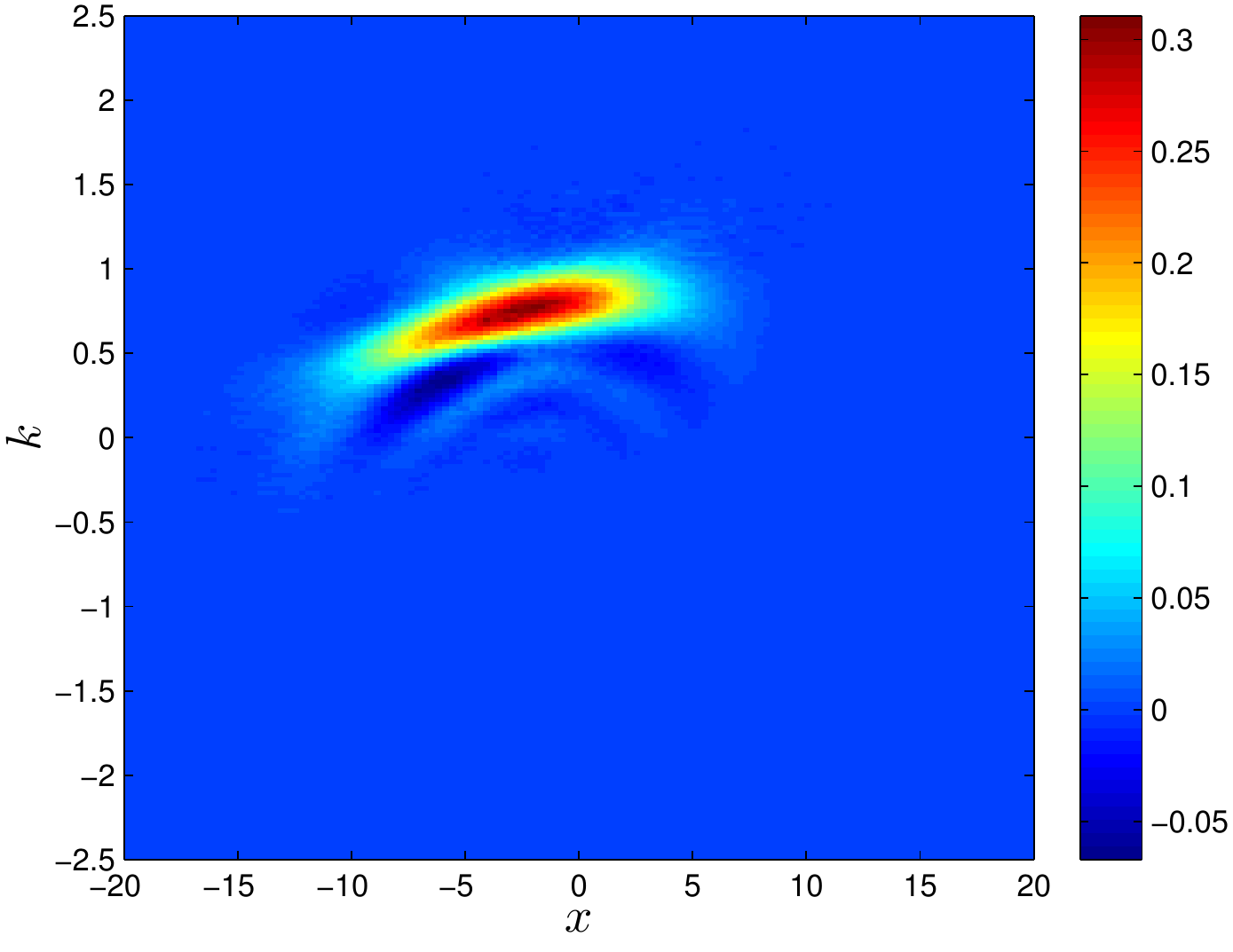}}}

   \subfigure[$t=30$fs (left: potential, middle: ASM, right: \textbf{sp1}).]{
   {\includegraphics[width=0.32\textwidth,height=0.26\textwidth]{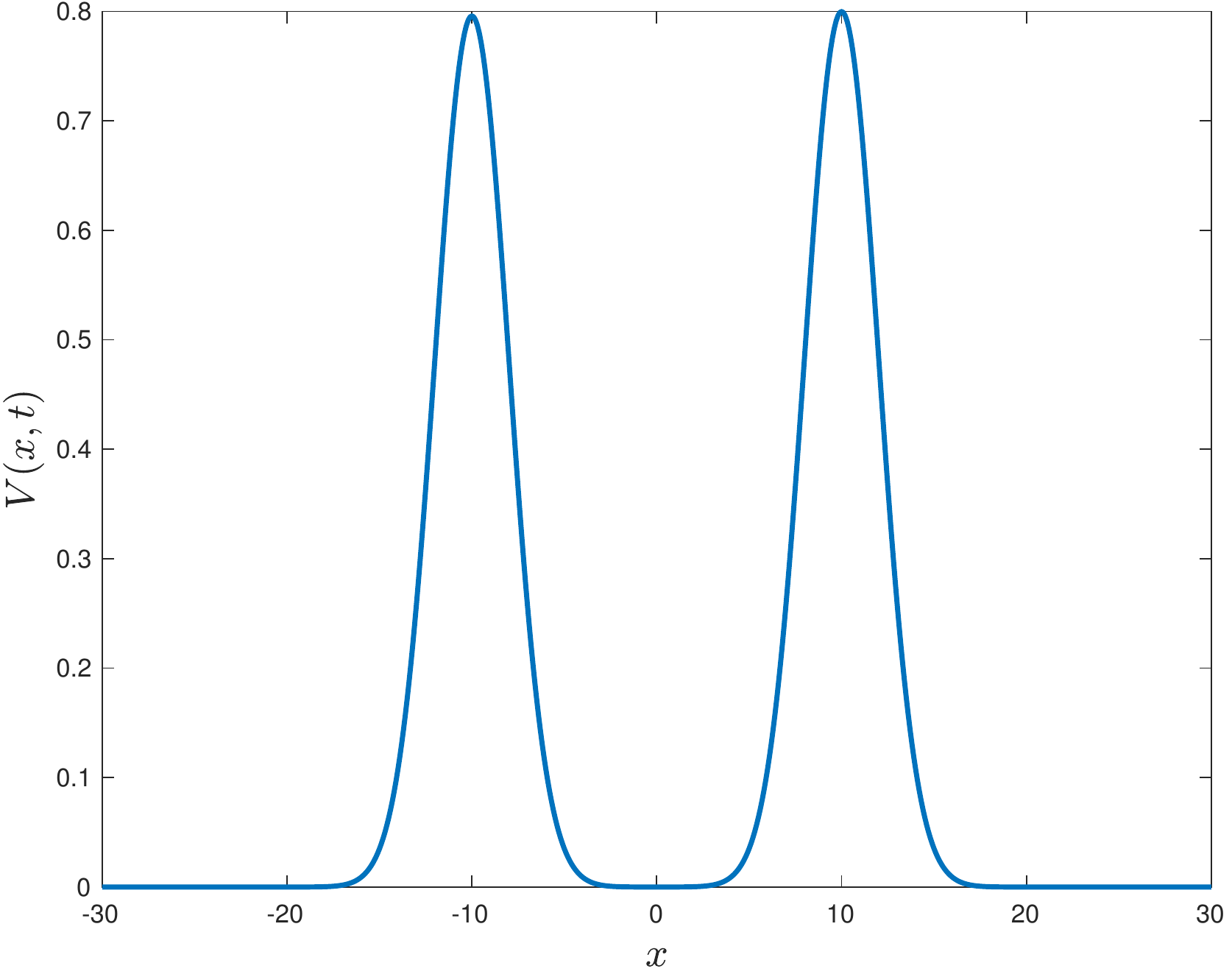}}
   {\includegraphics[width=0.32\textwidth,height=0.26\textwidth]{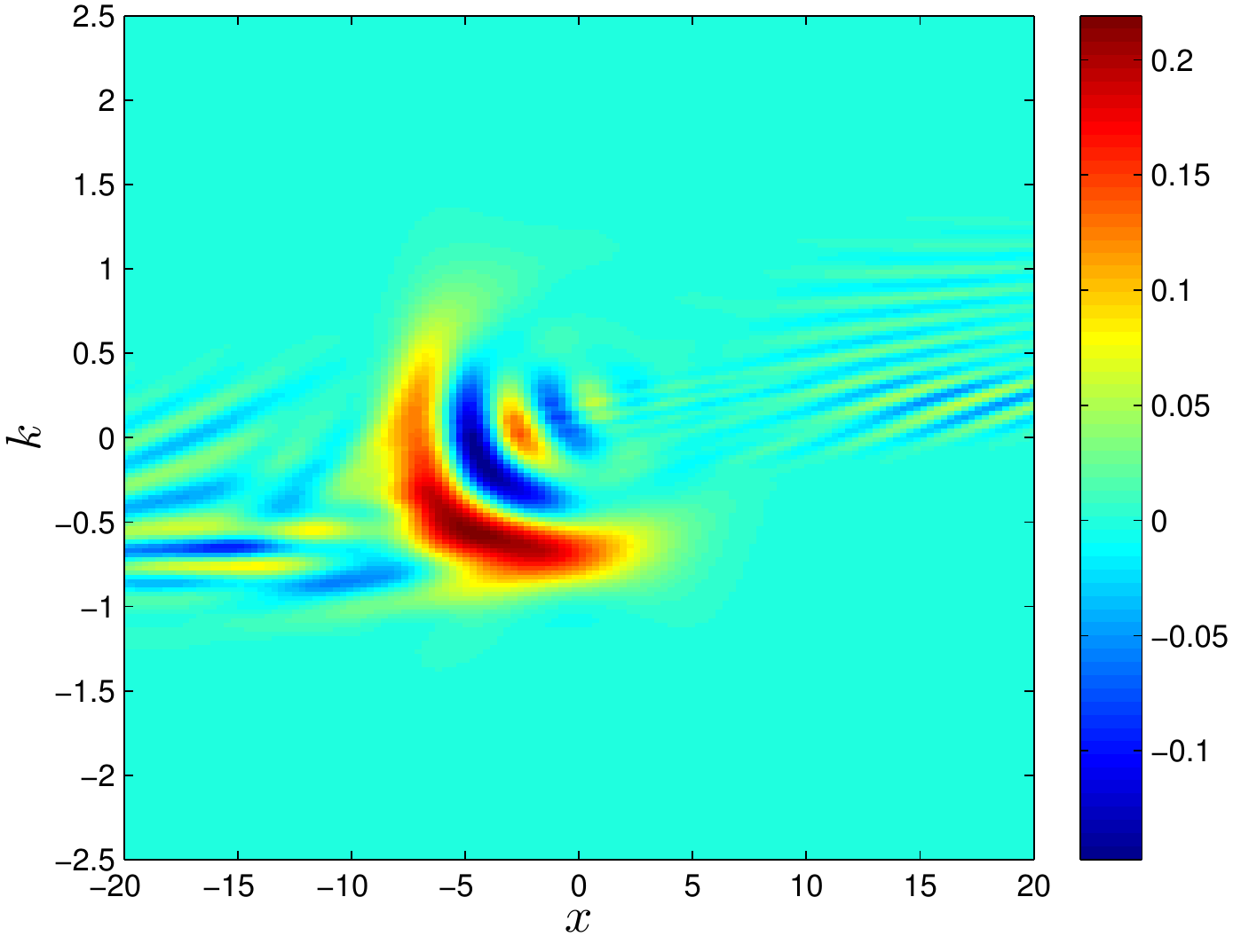}}
   {\includegraphics[width=0.32\textwidth,height=0.26\textwidth]{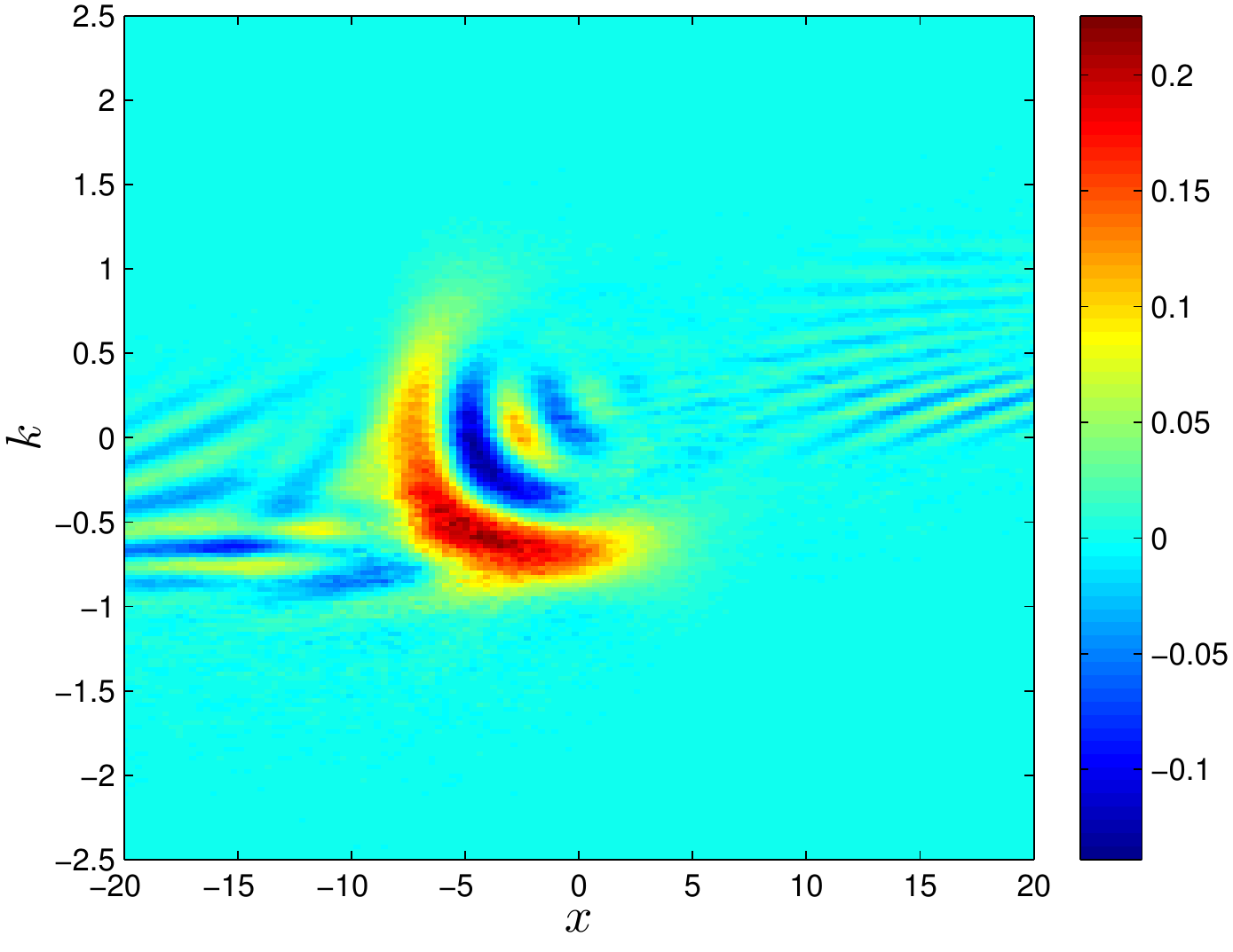}}}

   \subfigure[$t=45$fs (left: potential, middle: ASM, right: \textbf{sp1}).]{
   {\includegraphics[width=0.32\textwidth,height=0.26\textwidth]{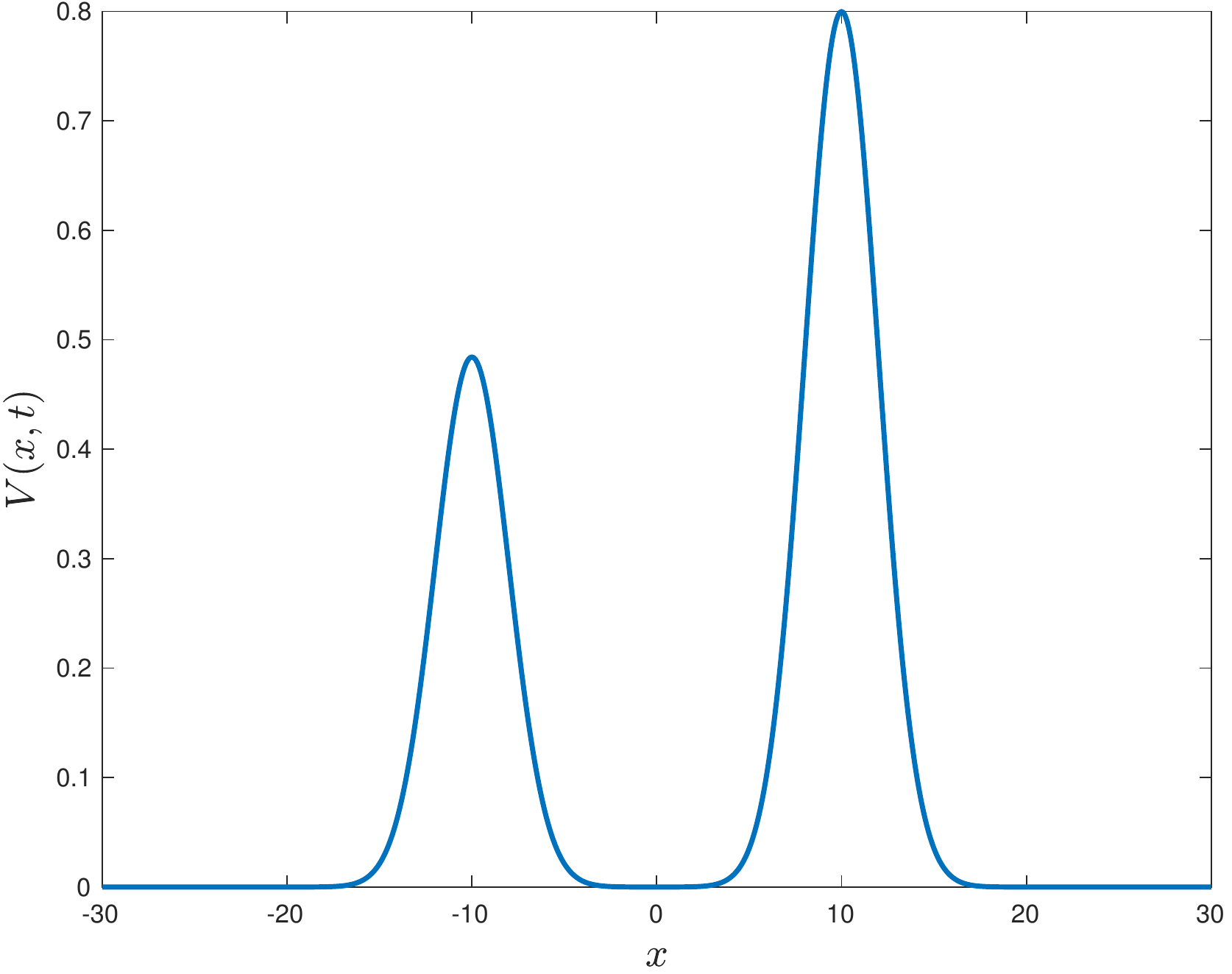}}
   {\includegraphics[width=0.32\textwidth,height=0.26\textwidth]{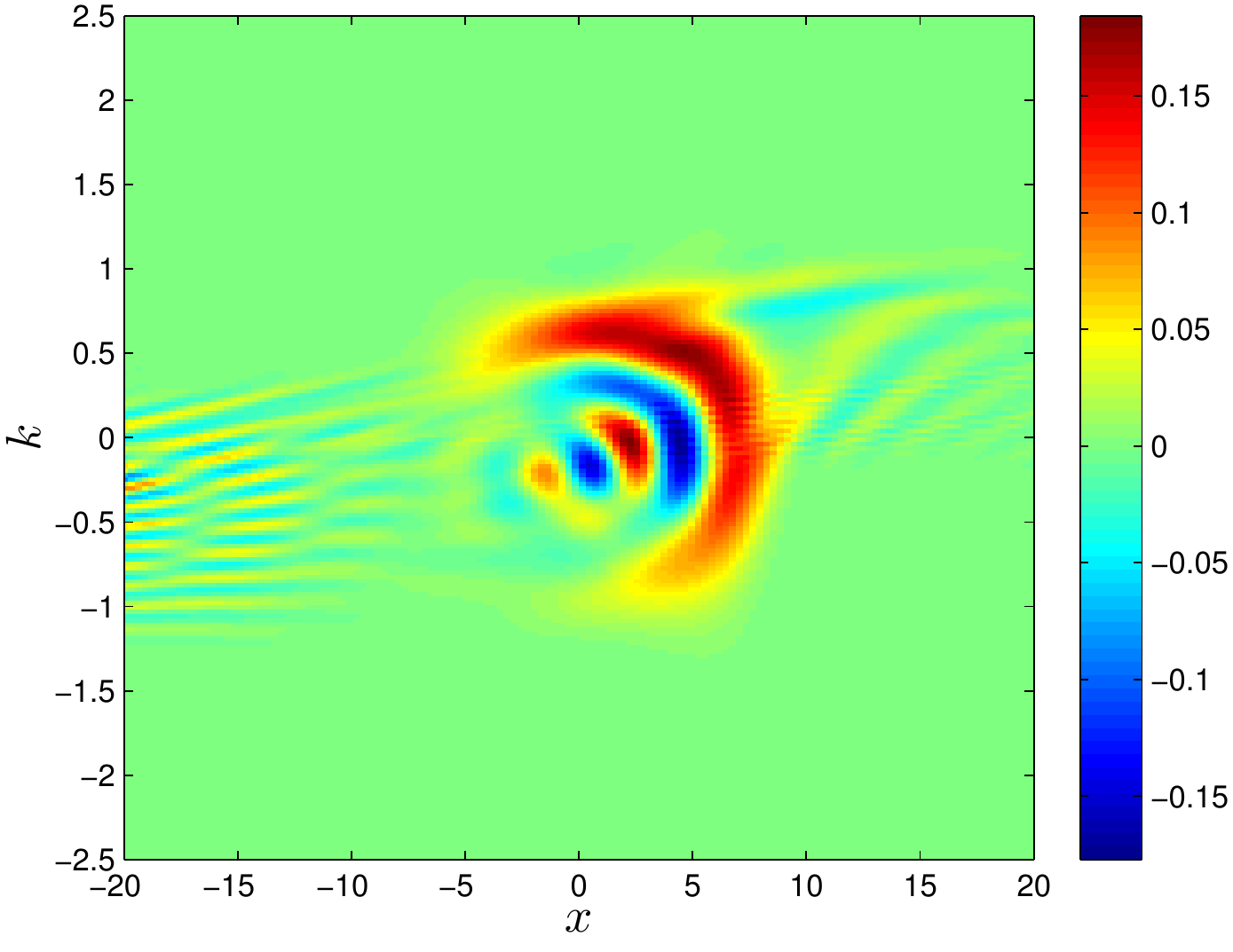}}
   {\includegraphics[width=0.32\textwidth,height=0.26\textwidth]{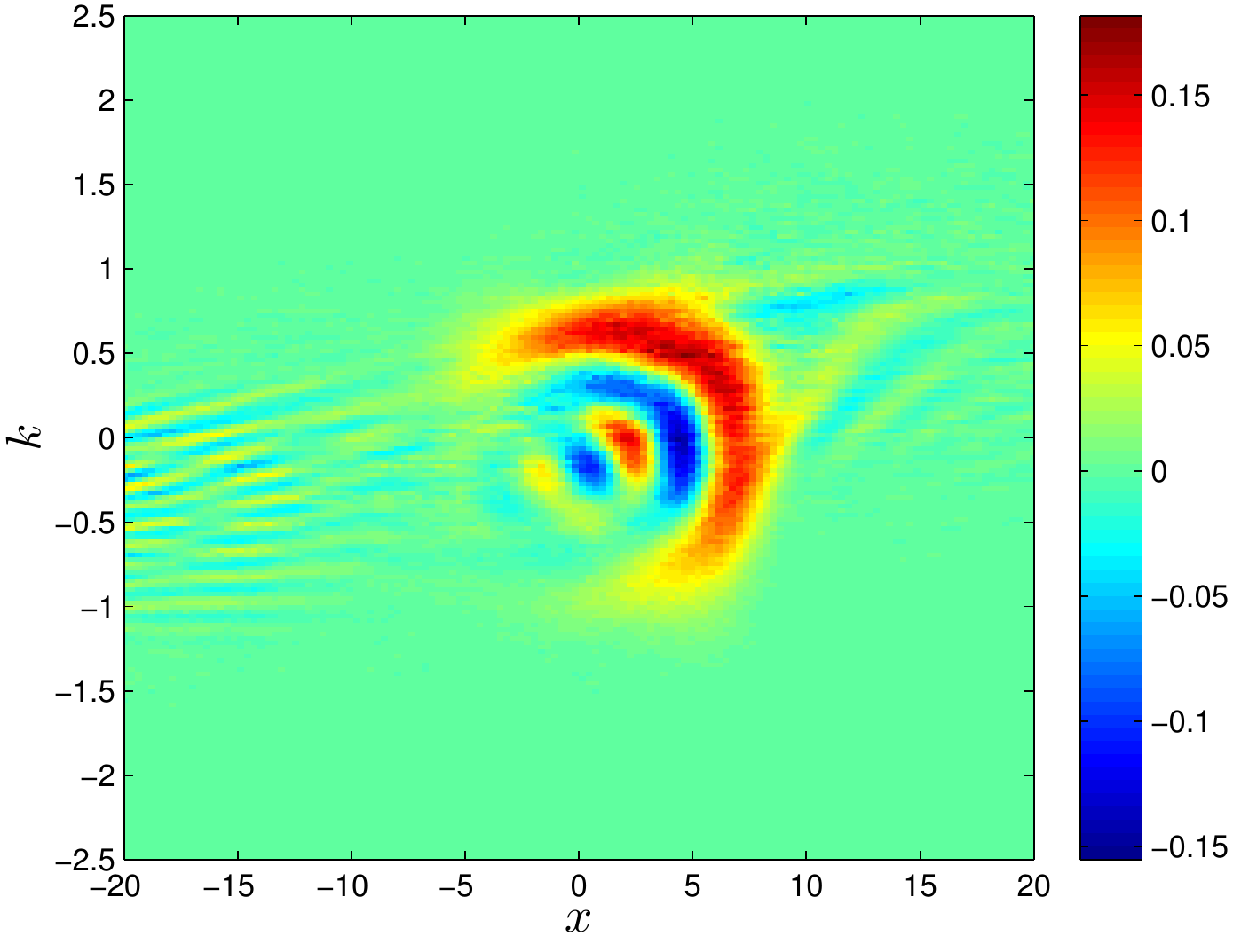}}}

   \subfigure[$t=60$fs (left: potential, middle: ASM, right: \textbf{sp1}).]{
   {\includegraphics[width=0.32\textwidth,height=0.26\textwidth]{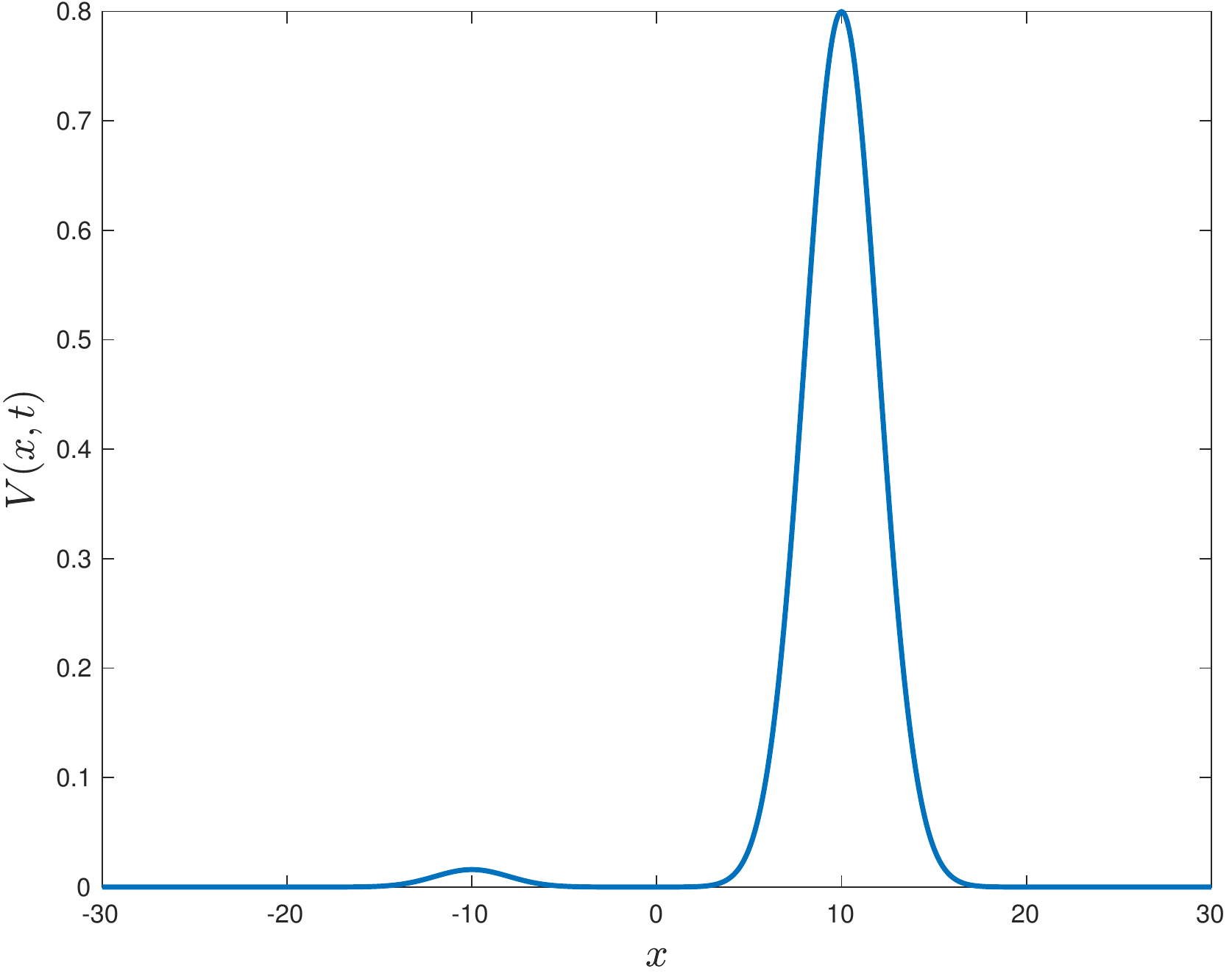}}
   {\includegraphics[width=0.32\textwidth,height=0.26\textwidth]{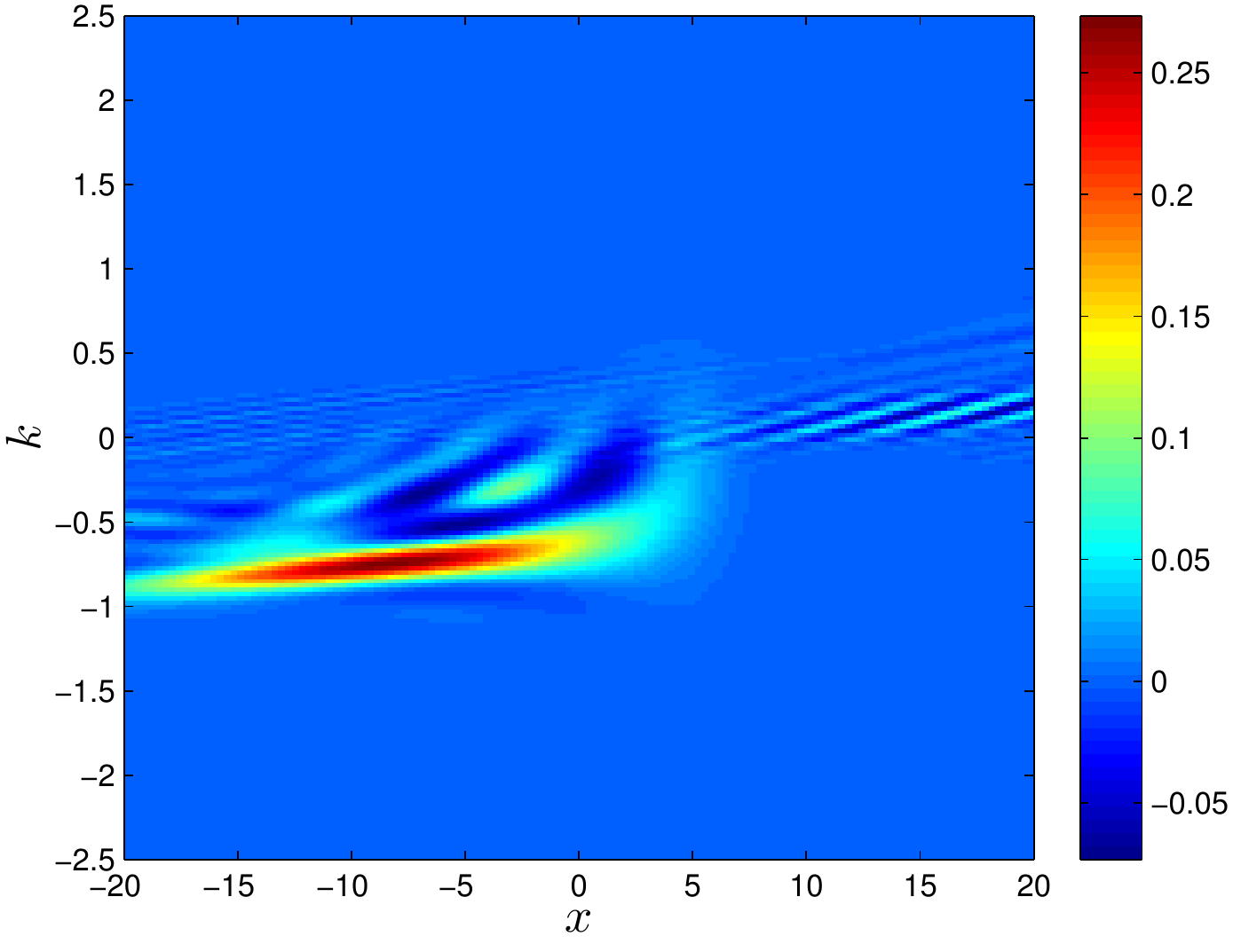}}
   {\includegraphics[width=0.32\textwidth,height=0.26\textwidth]{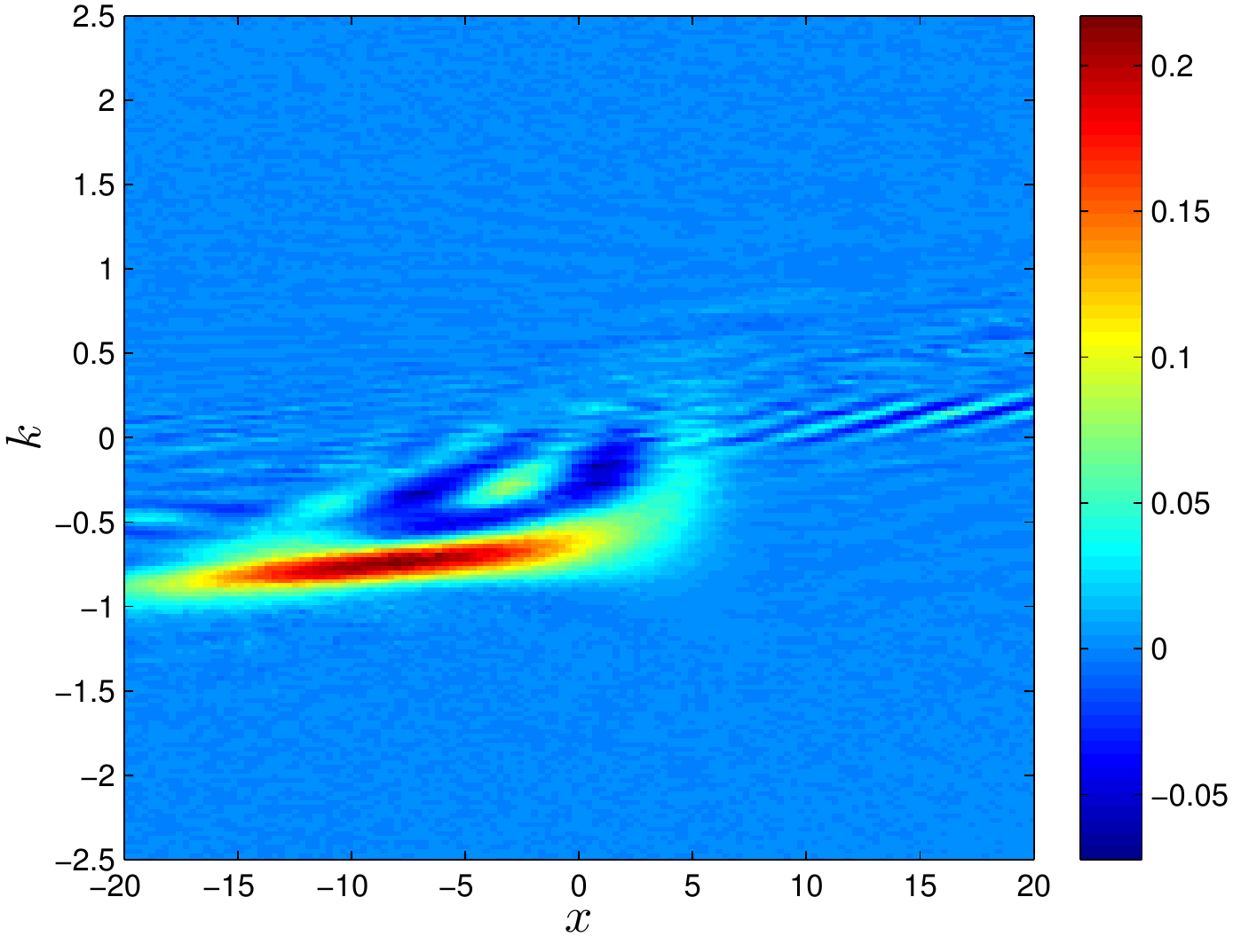}}}
 
  \caption{\small {The 2D Gaussian scattering under a time-varying barrier: the Wigner function at different instants $t=10, 30, 45, 60$fs. Different scattering phenomena are observed as the height of the left barrier changes in time.}}
  \label{fig.td_G}
\end{figure}

\begin{figure}[ht]
\subfigure[Convergence rate with respect to $N_\alpha$.]{\includegraphics[width=0.49\textwidth,height=0.27\textwidth]{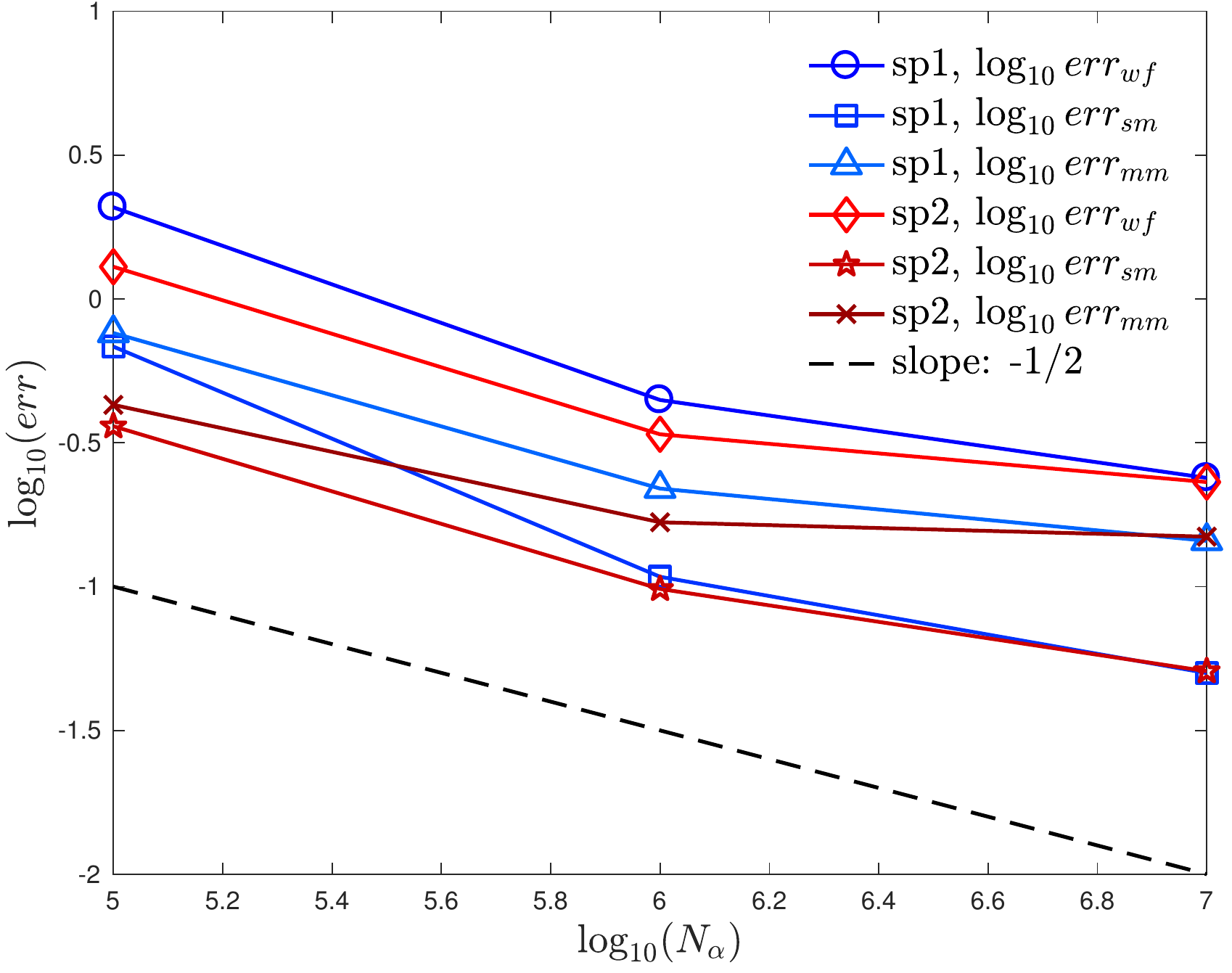}}
\subfigure[$\textup{err}_{wf}$ under different $\gamma$.]{\includegraphics[width=0.49\textwidth,height=0.27\textwidth]{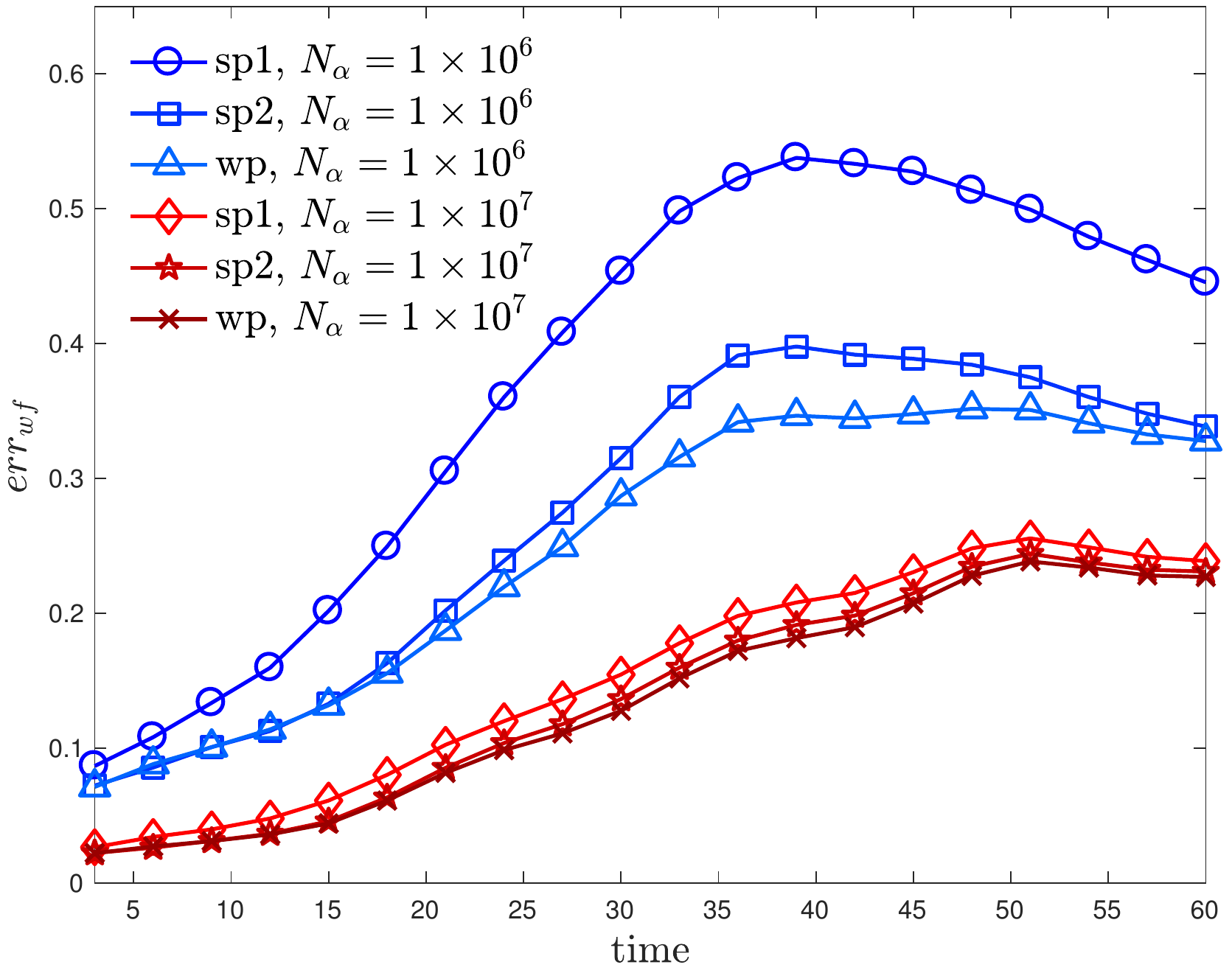}}
\subfigure[$\textup{err}_{sm}$ under different $\gamma$.]{\includegraphics[width=0.49\textwidth,height=0.27\textwidth]{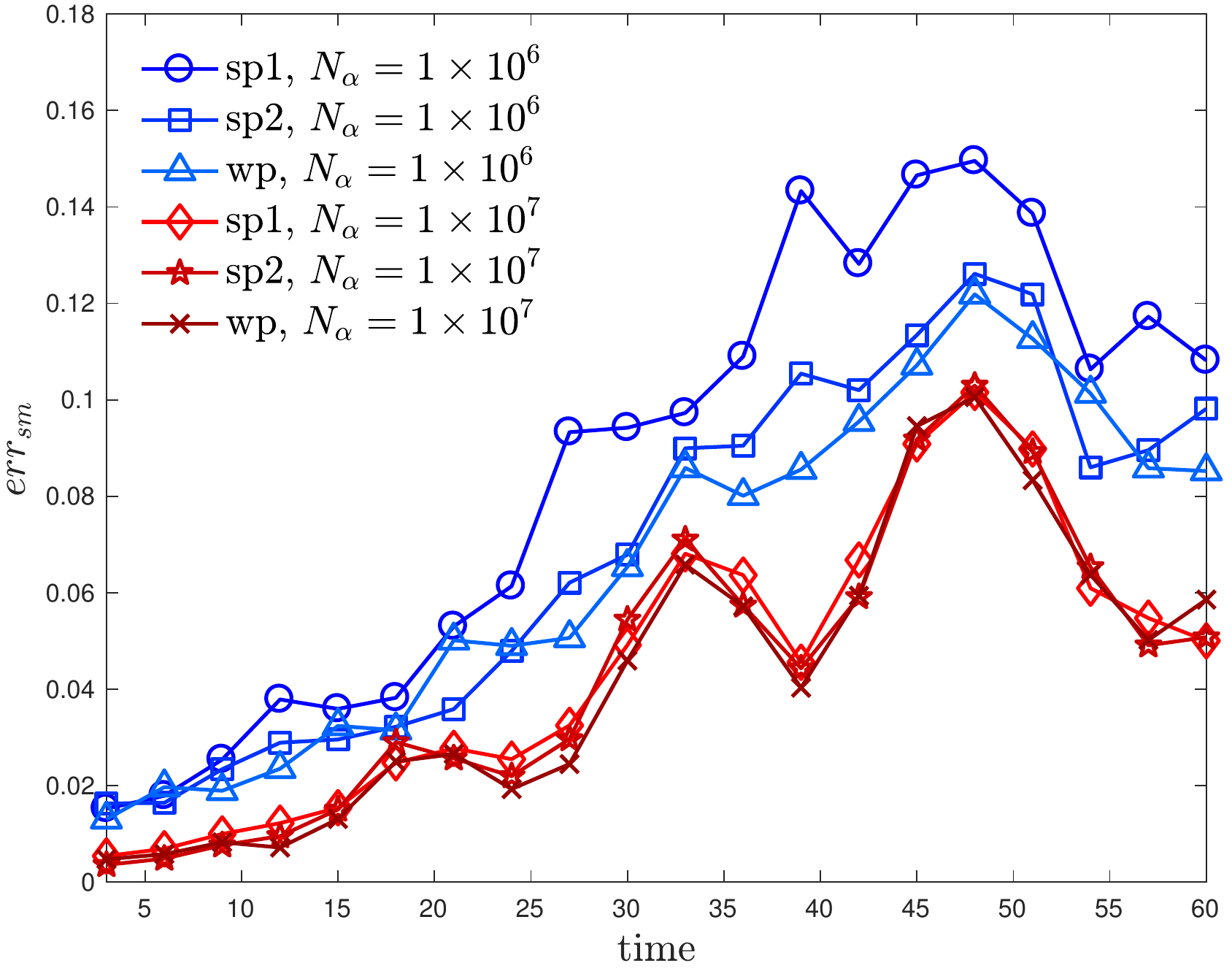}}
\subfigure[$\textup{err}_{mm}$ under different $\gamma$.]{\includegraphics[width=0.49\textwidth,height=0.27\textwidth]{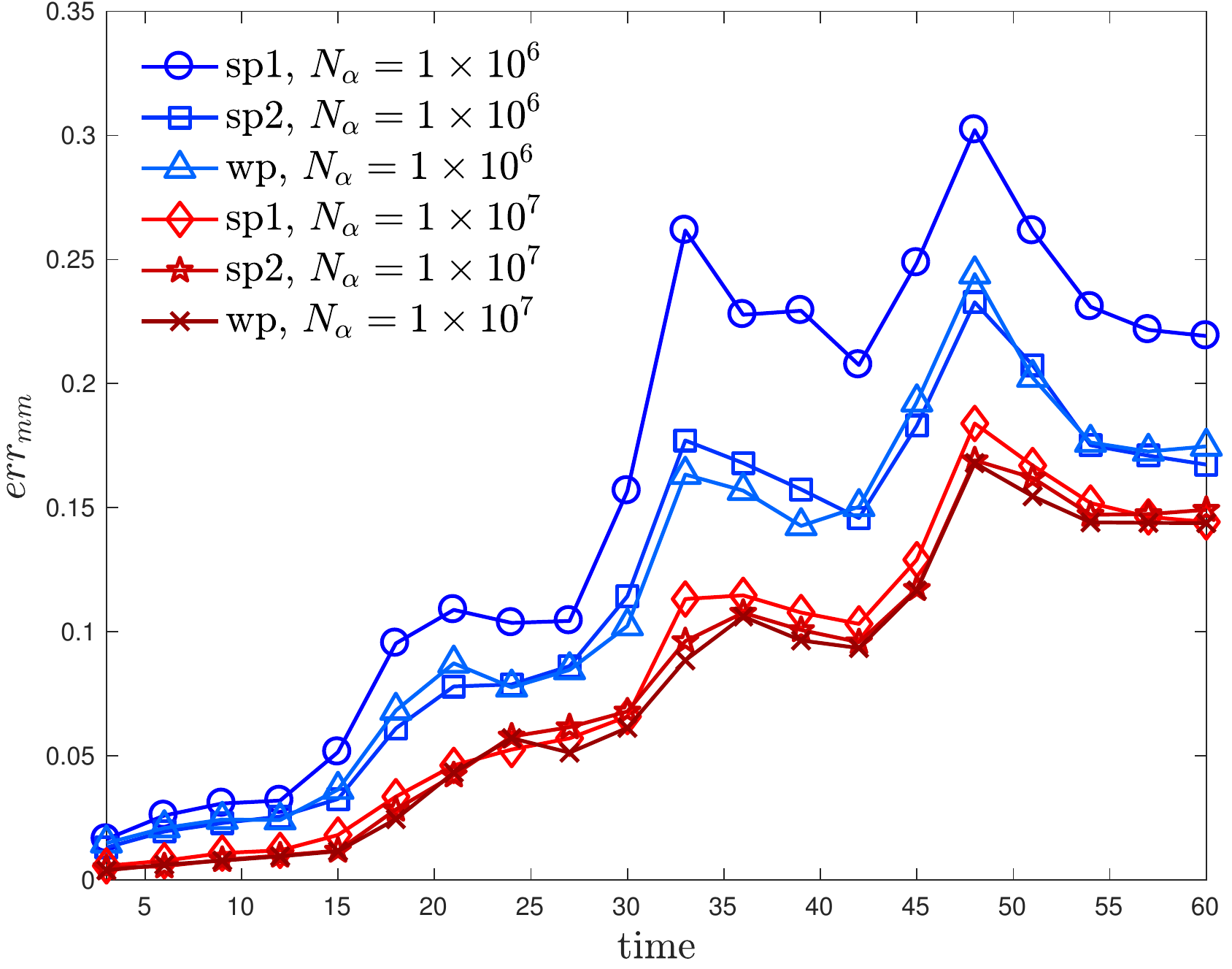}}
\caption{\small  {The 2D Gaussian scattering under a time-varying barrier: Comparison among \textbf{sp1}, \textbf{sp2} and \textbf{wp} and the convergence rate of \textbf{sp1} and \textbf{sp2} with respect to the sample size $N_\alpha$. Similar to the static potential case, \textbf{sp2} is more accurate than \textbf{sp1}, but less than \textbf{wp}. Here we set the auxiliary function $\gamma = 1$ and $T_A = 1$fs. }
}
\label{td_G_convergence}
\end{figure}

\subsection{Time-dependent potential }
Finally, we try to study a 2D Gaussian scattering under a time-dependent double barrier potential: 
\begin{equation}\label{eq:tDV}
V(x, t) = 0.8(0.5+0.5\cos(\pi + 0.1 t)) \me^{-{(x-10)^2}/{8}} + 0.8\me^{-{(x+10)^2}/{8}}.
\end{equation} 
Other parameters are identical to those in the above static Gaussian scattering tests, except that the initial position $x_0$ is reset to be $-15$ and a constant auxiliary function $\gamma= 1$ is adopted. The left barrier serves as a switch as its height varies harmonically in time and different scattering phenomena can be clearly observed for a relatively long final time $t_{fin} = 60$fs, as displayed in Fig.~\ref{fig.td_G}. The initial wavepacket touches the left barrier and penetrates it without any difficulty because the barrier height is very small at $t=10$fs. Then it begins to interact with the right barrier and is almost reflected back. As the height of the left barrier nearly attains its maximal value at $t=30$fs, it forces the wavepacket to be reflected again and confined in $[-10, 10]$. After the third reflection, the wavepacket travels across the left barrier at $t=60$fs as the barrier height decreases. From Fig.~\ref{fig.td_G}, it can be readily verified that both deterministic (i.e. ASM) and stochastic methods (for brevity, we only plot the results produced by \textbf{sp1}) succeed in capturing the scattering phenomenon and the fine oscillating structure. A detailed comparison is made between \textbf{sp1}, \textbf{sp2} and \textbf{wp} under the time-dependent potential \eqref{eq:tDV} and the results are shown in Fig.~\ref{td_G_convergence} where we investigate the convergence rate of both \textbf{sp1} and \textbf{sp2} with respect to the sample size $N_\alpha = 10^5, 10^6, 10^7$.  
It is clearly observed there that $\textbf{sp2}$ is more accurate than \textbf{sp1}, but the difference diminishes as $N_\alpha$ increases, and the convergence order is less than $-1/2$ due to other errors involved. In fact, all these observations are quite similar to those already mentioned in Section \ref{sec.bootstrap_filtering} for the time-independent potential.

\section{Conclusions and discussions}
\label{sec:con}

In this paper, we propose several efficient strategies to realize the signed-particle implementation of the Wigner branching random walk (WBRW).  Based on a unified theoretical framework of signed-particle implementations, we can interpret the multiplicative functional $\xi/\gamma$ as either the probability to generate new particles or make some replicas of offsprings, yielding \textbf{sp1} and \textbf{sp0-I}, respectively. In order to further reduce the variance, we employ a bootstrap filtering in the weighted-particle version, yielding \textbf{sp2}.
The numerical analysis on the bootstrap filtering, as well as the errors induced by resampling based on the uniform histogram, is given. Through detailed performance evaluations, we have shown the accuracy and efficiency of two proposed strategies, distinguished the differences of various approaches and uncovered the following facts.

\begin{description}

\item[(1)] WBRW implementation is more advantageous over the original signed-particle Wigner Monte Carlo method (\textbf{sp0}) and $\textbf{RC}$ in time stepping, since it alleviates the restriction on the time step $\Delta t$. Other strategies, such as the self-scattering technique and our improved version \textbf{sp0-I}, are also able to avoid the time discretization errors. In addition, the accuracy of \textbf{sp2} and \textbf{wp} can be further improved by adjusting the auxiliary function $\gamma$.

\item[(2)] Increasing the sample size $N_\alpha$ will systematically improve the accuracy, but the convergence order has some deviations from the theoretical order of $-1/2$ due to the deterministic errors induced by the resampling. 

\item[(3)] Both the accuracy and the growth of particle number of \textbf{sp1} is independent of $\gamma$. On the contrary, the \textbf{wp} is a variance reduction method and the choice of $\gamma$ will give a systematic improvement on the accuracy with the order $\mathcal{O}(\gamma^{-2})$, at the cost of higher computational complexity.

\item[(4)] The resampling procedure is indispensible for the consideration of not only efficiency but also accuracy. It helps to suppress both the exponential growth of particle number and the accumulation of stochastic errors. The resampling based on the uniform histogram performs quite well for $N_\alpha \ge N_{h}$. A balanced choice is $N_\alpha \approx 10 N_h$.  However, the efficiency of the resampling based on a uniform histogram is still undermined due to the curse of dimensionality. 

\item[(5)] \textbf{sp2} captures the merits of both \textbf{sp1} and \textbf{wp}, but introduces additional stochastic errors. It allows a variance reduction by increasing $\gamma$, and the histogram can be stored and operated as integer-value matrices.

\item[(6)] It's straightforward to generalize \textbf{sp1}, \textbf{sp2} and \textbf{wp} to the problems with time-dependent potentials.

\end{description}

Towards a wider application of WBRW implementation of quantum mechanics, 
the resampling strategy should be improved to meet the challenge in higher dimensional problems, and some advanced statistical tools, such as the tree-based density estimation and the kernel density estimation \cite{bk:HastieTibshiraniFriedman2009}, may be considered. 

\section*{Acknowledgement}
This research was supported by grants from the National Natural Science Foundation of China (Nos.~11471025, 11421101). Y. X. is partially supported by The Elite Program of Computational and Applied Mathematics for PhD Candidates in Peking University. 

\appendix
\section*{Appendix}

\begin{proof}[Proof of the Lemma \ref{theorem_bootstrap_bound}]

According to the first rule in Alg.~\ref{res_bootstrapping},  
\begin{equation*}
\left| \frac{1}{N}\sum_{i=1}^{N} w_i \varphi(\bx_i, \bk_i) - \frac{1}{N} \sum_{i=1}^N \varphi(\tilde{\bx}_{i}, \tilde{\bk}_i) \right| =  \left| \frac{1}{N} \sum_{i=1}^N r_i \varphi(\bx_i, \bk_i) - \frac{1}{N}\sum_{j=1}^{N_r} \varphi(\tilde{\bx}^{\prime}_j, \tilde{\bk}^{\prime}_j ) \right|,
\end{equation*}
where $\mathcal{\tilde{S}} = \{ (\tilde{\bx}^{\prime}_j, \tilde{\bk}^{\prime}_j)\}_{j=1}^{N_r}$ constitutes a subset that is randomly chosen from $\mathcal{S}$, and $r_i = w_i - k_i$. We allow $r_i = 0$ for some $i$. 

Now we define $S_N$ and $\tilde{S}_N$ as
\begin{equation*}
S_N= \sum_{i=1}^{N} r_i \cdot \frac{\varphi(\bx_i, \bk_i)}{N} , \quad  \tilde{S}_N  = \sum_{j=1}^{N_r} \frac{\varphi (\tilde{\bx}^{\prime}_j, \tilde{\bk}^{\prime}_j)}{N}.
\end{equation*}
It's easy to verify that $\tilde{S}_N$ is unbiased ($\mathbb{E}\tilde{S}_N = S_N$) because 
\begin{equation*}
\begin{split}
\mathbb{E} \left[\varphi (\tilde{\bx}^{\prime}_j, \tilde{\bk}^{\prime}_j)\right]  &= \sum_{i=1}^{N} \Pr\left\{ (\tilde{\bx}^{\prime}_j, \tilde{\bk}^{\prime}_j) = (\bx_i, \bk_i) \right\} \cdot \varphi ( \bx_i, \bk_i) \\
&= \sum_{i=1}^{N} \frac{w_i-k_i}{\sum_{i=1}^N (w_i - k_i)} \cdot \varphi ( \bx_i, \bk_i) = \sum_{i=1}^N \frac{r_i}{N_r} \cdot \varphi(\bx_i, \bk_i).
\end{split}
\end{equation*} 

According to the mutually independence of $(\tilde{\bx}_j^{\prime}, \tilde{\bk}_j^{\prime})$ and the Cauchy-Schwarz inequality, it yields
\begin{equation*}
\begin{split}
\mathbb{E} \left| \tilde{S}_N - S_N \right|^2  & = \mathbb{E} \left| \tilde{S}_N - \mathbb{E}\tilde{S}_N \right|^2  = \frac{1}{N^2}\sum_{j = 1}^{N_r} \mathbb{E} \left| \varphi(\bx_j^{\prime}, \bk_j^{\prime}) - \frac{1}{N_r}\sum_{i=1}^N r_i \cdot \varphi(\bx_i, \bk_i) \right|^2 \\
& \le \frac{2}{N^2} \left[ \sum_{j=1}^{N_r}  \mathbb{E} \left| \varphi  (\bx_j^{\prime}, \bk_j^{\prime}) \right|^2 +  \left(\sum_{i=1}^N \frac{r_i}{N_r}\right)^2 \cdot \left(\sum_{i=1}^N \left|\varphi(\bx_i, \bk_i) \right| \right)^2 \right] \\
& \le \frac{2( N_r + N) \Vert \varphi \Vert^2}{N^2} = (\frac{2N_r}{N^2}+\frac{2}{N}) \Vert \varphi \Vert^2 \le \frac{4 \Vert \varphi \Vert^2}{N}. 
\end{split}
\end{equation*}

\end{proof}

\begin{proof}[Proof of Theorem \ref{theorem_double_bootstrap_bound}]

By the triangle inequality, we have that
\begin{flalign*}
\begin{split}
&\mathbb{E} \left| \Big \langle \varphi, \frac{1}{N_\alpha}\sum_{i=1}^{N} w_i \delta_{(\bx_i, \bk_i)} \Big\rangle  - \Big \langle \varphi, \frac{\lambda^+}{N^+}\sum_{i=1}^{N^+} \delta_{\tilde{(\bx}_i^+, \tilde{\bk}_i^+)}+\frac{\lambda^-}{N^-}\sum_{i=1}^{N^-} \delta_{(\tilde{\bx}^-_i, \tilde{\bk}^-_i)} \Big \rangle \right |^2 \\
&\le  2\left| \frac{\sum_{i=1}^{N^+} w^+_i}{N_\alpha} \right|^2 \cdot\mathbb{E} \left| \Big \langle \varphi,  \frac{1}{N^+}\sum_{i=1}^{N^+} (N^+\tilde{w}_i^+) \cdot \delta_{(\bx_i^+, \bk^+_i)} \Big \rangle - \Big \langle \varphi, \frac{1}{N^+}  \sum_{i=1}^{N^+} \delta_{(\tilde{\bx}_i^+, \tilde{\bk}_i^+)} \Big \rangle \right|^2 \\
&+ 2\left| \frac{\sum_{i=1}^{N^-} w^-_i}{N_\alpha} \right|^2 \cdot\mathbb{E} \left| \Big \langle \varphi,  \frac{1}{N^-}\sum_{i=1}^{N^-} (N^-\tilde{w}_i^-) \cdot  \delta_{(\bx_i^-, \bk^-_i)} \Big \rangle - \Big \langle \varphi, \frac{1}{N^-}  \sum_{i=1}^{N^-} \delta_{(\tilde{\bx}_i^-, \tilde{\bk}_i^-)} \Big \rangle \right|^2 \\
&\le  |\lambda^+|^2 \cdot \frac{8 \Vert \varphi \Vert^2}{N^+} +  |\lambda^-|^2 \cdot \frac{8 \Vert \varphi \Vert^2}{N^-},
\end{split}
\end{flalign*}
where the second inequality utilizes Lemma \ref{theorem_bootstrap_bound}.

\end{proof}


\begin{proof}[Proof of the Theorem \ref{error_uniform_histogram}]

Consider the operator $A_\nu$ 
\begin{equation}\label{deviation_average}
A_\nu(\varphi)(\bx_i, \bk_i) = \varphi(\bx_i, \bk_i) - \frac{1}{\mu(\mathsf{D}_\nu)}\iint_{\mathsf{D}_\nu} \varphi(\bx, \bk) \D \bx \D \bk \tag{A.1}.
\end{equation}
Substituting Eq.~\eqref{Holder_condition} into Eq.~\eqref{deviation_average} yields
\begin{equation*}
\begin{split}
|A_\nu(\varphi)(\bx_i, \bk_i) | \le ~ & \frac{1}{\mu(\mathsf{D}_\nu)} \iint_{\mathsf{D}_\nu}  ~ \big | \varphi(\bx, \bk) - \varphi(\bx_i, \bk_i) \big | ~\D \bx \D \bk \\
\le ~&\frac{1}{\mu(\mathsf{D}_\nu)} \iint_{\mathsf{D}_\nu}  ~ \varepsilon^{\alpha} \Vert \varphi \Vert_{C^{0, \alpha}} ~\D \bx \D \bk = \varepsilon^{\alpha} \Vert \varphi \Vert_{C^{0, \alpha}}. 
\end{split}
\end{equation*}
Then we arrive at that
\begin{equation*}
\begin{split}
 \left | \Big \langle \varphi, \frac{1}{N_\alpha}\sum_{i=1}^{N} w_i \delta_{(\bx_i, \bk_i)} \Big \rangle - \langle \varphi, p_t \rangle \right | &= \left|  \frac{1}{N_\alpha} \sum_{i=1}^{N} \sum_{\nu=1}^{N_h} w_i A_\nu(\bx_i, \bk_i) \delta_{i \nu} \right| \\
& \le  \frac{1}{N_\alpha} \sum_{i=1}^{N}  |w_i|  \cdot \varepsilon^{\alpha}  \Vert \varphi \Vert_{C^{0, \alpha}} = (|\lambda^+| + |\lambda^-|) \varepsilon^{\alpha} \Vert \varphi \Vert_{C^{0, \alpha}},
\end{split}
\end{equation*}
where $\delta_{i\nu}$ is given by
\begin{equation*}
\delta_{i\nu}
=\left\{
\begin{split}
&1, \quad &(\bx_i, \bk_i) \in \mathsf{D}_\nu,\\
&0, \quad &(\bx_i, \bk_i) \notin \mathsf{D}_\nu.
\end{split}
\right.
\end{equation*}

For the second inequality, we start from the triangular inequality,
\begin{equation}\label{A.2}
\begin{split}
\mathbb{E}&\left| \Big \langle \varphi, \frac{1}{N_\alpha}\sum_{i=1}^{N} w_i \cdot \delta_{(\bx_i, \bk_i)} \Big \rangle - \langle \varphi, \tilde{p}_t \rangle \right|^2 \\
\le & 2\mathbb{E} \left| \Big \langle \varphi, \frac{1}{N_\alpha}\sum_{i=1}^{N} w_i \cdot \delta_{(\bx_i, \bk_i)} \Big \rangle - \Big \langle \varphi,  \frac{\lambda^+}{N^+} \sum_{i=1}^{N^+} \delta_{(\tilde{\bx}_i^+, \tilde{\bk}^+_i)} + \frac{\lambda^-}{N^-} \sum_{i=1}^{N^-} \delta_{(\tilde{\bx}_i^-, \tilde{\bk}^-_i)} \Big\rangle \right|^2\\
& + 2 \mathbb{E} \left| \Big \langle \varphi,  \frac{\lambda^+}{N^+} \sum_{i=1}^{N^+} \delta_{(\tilde{\bx}_i^+, \tilde{\bk}^+_i)} + \frac{\lambda^-}{N^-} \sum_{i=1}^{N^-} \delta_{(\tilde{\bx}_i^-, \tilde{\bk}^-_i)} \Big \rangle - \langle \varphi, \tilde{p}_t \rangle \right|^2.
\end{split}
\tag{A.2}
\end{equation}
Consequently, the first term of Eq.~\eqref{A.2} is bounded by Eq.~\eqref{double_bootstrap_bound}, and the second term is bounded by Eq.~\eqref{histogram_error}.

\end{proof}


\end{document}